\documentclass[11pt]{article}
\usepackage{paper}

\usepackage[dvipsnames]{xcolor}
\usepackage{hyperref}

\makeatletter
\newcommand{\TODO}[1]{\textcolor{red}{[TODO\@ifnotempty{#1}{: #1}]}}

\newcommand\num{\addtocounter{equation}{1}\tag{\theequation}}

\newtheorem{proposition}[theorem]{Proposition}

\title{\textbf{Near-optimal fitting of ellipsoids to random points
}}

\author[1]{Aaron Potechin\footnote{Email: \textit{potechin@uchicago.edu}. Supported in part by NSF grant CCF-2008920.}}
\affil[1]{University of Chicago}
\author[2]{Paxton Turner\footnote{Email: \textit{paxtonturner@g.harvard.edu}.}}
\author[2]{Prayaag Venkat\footnote{Email: \textit{pvenkat@g.harvard.edu}. Part of this work was done while visiting the Simons Institute for the Theory of Computing. Part of this work was done at Harvard, supported by an NSF Graduate Fellowship under grant DGE1745303 and Boaz Barak's Simons Investigator Fellowship, NSF grant DMS-2134157, DARPA grant W911NF2010021, and DOE grant DE-SC0022199, support from Oracle Labs and past support by the NSF, as well as the Packard and Sloan foundations and the BSF.}}
\affil[2]{Harvard}
\author[3]{Alexander S.\ Wein\footnote{Email: \textit{aswein@ucdavis.edu}. Part of this work was done while visiting the Simons Institute for the Theory of Computing, supported by a Simons-Berkeley Research Fellowship. Part of this work was done at Georgia Tech, supported by NSF grants CCF-2007443 and CCF-2106444.}}
\affil[3]{UC Davis}

\date{}

\begin{document}
\maketitle

\begin{abstract}
Given independent standard Gaussian points $v_1, \ldots, v_n$ in dimension $d$, for what values of $(n, d)$ does there exist with high probability an origin-symmetric ellipsoid that simultaneously passes through all of the points? This basic problem of fitting an ellipsoid to random points has connections to low-rank matrix decompositions, independent component analysis, and principal component analysis. Based on strong numerical evidence, Saunderson, Parrilo, and Willsky \cite{saunderson2011subspace, saunderson_parrilo_willsky13} conjectured that the ellipsoid fitting problem transitions from feasible to infeasible as the number of points $n$ increases, with a sharp threshold at $n \sim d^2/4$. We resolve this conjecture up to logarithmic factors by constructing a fitting ellipsoid for some $n = d^2/\mathrm{polylog}(d)$. Our proof demonstrates feasibility of the least squares construction of \cite{saunderson2011subspace, saunderson_parrilo_willsky13} using a convenient decomposition of a certain non-standard random matrix and a careful analysis of its Neumann expansion via the theory of graph matrices.

\end{abstract}

\newpage

\tableofcontents 

\section{Introduction}
Let $v_1, \ldots, v_n \in \R^d$ be a collection of points. We say that this collection has the \emph{ellipsoid fitting property} if there exists a symmetric matrix $X \in \R^{d \times d}$ such that $X \succeq 0$ and $v_i^T X v_i = 1$ for all $i \in [n]$. That is, the eigenvectors and eigenvalues of the matrix $X$ describe the directions and reciprocals of the squared-lengths of the principal axes of an origin-symmetric ellipsoid that passes through all of $v_1, \ldots, v_n$. From the definition, it is clear that testing whether the ellipsoid fitting property holds for a given set of points reduces to solving a certain semidefinite program. It is known that if $v_1, \ldots, v_n$ satisfy the ellipsoid fitting property, then $\pm v_1, \ldots, \pm v_n$ lie on the boundary of their convex hull\footnote{A point $v_i$ lies on the boundary of the convex hull of $\pm v_1, \ldots, \pm v_n$ if there exists $x \in \R^d$ such that $\ip{x}{v_i} = 1$ and $|\ip{x}{v_j}| \leq 1$ for all $j \neq i$.} and that the converse holds when $n \le d+1$ (Corollary 3.6 of \cite{saunderson2012diagonal}).

In this paper, we study the ellipsoid fitting property for \emph{random} points. Specifically, let $v_1, \ldots, v_n \sim \calN(0, I_d)$ be i.i.d.\ standard Gaussian vectors in $\R^d$. Treating $n = n(d)$ as a function of $d$, we ask: what is the largest value of $n$ for which $n$ standard Gaussian vectors have the ellipsoid fitting property with high probability\footnote{Here and throughout, \emph{high probability} means probability tending to $1$ as $d \rightarrow \infty$.} as $d \to \infty$? Since the probability of the ellipsoid fitting property is non-increasing as a function of $n$, it is natural to ask if it exhibits a sharp phase transition from 1 to 0 asymptotically as $n$ increases. 

If $n \le d + 1$, then with probability 1, the points $\pm v_1, \ldots, \pm v_n$ have the aforementioned convex hull
property and hence satisfy the ellipsoid fitting property. However, it turns out that for random points, the ellipsoid fitting property actually holds for much larger values of $n$. Intriguing experimental results due to Saunderson et al.~\cite{saunderson2011subspace,saunderson2012diagonal,saunderson_parrilo_willsky13} suggest that the ellipsoid fitting property undergoes a sharp phase transition at the threshold $n \sim d^2/4$. Formally, we restate their conjecture:
\begin{conjecture}
\label{conj:ellipsoid-fitting}
Let $\epsilon > 0$ be a constant and $v_1, \ldots, v_n \sim \calN(0, I_d)$ be i.i.d.\ standard Gaussian vectors in $\R^d$. 
\begin{enumerate}
    \item If $n \le (1-\epsilon) \frac{d^2}{4}$, then $v_1, \ldots, v_n$ have the ellipsoid fitting property with probability $1-o(1)$.
    \item If $n \ge (1+\epsilon) \frac{d^2}{4}$, then $v_1, \ldots, v_n$ have the ellipsoid fitting property with probability $o(1)$. 
\end{enumerate}
\end{conjecture}
By genericity of the random linear constraints and the fact that any $d \times d$ PSD matrix (in fact, symmetric matrix) is described by $d(d+1)/2$ parameters, it can be verified that the system of random linear constraints alone (without the PSD constraint) becomes infeasible with probability 1 if and only if $n > d(d+1)/2$ (see Lemma~\ref{lemma:invertible}). Fascinatingly, Conjecture~\ref{conj:ellipsoid-fitting} posits the existence of a range of values $n \in \left(\frac{d^2}{4}, \frac{d(d+1)}{2} \right)$ for which with high probability, there exists a \emph{symmetric} matrix satisfying the linear constraints, but no such \emph{positive semidefinite} matrix exists. Saunderson et al.~\cite{saunderson2011subspace,saunderson_parrilo_willsky13} made partial progress towards resolving the positive part of this conjecture: they showed that for any $\epsilon > 0$, when $n < d^{\,6/5- \epsilon}$, the ellipsoid fitting property holds with high probability. A special case of Theorem 1.4 of Ghosh, Jeronimo, Jones, Potechin, and Rajendran~\cite{ghosh2020sum}, developed in the context of certifying upper bounds on the Sherrington--Kirkpatrick Hamiltonian, guarantees that for any $\epsilon > 0$, when $n < d^{3/2 - \epsilon}$, there exists with high probability a fitting ellipsoid $X$ whose diagonal entries are all equal to $1/d$.

The ellipsoid fitting problem, a basic question in high-dimensional probability and convex geometry, is further motivated by connections to other problems in machine learning and theoretical computer science. First, Conjecture~\ref{conj:ellipsoid-fitting} was first formulated by Saunderson et al.~\cite{saunderson2011subspace,saunderson2012diagonal,saunderson_parrilo_willsky13} in the context of decomposing an observed $n \times n$ data matrix as the sum of a diagonal matrix and a random rank-$r$ matrix.  They proposed a convex-programming heuristic, called ``Minimum-Trace Factor Analysis (MTFA)'' for solving this problem and showed it succeeds with high probability if the ellipsoid fitting property for $n$ standard Gaussian vectors in $d = n-r$ dimensions holds with high probability.

Second, Podosinnikova et al.~\cite{podosinnikova2019overcomplete} identified a close connection between the ellipsoid fitting problem and the overcomplete independent component analysis (ICA) problem, in which the goal is to recover a mixing component of the model when the number of latent sources $n$ exceeds the dimension $d$ of the observations. They show that the ability of an SDP-based algorithm to recover a mixing component is related to the feasibility of a variant of the ellipsoid fitting problem in which the norms of the random points fluctuate with higher variance than in our model. They give experimental evidence that the SDP succeeds when $n < d^2/4$, the same phase transition behavior described in Conjecture~\ref{conj:ellipsoid-fitting}, and show rigorously that it succeeds for some $n = \Omega(d \log d)$.

Third, the ellipsoid fitting property for random points is directly related to the ability of a canonical SDP relaxation to certify lower bounds on the discrepancy of nearly-square random matrices. The discrepancy of random matrices is a topic of recent interest, with connections to controlled experiments~\cite{turner2020balancing}, the Ising Perceptron model from statistical physics~\cite{AubPerZde19}, and the negatively-spiked Wishart model~\cite{bandeira2019computational,venkat2022efficient}. A result implicit in the work of Saunderson, Chandrasekaran, Parrilo, and Willsky~\cite{saunderson2012diagonal} states that if the ellipsoid fitting property for $n$ Gaussian points in dimension $d$ holds 
with high probability, then the SDP fails to certify a non-trivial lower bound on the discrepancy of a $(n-d) \times n$ matrix with i.i.d.\ standard Gaussian entries (see Appendix~\ref{sec:disc} for further discussion). In addition, this provides further evidence of the algorithmic phase transition for the detection problem in the negatively-spiked Wishart model that was previously predicted by the low-degree likelihood ratio method~\cite{bandeira2019computational}. 

Finally, a current active area of research in theoretical computer science aims to give rigorous evidence for information-computation gaps in average-case problems by characterizing the performance of powerful classes of algorithms, such as the Sum-of-Squares (SoS) SDP hierarchy. Often, the most challenging technical results in this area involve proving lower bounds against these SDP-based algorithms. Moreover, there are relatively few examples for which predicted phase transition behavior has been sharply characterized (see e.g.~\cite{barak2019nearly,ghosh2020sum,hopkins2017power,hsieh2022algorithmic,jones2022sum,kothari2021stress,mohanty2020lifting,schoenebeck2008linear}), all proven using the same technique of ``pseudo-calibration''. We remark that proving the positive side of Conjecture~\ref{conj:ellipsoid-fitting} amounts to proving the feasibility of an SDP with random linear constraints. This also arises
in average-case SoS lower bounds, although the linear constraints for average-case SoS lower bounds are generally very intricate.

The main contribution of our work is to resolve the positive side of Conjecture \ref{conj:ellipsoid-fitting} up to logarithmic factors. (Recall that the negative side of Conjecture~\ref{conj:ellipsoid-fitting} has already been resolved up to a factor of 2.)

\begin{theorem}
\label{thm:main}
There is a universal constant $C > 0$ so that if $n \le  d^2/ \log^C(d)$, then $v_1, \ldots, v_n \sim \calN(0, I_d)$ have the ellipsoid fitting property with high probability.
\end{theorem}

As a first corollary of Theorem~\ref{thm:main}, we conclude that MTFA in this setting succeeds provided $r \leq n - \sqrt{n} \polylog(n)$, improving on the bound $r \leq n - \omega(n^{2/3})$ from a combination of the results of Saunderson et al.~\cite{saunderson2011subspace,saunderson_parrilo_willsky13} and Ghosh et al.~\cite{ghosh2020sum}. Second, Theorem~\ref{thm:main} implies the following ``finite-size'' phase transition result: a canonical SDP cannot distinguish between an $m \times n$ matrix with i.i.d.\ standard Gaussian entries and one with a planted Boolean vector in its in kernel when $m \leq n - \sqrt{n} \polylog(n)$ (see Appendix~\ref{sec:disc}), again improving on the bound $m \leq n - \omega(n^{2/3})$ that follows from \cite{ghosh2020sum}.

\paragraph{Experimental results}

It is natural to wonder whether our proof of Theorem~\ref{thm:main} can be sharpened to make further progress on Conjecture~\ref{conj:ellipsoid-fitting}. Our proof is based on a least-squares construction
that was first studied in \cite{saunderson2011subspace,saunderson_parrilo_willsky13} (see Section \ref{sec:technical-overview}). Although the least-squares construction always satisfies the linear constraints, in Section~\ref{sec:future-work} we corroborate experimental evidence of Saunderson et al.\ suggesting that it fails to be positive semidefinite strictly below the conjectured $n \sim d^2/4$ threshold. We also introduce a new method called the ``identity-perturbation" construction that also always satisfies the linear constraints and appears to improve on the least-squares construction in experiments, while having similar time complexity. Our simulations in Section \ref{sec:future-work} provide numerical evidence that the positive semi-definiteness of the least-squares and identity-perturbation constructions undergo sharp phase transitions at roughly $n \approx d^2/17$ and $n \approx d^2/10$, respectively. 
We did not run eperiments on the pseudo-calibration construction of~\cite{ghosh2020sum} because this construction has the drawback that it involves logarithmic degree polnomials of the input and is thus very hard to compute.

These results suggest that a full resolution of Conjecture~\ref{conj:ellipsoid-fitting} requires either sharply analyzing the pseudocalibration construction of~\cite{ghosh2020sum} (if it achieves the threshold $n \approx \frac{d^2}{4}$, which is unknown), inventing a new construction and analyzing it, or reasoning indirectly about the ellipsoid fitting property without considering any explicit candidate.

\paragraph{Related work}
We now discuss two closely related works that study a simpler variant of the ellipsoid fitting problem. In this variant, the constraints $v_i^T X v_i = 1$ in the definition of the ellipsoid fitting property are replaced by $\ip{X}{G_i} = 1$, where $G_1, \ldots, G_n \in \R^{d \times d}$ have i.i.d.\ standard Gaussian entries. Amelunxen, Lotz, McCoy, and Tropp~\cite{amelunxen2014living} give a general framework for characterizing phase transition behavior of convex programs with random constraints. Interestingly, their framework shows that the conclusion of Conjecture~\ref{conj:ellipsoid-fitting} is true for the simpler variant. Moreover, they explain that the occurrence of the phase transition at $n \sim d^2/4$ arises from the fact that $d(d+1)/4$ is the ``statistical dimension'' of the cone of $d \times d$ PSD matrices. The known proofs of these results are either based on conic geometry or Gaussian process techniques that crucially rely on the fact that the entries of the constraint matrices are i.i.d.\ and Gaussian. Despite the strikingly similar phase transition behavior for the two models of random constraints, it appears unlikely that these techniques can be used to resolve Conjecture~\ref{conj:ellipsoid-fitting}. In this simpler i.i.d.\ setting of Amelunxen et al., Hsieh and Kothari~\cite{hsieh2022algorithmic} show that when $n \leq d^2/ \polylog(d)$, the ellipsoid fitting SDP (which corresponds to the degree-2 SoS SDP) equipped with some additional symmetry constraints (corresponding to the degree-4 SoS SDP) is still feasible with high probability. 

Ghosh et al.~\cite{ghosh2020sum} consider the original setting in which the constraint matrices are of the form $v_i v_i^T$ and also impose the constraint that the diagonal entries of $X$ satisfy $X_{ii}=1/d$ for all $i \in [d]$. They show that for any $\epsilon > 0$ this SDP, even when augmented with more constraints corresponding to higher degree SoS, remains feasible with high probability for some $n = \Omega(d^{3/2 - \epsilon})$ and conjecture that this should even hold for some $n = \Omega(d^{2-\epsilon})$. The proofs of the results of Ghosh et al.~\cite{ghosh2020sum} and Hsieh and Kothari~\cite{hsieh2022algorithmic} are based on the pseudo-calibration technique. Due to technical complications that arise when analyzing higher degree SoS pseudocalibration constructions, \cite{ghosh2020sum} can only prove feasibility for some $n = \Omega(d^{3/2 - \epsilon})$. However, as we detail in Appendix~\ref{sec:pseudocalibration}, these technical complications do not arise when specialized to the degree-2 case, which gives an alternative proof of Theorem~\ref{thm:main}. 

On a technical level, our proof heavily relies on the recently introduced machinery of \textit{graph matrices} \cite{ahn2016graph}, a powerful tool for obtaining norm bounds of structured random matrices using a certain graphical calculus (see Section \ref{sec:graph_matrices}). Ours is among the first works to apply graph matrices outside of the Sum-of-Squares lower bound literature, and we expect graph matrices to be useful for other probabilistic applications beyond average-case complexity theory. 

Shortly after a revised version of this paper was posted on the arXiv, independent work of Kane and Diakonikolas~\cite{kane2022nearly} analyzed the identity-perturbation construction and showed it improves on the logarithmic factor in Theorem~\ref{thm:main}. Their short proof crucially uses the fact that the norms and directions of a standard Gaussian are independent. Our proof is more technically involved but can be adapted to handle non-Gaussian distributions whose coordinates are independent and sufficiently well-concentrated. Our work analyzes the least-squares construction, and its analysis also carries over easily to the analysis of the identity perturbation construction, as we demonstrate in Section \ref{appendix:identity_perturbation}.


\section{Technical overview}
\label{sec:technical-overview}
We now give an overview of the proof of Theorem~\ref{thm:main}. To begin, we introduce some convenient notation. Define the linear operator $\calA : \R^{d \times d} \rightarrow \R^n$ by $\calA(X) := (v_1^T X v_1, \ldots, v_n^T X v_n)^T$ and let $\calA^\dagger$ be its pseudoinverse. The fitting ellipsoid in Theorem~\ref{thm:main} is obtained  via the \textit{least-squares construction}:
\begin{align}
\label{eqn:least_squares}
    X_{\text{LS}} = \calA^{\dagger} (1_n),
\end{align} 
which is the minimum Frobenius norm solution to the linear constraints. This construction was first studied by Saunderson et al. \cite{saunderson2011subspace,saunderson_parrilo_willsky13}. Our analysis builds on their work and also introduces additional probabilistic and linear-algebraic ideas, such as the application of graph matrices, leading to nearly-sharp bounds for the ellipsoid fitting problem. 

To prove Theorem~\ref{thm:main}, it suffices to verify that $\calA(X_{\text{LS}}) = 1_n$ and $X_{\text{LS}} \succeq 0$ with high probability, for appropriate values of $n$ and $d$. The first condition can be easily verified: with probability 1, the $n \times n$ matrix $\calA \calA^*$ is invertible (Lemma~\ref{lemma:invertible}), so we may write $\calA^\dagger = \calA^* (\calA \calA^*)^{-1}$ and compute that indeed $\calA(X_{\text{LS}}) = 1_n$, where the adjoint $\calA^* : \R^{n} \rightarrow \R^{d \times d}$ satisfies $\calA^* (c) = \sum_{i=1}^n c_i v_i v_i^T$. 







The challenging part of the proof is to verify that $X_{\text{LS}} \succeq 0$ with high probability. We now give some intuition for why this condition holds. First, observe that if we take $X_0 = \frac{1}{d} I_d$, then a simple application of a tail bound for the $\chi^2$ distribution and a union bound over the $n$ constraints yields $\norm{\calA (X_0) - 1_n}_\infty = O(\sqrt{\log (n)/d})$ with high probability. In words, $X_0$ defines an ellipsoid that approximately fits the points $v_1, \ldots, v_n$ and whose eigenvalues are well-separated from 0. 

Second, there is a sense in which $X_{\text{LS}}$ is (approximately) a projection of $X_0$ onto the affine subspace $\{X \in \R^{d \times d}: \calA(X) = 1_n\}$. Recall that $X_{\text{LS}}$ can be expressed as the solution of the following optimization problem:
\[
\min_{X \in \R^{d \times d},\, \calA(X) = 1_n} \norm{X}_F^2.
\]
In fact, since the above minimization is over $X$ that satisfy $\calA(X) = 1_n$, $X_{\text{LS}}$ is also the solution of
\[
\min_{X \in \R^{d \times d},\, \calA(X) = 1_n} \norm{X - \frac{1}{dn} \sum_{i=1}^n v_i v_i^T}_F^2.
\]
For $n \gg d$, it is the case that $\frac{1}{dn} \sum_{i=1}^n v_i v_i^T \approx X_0$ with high probability. Thus, we interpret $X_{\text{LS}}$ as an (approximate) projection of $X_0$ onto the affine subspace $\{X \in \R^{d \times d}: \calA(X) = 1_n\}$.

We now provide an outline of the proof that $X_{\text{LS}} \succeq 0$ and describe some of its challenges. A basic approach is to center around the deterministic matrix $M = (d^2 + d)I_n + d 1_n 1_n^T$. A straightforward rearrangement yields
\begin{equation*}
X_{\text{LS}} = \calA^* (\calA \calA^*)^{-1} 1_n = \calA^* (I_n - M^{-1} \tilde \Delta)^{-1} M^{-1}1_n,
\end{equation*}
where $\tilde \Delta = M - \calA \calA^*$. To invert the matrix $I_n - M^{-1} \tilde \Delta$, we may expand it as a Neumann series
\begin{equation}
\label{eqn:Neumann}
(I_n - M^{-1} \tilde \Delta)^{-1} = I_n +  M^{-1} \tilde \Delta + \sum_{i=2}^\infty (M^{-1}\tilde \Delta)^i,
\end{equation}
that converges if $\norm{M^{-1} \tilde \Delta}_{op} < 1$. Observe that if the vector
\[
u  = (\calA \calA^*)^{-1} 1_n= (I_n - M^{-1} \tilde \Delta)^{-1} M^{-1}1_n
\]
has non-negative coordinates with high probability, then we may immediately conclude that $X_{\text{LS}} \succeq 0$ since $\calA^* (u) = \sum_i u_i v_i v_i^T$ is automatically PSD. 

While it is possible to show that when $n \ll d^{3/2}$, the vector $u$ indeed has positive coordinates, this argument suffers from a significant problem. It turns out that there is a phase transition at $n \asymp d^{3/2}$, beyond which the vector $u$ switches from having non-negative coordinates to having both positive and negative ones. One reason for this is that the approximation $M^{-1} \tilde\Delta \approx I_n$ breaks down at the same $n \asymp d^{3/2}$ barrier. A previous version of this paper contained an error related to this non-negativity phenomenon. In Section~\ref{sec:prev-approach}, we describe how this error can be fixed. However, we now present a cleaner approach that avoids this issue altogether.  

To handle the positive and negative coordinates of $u$ requires a different approach that more precisely takes into account its correlations with $\calA^*$. We achieve this by first removing a rank-two component from $\calA \calA^*$ that prevents it from being close to the identity when $n \gg d^{3/2}$. Define
\begin{align*}
	B = \calA \calA^* - (w 1_n^T + 1_n w^T + d 1_n 1_n^T)
\end{align*}
where $w \in \R^n$ is defined by $w_i = \| v_i \|^2 - d$.  As we show in Lemma \ref{lemma:B-concentrates}, $B$ is close to $(d^2 + d)I_n$ for all $n \leq d^2/\log^C(d)$. For this reason, $B$ is well-behaved and amenable to Neumann expansion arguments. 

Next, since $\calA \calA^*$ is the sum of $B$ and a low rank matrix, we obtain a convenient expression for $(\calA \calA^*)^{-1}$ using the Woodbury matrix formula \cite{woodbury1950inverting}, which results in the following useful decomposition of the vector $u$ (see Lemma \ref{lemma:woodbury}): 
\[
u = (\calA \calA^*)^{-1} 1_n
= \rho \cdot \big(  \lambda_1 \,  B^{-1} 1_n  + \lambda_2 \, B^{-1} w \big) ,
\]
where $\rho, \lambda_1, \lambda_2$ are certain scalar random variables. We show that $\rho > 0$, $\lambda_1 = 1 + 1_n^T B^{-1} w \sim 1$, and $\lambda_2 = -1_n^T B^{-1} 1_n = o(n/d^2) \ll \lambda_1$ with high probability. The proof then reduces to showing that
\begin{align}
\label{eqn:pos_part}
    \calA^* (B^{-1} 1_n)  &\succeq (1 - o(1)) \frac{n}{d^2} I_d \\
    \| \calA^* (B^{-1} w) \|_{op} &= o(1). 
    \label{eqn:mixed_part}
\end{align}
Intuitively, \eqref{eqn:pos_part} has non-negative coordinates since $B$ is close to a multiple of the identity for the entire range $n \leq d^2/\log^C(d)$ and we have $B^{-1} 1_n \approx 1_n$. With the same intuition, we expect that $B^{-1} w$ behaves like a multiple of $w$, which has i.i.d.\ centered coordinates. If $w$ were independent of $\calA^*$, significant cancellation would happen among the rank one vectors $\{ v_i v_i^T\}_{i = 1}^n$ , yielding the bound \eqref{eqn:mixed_part} (by matrix Bernstein  or its variants, see e.g.~\cite{tropp2012user}). 

However, making this argument precise to take into account interactions between $\calA^*$ and $B^{-1}w$ is a considerable technical challenge. To handle this, we expand $B^{-1}$ as a Neumann series, similarly to \eqref{eqn:Neumann}. We then analyze terms of this series using the framework of graph matrices \cite{ahn2016graph}. Graph matrices provide a powerful tool for controlling the operator norm of certain matrices whose entries can be expressed as low-degree polynomials in i.i.d.\ random variables. Graph matrices serve to transform the analytic problem of controlling the operator norm of a random matrix $X$ into a more tractable combinatorial one that involves studying certain weights of graphs associated to $X$. This part of the argument forms the bulk of our analysis and is detailed in Section \ref{sec:graph_matrices}.

\section{Future work}
\label{sec:future-work}
\paragraph{Towards the positive side of Conjecture~\ref{conj:ellipsoid-fitting}}
Towards understanding whether an explicit construction can be used resolve the positive side of Conjecture~\ref{conj:ellipsoid-fitting}, we now discuss the following ``identity perturbation'' construction that is inspired by previous work~\cite{saunderson_parrilo_willsky13}:
\[
X_{\text{IP}} = \frac{1}{d}I_d + \calA^*(\alpha) = \frac{1}{d}I_d + \sum_{i=1}^n \alpha_i v_i v_i^T,
\]
where $\alpha \in \R^n$ is defined to be the unique solution of $\calA (\calA^*(\alpha)) = 1_n - \calA(\frac{1}{d}I_d)$.\footnote{A similar construction is suggested in~\cite{saunderson_parrilo_willsky13}, although no specific initialization is given.} By definition of $\alpha$, it always holds that $\calA (X_{\text{IP}}) = 1_n$. In words, $X_{\text{IP}}$ is obtained from the approximately fitting ellipsoid $\frac{1}{d}I_d$ by adding multiples of the constraint matrices $\{v_i v_i^T\}_{i=1}^n$ so that it exactly satisfies the linear constraints. 

Our experimental results are depicted in Figure~\ref{fig:plots}. In summary, there appear to be two constants $c_{\text{LS}} \approx 1/17$ and $c_{\text{IP}} \approx 1/10$ such that the probabilities of PSD-ness of $X_{\text{LS}}$ and $X_{\text{IP}}$ undergo phase transitions from 1 to 0 asymptotically at $n = c_{\text{LS}} d^2$ and $n = c_{\text{IP}} d^2$, respectively. We emphasize that $c_{\text{LS}} < c_{\text{IP}} < 1/4$, meaning that there actually appear to be \emph{three} distinct phase transitions related to the ellipsoid fitting problem. These results suggest that it is unlikely that the positive side of Conjecture~\ref{conj:ellipsoid-fitting} can be resolved by a sharper analysis of either of these two natural constructions.

In this work, we show that both the least-squares  and identity perturbation constructions are positive semidefinite provided that $n \leq d^2/\text{polylog}(d)$. However, we still believe it is an interesting problem to sharply characterize the behavior of $X_{\text{LS}}$ and $X_{\text{IP}}$. Given that $X_{\text{IP}}$ appears to outperform $X_{\text{LS}}$, we now explain how one might approach this problem for $X_{\text{IP}}$. Again the central challenge is to show that $X_{\text{IP}} \succeq 0$ with high probability. Observe that $\alpha = (\calA \calA^*)^{-1}b$ and so $X_{\text{IP}} \succeq 0$ is implied by $\norm{\calA^* ((\calA \calA^*)^{-1}b)}_{op} \leq 1/d$. We immediately recognize that to proceed with the analysis, we must invert $\calA \calA^*$ as in the analysis of $X_{\text{LS}}$. Applying the Neumann series expansion thus encounters the same bottlenecks as in the analysis of $X_{\text{LS}}$. 
We leave the problem of precisely characterizing the eigenvalues and eigenvectors of the random inner-product matrix $\calA \calA^*$ as a direction for future research. 

Additionally, we remark that computing either of $X_{\text{LS}}$ or $X_{\text{IP}}$ amounts to applying the inverse of a certain $n \times n$ matrix to a vector. In contrast, testing whether the ellipsoid fitting property holds for a given set of points involves solving a semidefinite program, which requires a large polynomial runtime. To the best of our knowledge, it is an open question to find a faster algorithm achieving the conjectured threshold $n \sim d^2/4$, even in simulations. 

\paragraph{Towards the negative side of Conjecture~\ref{conj:ellipsoid-fitting}}
As noted earlier, a simple dimension-counting argument (see Lemma~\ref{lemma:invertible}) shows that when $n > d(d+1)/2$, the linear constraints alone are infeasible with probability 1. Any proof of the failure of the ellipsoid fitting property with high probability for $n > cd^2$ for a constant $c \in (1/4, 1/2)$ would likely yield significant insight into Conjecture~\ref{conj:ellipsoid-fitting}.

\paragraph{Applications to other random SDPs}
In Appendix~\ref{sec:disc}, we prove a negative result showing that a certain SDP (which corresponds to the degree-2 SoS SDP relaxation) cannot certify a non-trivial lower bound on the discrepancy of random Gaussian matrices with $m$ rows and $n$ columns when $m < n - \sqrt{n} \polylog(n)$. As we have mentioned, for simpler variants of the ellipsoid fitting problem, there are results of this type for SDPs corresponding to higher-degree SoS relaxations (e.g.~\cite{ghosh2020sum,hsieh2022algorithmic}). Is it true that higher-degree SoS SDPs also fail to certify non-trivial discrepancy lower bounds in the regime described above?

More generally, can one apply either the least-squares or identity-perturbation constructions to prove average-case SDP lower bounds for other problems? We expect that these constructions are tractable to analyze for SDPs with a PSD constraint and ``simple'' random linear constraints, such as the degree-2 SoS SDP relaxation of the clique number (see e.g.\ Section 2.2 of~\cite{barak2019nearly}) and SoS relaxations of random systems of polynomial equations of the type in~\cite{hsieh2022algorithmic}.

\begin{figure}
    \centering
    \begin{subfigure}[b]{0.3\textwidth}
        \includegraphics[width=\textwidth]{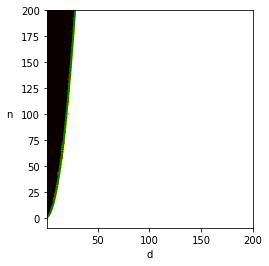}
        \caption{Ellipsoid fitting SDP, $c = 1/4$}
        \label{fig:sdp}
    \end{subfigure}
    ~ 
    \begin{subfigure}[b]{0.3\textwidth}
        \includegraphics[width=\textwidth]{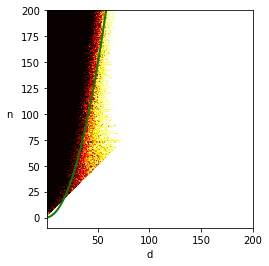}
        \caption{Least-squares, $c = 1/17$}
        \label{fig:ls}
    \end{subfigure}
    ~ 
    \begin{subfigure}[b]{0.3\textwidth}
        \includegraphics[width=\textwidth]{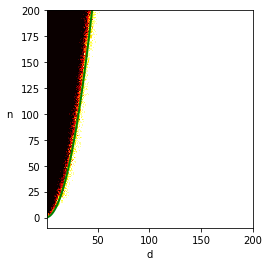}
        \caption{Identity perturbation, $c = 1/10$}
        \label{fig:ip}
    \end{subfigure}
    \caption{For each $(n,d)$ with $1 \leq d \leq n \leq 200$, we generated 10 independent instances of the ellipsoid fitting problem with $n$ points in $\R^d$ and computed the fraction of instances for which each of the three constructions (original SDP, least-squares, and identity perturbation) was a valid fitting ellipsoid. The color of each cell corresponds to the fraction of ``successful'' instances, increasing in the following order: black (zero), red, orange, yellow, white (one). In each plot, the green curve corresponds to a function of the form $n(d) = cd^2$ for some constant $c$. See Appendix~\ref{sec:experiment-details} for further details.}
    \label{fig:plots}
\end{figure}
\section{Proof of Theorem~\ref{thm:main}}

As discussed in Section~\ref{sec:technical-overview}, it suffices to show that $X = X_{\text{LS}} \succeq 0$. We make the simplification that $n = d^2/\text{polylog}(d)$, as recorded in the following remark.
\begin{remark}
\label{rmk:n=d2-polylog}
By monotonicity (with respect to $n$) of the probability of the ellipsoid fitting property holding, it suffices to fix $n = d^2/\log^C (d)$ for some sufficiently large constant $C > 0$ to be determined. In fact, all of our technical lemmas below hold under the more general assumption that $d \leq n \leq d^2/\log^C (d)$.
\end{remark} 

We proceed to showing $X \succeq 0$ by first separating out the low-rank and high-rank terms from $\calA \calA^*$ and then expanding the inverse as a Neumann series. Define the vector $w \in \R^n$ by $w_i = \norm{v_i}_2^2 - d$ for every $i \in [n]$. Next, define the rank 2 matrix $W = w 1_n^T + 1_n w^T + d 1_n 1_n^T \in \R^{n \times n}$, the high-rank matrix $\Gamma = \mathcal{A} \mathcal{A}^* - W - \alpha I_n \in \R^{n \times n}$, where $\alpha = d^2 + d$, and $B = \Gamma + \alpha I_n$.
Then, we have the following decomposition:
\begin{align*}
	\mathcal{A} \mathcal{A}^* 
	&= \big( \mathcal{A} \mathcal{A}^*  -
	(w 1_n^T + 1_n w^T + d 1_n 1_n^T) 
	- \alpha I_n \big)
	+ (w 1_n^T + 1_n w^T + d 1_n 1_n^T) 
	+ \alpha I_n 
	\\& = \Gamma + W + \alpha I_n
	\\&= B+W.
\end{align*}

The following lemma is a consequence of the Woodbury matrix identity \cite{woodbury1950inverting}. We defer the proof to Appendix \ref{appendix:woodbury}. 

\begin{lemma}
\label{lemma:woodbury}
	Let $B = \Gamma + \alpha I_n$.  We have
	\begin{align}
	(\calA \calA^*)^{-1} 1_n = \frac{1}{s^2 - r u} \cdot 
		\bigg( (1 + 1_n^T B^{-1} w)  B^{-1} 1_n  - (1_n^T B^{-1} 1_n) B^{-1} w   \bigg)
	\end{align}
where $r, s, u$ are defined as
\begin{align*}
	 \begin{pmatrix}
		r & s \\
		s & u
	\end{pmatrix} := \begin{pmatrix}
		1_n^T B^{-1} 1_n & 1 + 1_n^T B^{-1} w \\
		1 + 1_n^T B^{-1} w & -d + w^T B^{-1} w
	\end{pmatrix}. 
\end{align*}
\end{lemma}

\noindent By Lemma~\ref{lemma:woodbury}, we have that
\begin{align*}
	X &= \frac{1}{s^2 - ru}
	\calA^*\bigg( (1 + 1_n^T B^{-1} w)  B^{-1} 1_n  - (1_n^T B^{-1} 1_n) B^{-1} w \bigg)  \\
	&= \frac{1}{s^2 - ru} \left( (1 + 1_n^T B^{-1} w)  \, \calA^* (B^{-1} 1_n) 
	-  (1_n^T B^{-1} 1_n) \calA^* (B^{-1} w) \right).
\end{align*}
Clearly, $X \succeq 0$ follows if the next two conditions are satisfied:
\begin{equation}
\label{cond:pos-coeff}
s^2 - ru \geq 0,
\end{equation}
and
\begin{equation}
\label{cond:psd}
 (1 + 1_n^T B^{-1} w)  \, \calA^* (B^{-1} 1_n) -  (1_n^T B^{-1} 1_n) \calA^* (B^{-1} w) \succeq 0.
\end{equation}

We verify that these two conditions are satisfied with high probability by invoking the following lemmas, whose proofs are deferred to the next section.

\begin{lemma}
\label{lemma:1T-Binv-1}
There is some constant $C > 0$ such that if $d \le n \le d^2/\log^C (d)$, then  $1_n^T B^{-1} 1_n = \Theta (n/d^2)$ with high probability.
\end{lemma}

\begin{lemma}
\label{lemma:wT-Binv-w}
There is some constant $C > 0$ such that if $d \le n \le d^2/\log^C (d)$, then $w^T B^{-1} w = \tilde  O (n/d)$ with high probability.
\end{lemma}

\begin{lemma}
\label{lemma:1T-Binv-w}
There is some constant $C > 0$ such that if $d \le n \le d^2/\log^C (d)$, then $|1_n^T B^{-1} w| = o(1)$ with high probability.
\end{lemma}

\begin{lemma}
\label{lemma:Astar-Binv-1}
There is some constant $C > 0$ such that if $d \le n \le d^2/ \log^C (d)$, then $\calA^* (B^{-1} 1_n) \succeq (1-o(1)) \frac{n}{d^2} I_d$ with high probability.
\end{lemma}

\begin{lemma}
\label{lemma:Astar-Binv-w}
There is some constant $C > 0$ such that if $d \le n \le d^2/\log^C (d)$, then $\norm{\calA^* (B^{-1} w)}_{op} = o(1)$ with high probability.
\end{lemma}

For Condition~\eqref{cond:pos-coeff}, if $n \le d^2/ \log^C (d)$ for a sufficiently large constant $C$, we have that with high probability
\begin{align}
\label{eqn:s2-ru}
s^2 - ru \geq -ru = 1^T_n B^{-1} 1_n(d  - w^T B^{-1} w) =   \Theta(n/d^2) (d - \Tilde{O}(n/d)) =  \Theta(n/d) \geq 0,
\end{align} 
for sufficiently large $n,d$, by Lemmas~\ref{lemma:1T-Binv-1} and \ref{lemma:wT-Binv-w}. For Condition~\eqref{cond:psd}, if $n \le d^2/ \log^C (d)$ for a sufficiently large constant $C$, we have that with high probability
\begin{align*}
(1 + 1_n^T B^{-1} w)  \calA^* (B^{-1} 1_n) -  (1_n^T B^{-1} 1_n) \calA^* (B^{-1} w) &\succeq (1- o(1)) \calA^* (B^{-1} 1_n) - \Theta \left( \frac{n}{d^2}\right) \norm{\calA^* (B^{-1} w)}_{op}  I_d \\
&\succeq \left( (1- o(1)) (1- o(1)) \frac{n}{d^2} - \Theta \left( \frac{n}{d^2}\right) \norm{\calA^* (B^{-1} w)}_{op} \right) I_d \\ 
&=\left( (1- o(1)) (1- o(1)) \frac{n}{d^2} - o \left( \frac{n}{d^2}\right) \right) I_d \succeq 0,
\end{align*}
for sufficiently large $n,d$, by Lemmas~\ref{lemma:1T-Binv-w}, \ref{lemma:1T-Binv-1}, \ref{lemma:Astar-Binv-1}, and \ref{lemma:Astar-Binv-w}.

\section{Proofs of remaining technical lemmas}

The proofs of the remaining technical lemmas all make use of the following result, whose proof is postponed to Section \ref{sec:shapes}. 
\begin{lemma}
\label{lemma:B-concentrates}
There is some constant $C > 0$ such that if $d \le n \le d^2/\log^C (d)$, then with high probability, $\norm{B - \alpha I_n}_{op} = \Tilde{O}(d\sqrt{n})$.
\end{lemma}
We now show that Lemmas~\ref{lemma:1T-Binv-1} and \ref{lemma:wT-Binv-w} follow from Lemma~\ref{lemma:B-concentrates}.

\begin{proof}[Proof of Lemma~\ref{lemma:1T-Binv-1}]
By assumption on $n$ and Lemma~\ref{lemma:B-concentrates}, with high probability, it holds that
\[
 0 \preceq (\alpha - \Tilde{O}(d\sqrt{n})) I_n \preceq B \preceq (\alpha + \Tilde{O}(d\sqrt{n})) I_n.
\]
This implies that
\[
(\alpha + \Tilde{O}(d \sqrt{n}))^{-1} I_n \preceq B^{-1} \preceq (\alpha - \Tilde{O}(d \sqrt{n}))^{-1} I_n.
\]
The proof is complete by combining the previous line with the following fact:
\[
\lambda_{min}(B^{-1}) \|1_n\|^2 \le 1_n^T B^{-1} 1_n \le \lambda_{max}(B^{-1}) \|1_n\|^2.
\qedhere \]
\end{proof}

\begin{proof}[Proof of Lemma~\ref{lemma:wT-Binv-w}]
As in the proof of Lemma~\ref{lemma:1T-Binv-1}, by assumption on $n$ and Lemma~\ref{lemma:B-concentrates}, it holds with high probability that:
\begin{align}
\label{eqn:wwT_trace}
    w^T B^{-1} w = \Theta\left(\frac{1}{d^2} \right) \cdot \norm{w}_2^2.
\end{align}
To complete the proof, it suffices to show that $\norm{w}_2^2 = \sum_{i=1}^n (\norm{v_i}_2^2 - d)^2 = \tilde O(nd)$ with high probability. 

Note that for fixed $i$, we have conservatively that
\begin{align} 
\label{eqn:w_coord}
|\| v_i \|_2^2 - d| \leq C(\log n) \sqrt{d} 
\end{align} 
with probability $n^{-C \cdot \Omega(1)}$ by Bernstein's inequality~\cite{vershynin2018high}. Now for a large enough constant $C>0$, using the union bound we have that \eqref{eqn:w_coord} holds for all $1 \leq i \leq n$. Immediately we have $\norm{w}_2^2 = \tilde O(nd)$, proving Lemma \ref{lemma:wT-Binv-w} after combining with \eqref{eqn:wwT_trace}.
\end{proof}

\subsection{Proof of Lemma~\ref{lemma:Astar-Binv-1}}
Define the matrix $\Delta = -\Gamma = \alpha I_n - B \in \R^{n \times n}$. By Lemma~\ref{lemma:B-concentrates}, we have that $\norm{\Delta}_{op} = \Tilde{O}(\max(n, d\sqrt{n}))$ with high probability. By our assumption that $n = O(d^2/\polylog (d))$, we have $\norm{\alpha^{-1} \Delta}_{op} < 1$ (for $d$ large enough). We may then conduct the following (convergent) Neumann series expansion:
\begin{align*}
B^{-1} &= (\alpha I_n - \Delta)^{-1} \\
&= \alpha^{-1} (I_n - \alpha^{-1} \Delta)^{-1} \\
&= \alpha^{-1} \sum_{k=0}^\infty (\alpha^{-1} \Delta)^k.
\end{align*}

Thus, we have that
\[
\lambda_{min}(\calA^* (B^{-1} 1_n)) \geq \alpha^{-1} \lambda_{min} (\calA^{*}(1_n)) - \sum_{k=1}^\infty \alpha^{-(k+1)} \norm{\calA^*(\Delta^k 1_n)}_{op}.
\]
It is a standard fact from random matrix theory (see e.g.\ Theorem 4.7.1 of \cite{vershynin2018high}) that when $n = \omega(d)$, then with high probability:
\[
\lambda_{min}(\calA^{*}(1_n)) = \lambda_{min}\left(\sum_{i=1}^n v_i v_i^T\right) = (1-o(1))n.
\]
To complete the proof, it suffices to show that with high probability:
\[
\sum_{k=1}^\infty \alpha^{-(k+1)} \norm{\calA^*(\Delta^k 1_n)}_{op} = o\left(\frac{n}{d^2}\right).
\]
To this end, introduce a truncation parameter $T \in \N$ and write:
\[
\sum_{k=1}^\infty \alpha^{-(k+1)} \norm{\calA^*(\Delta^k 1_n)}_{op}  = \sum_{k=1}^{T-1} \alpha^{-(k+1)} \norm{\calA^*(\Delta^k 1_n)}_{op}  + \sum_{k=T}^\infty \alpha^{-(k+1)} \norm{\calA^*(\Delta^k 1_n)}_{op}.
\]
Now, take $T=2$ and recall $n \le d^2/\polylog (d)$. The proof is complete by invoking Lemma~\ref{lemma:Astar-Delta-k-ones} below with $k=1$ to control the first summation and Lemma~\ref{lemma:series-tail-Astar-Binv-ones} below with $T = 2$ to control the second summation. 

Although Lemma \ref{lemma:series-tail-Astar-Binv-ones} below is only required with $T = 2$ in order to prove Lemma \ref{lemma:Astar-Binv-1}, its general form with $T \geq 2$ is crucial to the proofs of Lemmas \ref{lemma:1T-Binv-w} and \ref{lemma:Astar-Binv-w}. 

\begin{lemma}
\label{lemma:series-tail-Astar-Binv-ones}
Suppose $T \geq 1$. There is some constant $C > 0$ such that if $d \le n \le d^2 / \log^C(d)$ then with high probability, it holds that
\[
\sum_{k=T}^\infty \alpha^{-(k+1)} \norm{\calA^*(\Delta^k 1_n)}_{op} =  \Tilde{O} \left( \frac{\sqrt{n}}{d} \right)^{T+1}.
\]
\end{lemma}
\begin{proof}
Note that 
\[
\norm{\calA^*(\Delta^k 1_n)}_{op} \leq \norm{\calA^*}_{2 \rightarrow op} \norm{\Delta}_{op}^k \norm{1_n}_2.
\]
By Lemma~\ref{lemma:B-concentrates}, $\alpha^{-1} \norm{\Delta}_{op} = \Tilde{O}(\sqrt{n}/d)$ with high probability by assumption on $n$. Combining these with the fact that $\norm{\calA^*}_{2 \rightarrow op} = O(d)$ with high probability when $n = o(d^2)$ (see Lemma 3 of \cite{saunderson2011subspace}), we may conclude by the geometric decay of the terms in the series that \[
\sum_{k=T}^\infty \alpha^{-(k+1)} \norm{\calA^*(\Delta^k 1_n)}_{op} = (\alpha^{-T} \norm{\Delta}^T_{op}) \cdot \Tilde{O}(\alpha^{-1}d\sqrt{n}) = \Tilde{O} \left( \frac{\sqrt{n}}{d} \right)^{T+1}.
\]
\end{proof}

\begin{lemma}
\label{lemma:Astar-Delta-k-ones}
Let $k \in \Z_{\geq 1}$ be fixed. Then with probability $1 - n^{-\Omega(1)}$, it holds that
\[
\norm{\calA^*(\Delta^k 1_n)}_{op} \leq (\log n)^{O(k)} \cdot \sqrt{d}n^{3/4} \cdot O( \sqrt{n}d)^k.
\]
\end{lemma}
The proof of this lemma is deferred to Section~\ref{sec:graph_matrices}.

\subsection{Proof of Lemma~\ref{lemma:Astar-Binv-w}}
Let $T \in \N$ be a truncation parameter. Using the same power series expansion as in the proof of Lemma~\ref{lemma:Astar-Binv-1} and the triangle inequality, we have that
\[
\norm{\calA^*(B^{-1}w)}_{op} \leq \sum_{k=0}^{T-1} \alpha^{-(k+1)} \norm{\calA^*(\Delta^k w)}_{op}  + \sum_{k=T}^\infty \alpha^{-(k+1)} \norm{\calA^*(\Delta^k w)}_{op}.
\]
Now, let $C > 0$ be an absolute constant whose value we determine in the following, let $T = C \log(d)$ be an integer and let $n = d^2 / \log^C (d)$. Our choice of $C>0$ depends on Lemmas~\ref{lemma:series-tail-Astar-Binv-w} and~\ref{lemma:Astar-Delta-k-w} that are stated below. There exists a sufficiently large choice of absolute constant $C>0$ such that invoking Lemma~\ref{lemma:series-tail-Astar-Binv-w} with $T = C \log (d)$ ensures the second summation above is $o(1)$ with high probability. There also exists a sufficiently large choice of absolute constant $C>0$ such that a union bound and invocation of Lemma~\ref{lemma:Astar-Delta-k-w},  for all $k \in \{0,\ldots, T-1\}$ ensures the first summation above is $o(1)$ with high probability. Setting $C$ to be the maximum of these two choices completes the proof.

\begin{lemma}
\label{lemma:series-tail-Astar-Binv-w}
Suppose $T \geq 1$. There is some constant $C > 0$ such that if $d \le n \le d^2 / \log^C(d)$ then with high probability, it holds that
\[
\sum_{k=T}^\infty \alpha^{-(k+1)} \norm{\calA^*(\Delta^k w)}_{op} = \sqrt{d}  \cdot \Tilde{O} \left( \frac{\sqrt{n}}{d} \right)^{T}.
\]
\end{lemma}
\begin{proof}
Note that 
\[
\norm{\calA^*(\Delta^k w)}_{op} \leq \norm{\calA^*}_{2 \rightarrow op} \norm{\Delta}_{op}^k \norm{w}_2.
\]
By Lemma~\ref{lemma:B-concentrates} and assumption on $n$, $\alpha^{-1} \norm{\Delta}_{op} = \Tilde{O}(\sqrt{n}/d)$ with high probability. A standard calculation (see the proof of Lemma \ref{lemma:wT-Binv-w}) reveals that $\norm{w}_2 = \Tilde{O}(\sqrt{nd})$ with high probability. Combining these with the fact that $\norm{\calA^*}_{2 \rightarrow op} = O(d)$ with high probability when $n = o(d^2)$ (see Lemma 3 of \cite{saunderson2011subspace}), we may conclude that \[
\sum_{k=T}^\infty \alpha^{-(k+1)} \norm{\calA^*(\Delta^k w)}_{op} = (\alpha^{-T} \norm{\Delta}^T_{op}) \cdot \Tilde{O}(\alpha^{-1}d^{3/2}\sqrt{n}) = \sqrt{d} \cdot \Tilde{O} \left( \frac{\sqrt{n}}{d} \right)^{T}.
\]
\end{proof}

\begin{lemma}
\label{lemma:Astar-Delta-k-w}
Let $k \in \Z_{\geq 0}$. Then with probability $1 - n^{-\Omega(1)}$, it holds that
\[
\norm{\calA^*(\Delta^k w)}_{op} \leq (\log n)^{O(k)} \cdot d\sqrt{n} \cdot O( \sqrt[4]{n}d^{3/2})^k.
\]
\end{lemma}

The proof of this lemma is deferred to Section~\ref{sec:graph_matrices}.

\subsection{Proof of Lemma~\ref{lemma:1T-Binv-w}}
Let $T \in \N$ be a truncation parameter. Using the same power series expansion as in the proof of Lemma~\ref{lemma:Astar-Binv-1} and the triangle inequality, we have that
\[
|1^T_n B^{-1}w| \leq \sum_{k=0}^{T-1} \alpha^{-(k+1)} |1^T_n\Delta^k w|  + \sum_{k=T}^\infty \alpha^{-(k+1)} |1^T_n\Delta^k w|.
\]
The argument requires Lemmas~\ref{lemma:series-tail-1T-Binv-w} and ~\ref{lemma:1T-Delta-k-w} stated below. Now, let $C > 0$ be some constant whose value we determine in the following, let $T = C \log(d)$ and let $n = d^2 / \log^C (d)$. There exists a sufficiently large choice of $C$ such that invoking Lemma~\ref{lemma:series-tail-1T-Binv-w} with $T = C \log (d)$ an integer ensures the second summation above is $o(1)$ with high probability. There also exists a sufficiently large choice of $C$ such that invoking Lemma~\ref{lemma:1T-Delta-k-w} for all $k \in \{0,\ldots, T-1\}$ with $\epsilon = o(1/T) = o(1/\log (d))$ ensures the first summation above is $o(1)$ with high probability. Setting $C$ to be the maximum of these two choices completes the proof.

\begin{lemma}
\label{lemma:series-tail-1T-Binv-w}
Suppose $T \geq 1$. There is some constant $C > 0$ such that if $d \le n \le d^2 / \log^C(d)$ then with high probability, it holds that
\[
\sum_{k=T}^\infty \alpha^{-(k+1)} |1^T_n\Delta^k w| =  \sqrt{d} \cdot \Tilde{O} \left( \frac{\sqrt{n}}{d} \right)^{T}.
\]
\end{lemma}
\begin{proof}
Note that 
\[
|1^T_n\Delta^k w| \leq \norm{1_n}_{2} \norm{\Delta}_{op}^k \norm{w}_2.
\]
By Lemma~\ref{lemma:B-concentrates} and assumption on $n$, we have $\alpha^{-1} \norm{\Delta}_{op} = \Tilde{O}(\sqrt{n}/d)$ with high probability when $n = o(d^2)$. A standard calculation (see the proof of Lemma \ref{lemma:wT-Binv-w}) reveals that $\norm{w}_2 = \Tilde{O}(\sqrt{nd})$ with high probability. Combining these, we may conclude that \[
\sum_{k=T}^\infty \alpha^{-(k+1)} |1^T_n\Delta^k w| = (\alpha^{-T} \norm{\Delta}^T_{op}) \cdot \Tilde{O}(\alpha^{-1}n\sqrt{d}) = \sqrt{d} \cdot \Tilde{O} \left( \frac{\sqrt{n}}{d} \right)^{T}.
\]
\end{proof}

\begin{lemma}
\label{lemma:1T-Delta-k-w}
Let $k \in \Z_{\geq 0}$. Then with probability $1 - n^{-\Omega(1)}$, it holds that
\[
|1^T_n\Delta^k w| \leq \norm{\calA^*(\Delta^k w)}_{op} \leq (\log n)^{O(k)} \cdot d\sqrt{n} \cdot O( \sqrt[4]{n}d^{3/2})^k.
\]
\end{lemma}
The proof of this lemma is nearly identical to that of Lemma~\ref{lemma:Astar-Binv-w} and is deferred to Section~\ref{sec:graph_matrices}. 

\section{Graph matrices}
\label{sec:graph_matrices}
\subsection{Background}

We use the theory of \textit{graph matrices} to derive operator norm bounds on various random matrices that arise in our analysis. Graph matrices provide a natural basis for decomposing matrices whose entries depend on random inputs, where this dependence has lots of symmetry but may be nonlinear. For our setting, we can define graph matrices as follows. These definitions are a special case of the definitions in \cite{ahn2016graph} and are equivalent to the definitions in \cite{ghosh2020sum} except that instead of summing over ribbons, we sum over injective maps. This gives a constant factor difference (see Remark 2.17 of \cite{ahn2016graph}) in the final norm bounds.

In our analysis, many of the matrices we study, such as $\Delta^k$, are $n \times n$ and have entries that are sums of terms of the form
\begin{align}
\label{eqn:generic_term}
M_{i_1, i_r} = \sum_{\substack{i_2, \ldots, i_{r-1} \\ k_1, \ldots, k_s} }
\prod_{(x,y) \in E}
f_{x, y}( v_{i_x, k_y} )
\end{align}
where $E \subset [r] \times [s]$, $ v_{i_x, k_y}$ is the $k_y$ coordinate of $v_{i_x}$, $\{ f_{x, y}\}$ are low-degree Hermite polynomials, and the indices of summation obey certain restrictions, including that $i_2, \ldots, i_{r-1}$ are distinct as well as $k_1, \ldots, k_{s}$.

The framework of graph matrices provides a convenient way of encoding these restrictions and attaining good norm bounds. Concretely, each matrix as in term \eqref{eqn:generic_term} can be represented by a `shape' consisting of a graph with $r$ circle vertices, $s$ square vertices, and integer edge labels. For a term like \eqref{eqn:generic_term} which is an $n \times n$ matrix, there are two distinguished circle vertices that represent $i_1$ and $i_r$. The edges in the shape are specified by $E \subset [r] \times [s]$, and the vertices specify (distinct) indices of summation. The remaining circle vertices each represent an index of summation over $1 \leq i \leq n$ (i.e., one of $i_2, \ldots, i_{r-1}$) and a square vertex is used to represent an index of summation over $1 \leq k \leq d$ (i.e., one of $k_1, \ldots, k_s$, each of which indexes the dimension). The integer edge labels of the shape denote the degree of the Hermite polynomial that is applied to the random variable $v_{i_x, k_y}$. We make this precise with the following definitions.

\begin{definition}[Normalized Hermite polynomials, see e.g.~\cite{odonnell}, Chapter 11.2]
\label{def:hermite} 
    Define the sequence of normalized Hermite polynomials $h_0, h_1, h_2,\ldots$ by
    \[
h_j(z)  = \frac{1}{\sqrt{j!}} \cdot H_j(z),
    \]
    where $H_j$ are defined uniquely by the following formal power series in $z$: 
    \[
\exp( tz - \frac{1}{2} t^2)
=\sum_{j = 0}^\infty \frac{1}{j!} H_j(z) t^j. 
    \]
\end{definition}
The first few Hermite polynomials are
\[
h_0(z) = 1, \;\;
h_1(z) = z, \;\; h_2(z) = \frac{1}{\sqrt{2}} (z^2 - 1), \;\;
h_3(z) = \frac{1}{\sqrt{6}}(z^3 - 3z), \;\;\ldots .
\]
Recall $\E_{Z \sim N(0,1)}[ h_j(Z) h_k(Z)] = \delta_{jk}$, where $\delta_{jk}$ denotes the Kronecker function.

\begin{definition}
A \textit{shape} $\alpha$ is a graph that consists of the following:
\begin{enumerate}
    \item A set of vertices $\calV(\alpha)$. Each vertex is either a \textit{square} or \textit{circle}. We take $\calV_{\circ}(\alpha)$ to be the set of circle vertices in $\calV(\alpha)$ and we take $\calV_{\Box}(\alpha)$ to be the set of square vertices in $\calV(\alpha)$.
    
    \item Distinguished tuples of vertices $U_{\alpha}$, $V_{\alpha}$ (which may intersect), which we call the \textit{left} and \textit{right} vertices of $\alpha$, respectively. We also define the set of \textit{middle} vertices as $W_{\alpha} = \calV(\alpha) \setminus (U_\alpha \cup V_\alpha)$\footnote{We abuse notation slightly by identifying tuples with the set composed of the union of their elements.}. We take $U_{\alpha,\circ}$ to be the circle vertices of $U_{\alpha}$ (in the same order) and we take $U_{\alpha,\Box}$ to be the square vertices of $U_{\alpha}$ (in the same order). Similarly, we take $V_{\alpha,\circ}$ to be the circle vertices of $V_{\alpha}$ (in the same order) and we take $V_{\alpha,\Box}$ to be the square vertices of $V_{\alpha}$ (in the same order). We always take $U_\alpha = (U_{\alpha, \circ}, U_{\alpha, \Box})$ and $V_\alpha = (V_{\alpha, \circ}, V_{\alpha, \Box})$ so that circle vertices preceed square vertices in order.
    
    \item A set $E(\alpha)$ of edges, where each edge is between a circle vertex and a square vertex. For each edge $e \in E(\alpha)$, we have a label $l_e \in \Z_{\geq 1}$. We define $|E(\alpha)| := \sum_{e \in E(\alpha)}{l_e}$. If a shape contains a multi-edge (i.e., two or more edges with the same endpoints), we call it \textit{improper} (and \textit{proper} otherwise). In a multi-edge, each copy of the edge has its own label. We represent an edge with endpoints $u$ and $v$ and label $l$ by the notation $\{u,v\}_l$; we use the simpler notation $\{u,v\}$ when $l=1$.
\end{enumerate}
\end{definition}

\begin{definition}
\label{def:graph_matrix}
Given a shape $\alpha$, we define $M_{\alpha}$ to be the $\frac{n!d!}{(n- |U_{\alpha,\circ}|)!(d- |U_{\alpha,\Box}|)!} \times \frac{n!d!}{(n- |V_{\alpha,\circ}|)!(d- |V_{\alpha,\Box}|)!}$ matrix with entries
\begin{equation}
\label{eqn:graph-matrix-entries}
M_{\alpha}(A,B) = \sum_{{{\pi_{\circ}: \calV_{\circ}(\alpha) \to [n], \pi_{\Box}: \calV_{\Box}(\alpha) \to [d]: \atop 
\pi_{\circ}, \pi_{\Box} \text{ are injective}} \atop 
\pi_{\circ}(U_{\alpha,\circ}) = A_{\circ}, \pi_{\Box}(U_{\alpha,\Box}) = A_{\Box},} \atop 
\pi_{\circ}(V_{\alpha,\circ}) = B_{\circ}, \pi_{\Box}(V_{\alpha,\Box}) = B_{\Box}}{\left(\prod_{e = \{u,v\} \in E(\alpha): u \in \calV_{\circ}(\alpha), v \in \calV_{\Box}(\alpha)}{h_{l_e}(v_{\pi_{\circ}(u), \pi_{\Box}(v)})}\right)}    
\end{equation}
where $A = ( A_{\circ}, A_{\Box})$ is an ordered tuple such that $A_{\circ}$ is an ordered tuple of $|U_{\alpha,\circ}|$ elements from $[n]$ and $A_{\Box}$ is an ordered tuple of $|U_{\alpha,\Box}|$ elements from $[d]$, and $B = ( B_{\circ}, B_{\Box})$ is an ordered tuple such that $B_{\circ}$ is an ordered tuple of $|V_{\alpha,\circ}|$ elements from $[n]$ and $B_{\Box}$ is an ordered tuple of $|V_{\alpha,\Box}|$ elements from $[d]$. 
\end{definition}
In the next section, we illustrate this definition by deriving the graph matrix representations of various matrices arising in our analysis. The proofs of Lemmas~\ref{lemma:Astar-Binv-1} and~\ref{lemma:Astar-Binv-w} boil down to obtaining norm bounds on $\calA^* (\Delta^k z)$ for $z \in \{1_n, w\}$. Such a matrix is $ d \times d$ and can be expressed as a sum of terms that are similar to \eqref{eqn:generic_term}:
\begin{align}
M_{k_1, k_s} = \sum_{\substack{i_1,i_2, \ldots, i_{r-1}, i_r \\ k_2, \ldots, k_{s-1}} }
\prod_{(x,y) \in E}  
f_{x, y}( v_{i_x, k_y} )
\end{align} 
where again this is a graph matrix, and restrictions on the indices are encoded by an associated shape as described in Definition~\ref{def:graph_matrix}. The difference between this and \eqref{eqn:generic_term} is that the distinguished vertices are now both squares instead of circles. Also, note that the restrictions imply that $i_1, \ldots, i_r$ are distinct, as well as $k_2, \ldots, k_{s-1}$.

\subsection{Graph matrix representations}
\label{sec:shapes}
In this section, we derive the graph matrix representations of various matrices that arise in our analysis. For the purposes of computing ${\calA}{\calA}^{*}$, we can view $\calA$ as an $n \times d^2$ matrix $A$ with rows indexed by $i \in [n]$ and columns indexed by an ordered pair of indices in $(j,k) \in [d] \times [d]$, with entry $A_{i, (j,k)} = (v_i)_j (v_i)_k$. Given this entry-wise expression, the correctness of the graph matrix representation of $A$ below can be directly verified by inspecting Equation~\eqref{eqn:graph-matrix-entries} for the shapes below.

We decompose $A$ as $A = M_{\alpha_{A1}} + M_{\alpha_{A2}}$ and $A^{*} = M_{\alpha_{A1}}^T + M_{\alpha_{A2}}^T$ for the following shapes $\alpha_{A1}$ and $\alpha_{A2}$ 
where we make the dimensions of $M_{\alpha_{A1}}$ and $M_{\alpha_{A2}}$ match by filling in the missing columns with zeros. These shapes are illustrated in Figure~\ref{fig:A-shapes}. Note that $\alpha_{A2}$ is improper.
For each shape $\alpha$ considered below, its vertices $\calV(\alpha)$ are given by $U_\alpha \cup V_\alpha \cup W_\alpha$:
\begin{itemize}
    \item $U_{\alpha_{A1}} = (u)$ where $u$ is a circle vertex, $V_{\alpha_{A1}} = (x_1,x_2)$ where $x_1,x_2$ are square vertices, $W_{\alpha_{A1}} = \{\}$ and $E(\alpha_{A1}) = \{\{u,x_1\}, \{u,x_2\}\}$. 
    The matrix $M_{\alpha_{A1}}$with zeros filled in for the columns of $M_{\alpha_{A2}}$ has dimensions $n \times d^2$. Its $(i, (j,k))$ entry, for $i \in [n]$ and $(j,k) \in [d] \times [d]$ with $j \neq k$, is given by:
    \[
    M_{\alpha_{A1}}(i, (j,k)) = h_1((v_i)_j) h_1 ((v_i)_k) = (v_i)_j (v_i)_k.
    \]
    Its $(i, (j,j))$ entry, for $i \in [n]$ and $j \in [d]$, is zero.
    
    \item $U_{\alpha_{A2}} = (u)$ where $u$ is a circle vertex, $V_{\alpha_{A2}} = (x,x)$ where $x$ is a square vertex, $W_{\alpha_{A2}} = \{\}$ and $E(\alpha_{A2}) = \{\{u,x\}, \{u,x\}\}$. 
    The matrix $M_{\alpha_{A2}}$with zeros filled in for the columns of $M_{\alpha_{A1}}$ has dimensions $n \times d^2$. Its $(i, (j,j))$ entry, for $i \in [n]$ and $j \in [d]$, is given by:
    \[
    M_{\alpha_{A1}}(i, (j,j)) = h_1((v_i)_j) h_1 ((v_i)_j) = (v_i)_j^2.
    \]
    Its $(i, (j,k))$ entry, for $i \in [n]$ and $j,k \in [d]$ with $j \neq k$, is zero.
\end{itemize}

\begin{figure}
    \centering
    \begin{subfigure}[t]{0.3\textwidth}
        \includegraphics[width=\textwidth]{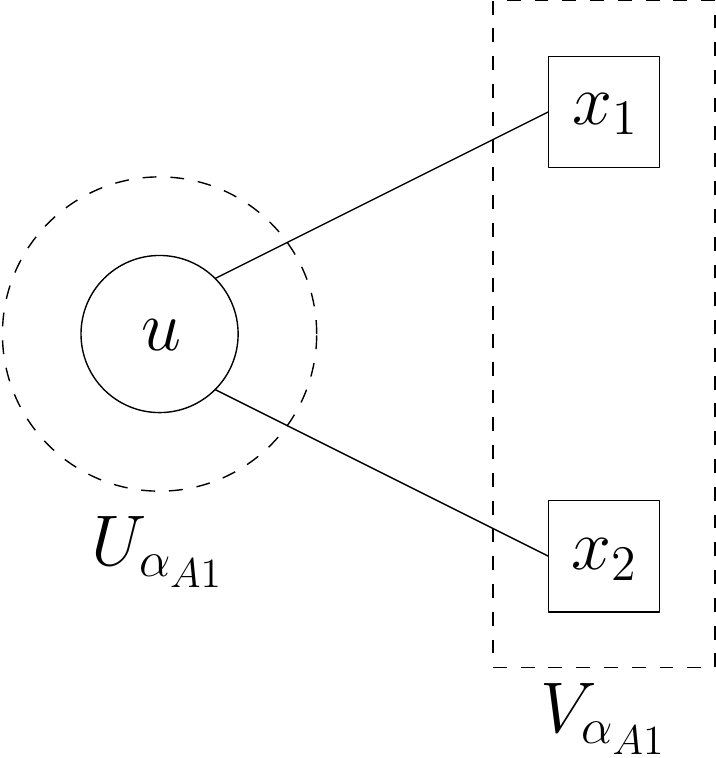}
        \caption{Shape $\alpha_{A1}$.}
        \label{fig:alphaA1}
    \end{subfigure}
    \qquad
    \begin{subfigure}[t]{0.3\textwidth}
        \includegraphics[width=\textwidth]{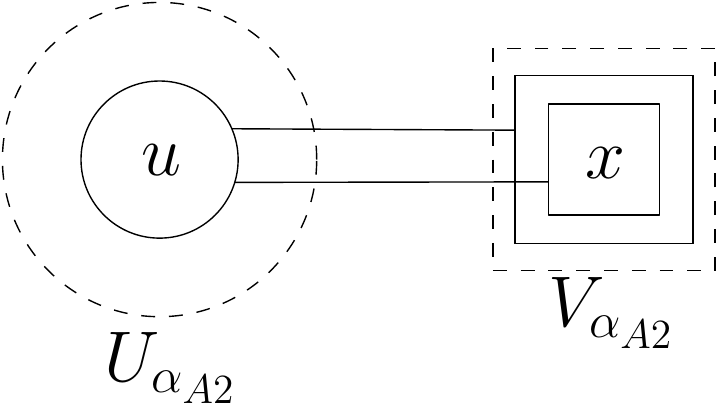}
        \caption{Shape $\alpha_{A2}$.}
        \label{fig:alphaA2}
    \end{subfigure}

    \caption{Shapes appearing in $A$. Here, we depict the shape $\alpha_{A2}$ by drawing the two identical copies of the square vertex $x$ as two overlapping squares sharing the label $x$. Note that the edge $\{u,x\}$ is a multi-edge, so the shape is improper.}
    \label{fig:A-shapes}
\end{figure} 

Multiplying $A$ and $A^{*}$, we see that ${\calA}{\calA}^{*} = {A}A^{*} \in \R^{n \times n}$ has $(i,j)$ entry $(\calA \calA^*)_{ij} = \ip{v_i}{v_j}^2$. We then obtain the following graph matrix representation:
\[
{\calA}{\calA}^{*} = M_{\alpha_1} + M_{\alpha_2} + M_{\alpha_3} + M_{\alpha_{4'}}
\]
where $\alpha_1$, $\alpha_2$, $\alpha_3$, and $\alpha_{4'}$ are the following shapes (note that $\alpha_2$, $\alpha_3$, and $\alpha_{4'}$ are improper):
\begin{itemize}
    \item $U_{\alpha_{1}} = (u)$ and  $V_{\alpha_{1}} = (v)$ where $u, v$ are circle vertices, $W_{\alpha_{1}} = \{x_1,x_2\}$ where $x_1,x_2$ are square vertices, and $E(\alpha_{1}) = \{\{u,x_1\}, \{u,x_2\}, \{x_1,v\},\{x_2,v\}\}$; see Figure~\ref{fig:alpha1}. The matrix $M_{\alpha_{1}}$ has dimensions $n \times n$. Its $(i, j)$ entry, for $i,j \in [n]$ with $i \neq j$ is given by:
    \[
    M_{\alpha_{1}}(i, j) = \sum_{k,l \in [d], k \neq l} h_1((v_i)_k) h_1((v_j)_k) h_1((v_i)_l) h_1((v_j)_l) = \sum_{k,l \in [d], k \neq l} (v_i)_k (v_j)_k (v_i)_l (v_j)_l.
    \]
    If $i = j$, then note that $M_{\alpha_1}(i,j) = 0$.
    
    \item $U_{\alpha_{2}} = (u)$ and  $V_{\alpha_{2}} = (v)$ where $u, v$ are circle vertices, $W_{\alpha_{2}} = \{x\}$ where $x$ is a square vertex, and $E(\alpha_{2}) = \{\{u,x\}, \{u,x\}, \{x,v\},\{x,v\}\}$; see Figure~\ref{fig:alpha2}. The matrix $M_{\alpha_{2}}$ has dimensions $n \times n$. Its $(i, j)$ entry, for $i,j \in [n]$ with $i \neq j$ is given by:
    \[
    M_{\alpha_{2}}(i, j) = \sum_{k \in [d]} h_1((v_i)_k)^2 h_1((v_j)_k)^2 = \sum_{k \in [d]} (v_i)^2_k (v_j)^2_k.
    \]
    If $i = j$, then note that $M_{\alpha_2}(i,j) = 0$.
    
    \item $U_{\alpha_{3}} = V_{\alpha_{3}} = (u)$ where $u$ is a circle vertex, $W_{\alpha_{3}} = \{x_1,x_2\}$ where $x_1,x_2$ are square vertices, and $E(\alpha_{3}) = \{\{u,x_1\}, \{u,x_1\}, \{u,x_2\}, \{u,x_2\}\}$; see Figure~\ref{fig:alpha3}.
    
    \item $U_{\alpha_{4'}} = V_{\alpha_{4'}} = (u)$ where $u$ is a circle vertex, $W_{\alpha_{4'}} = \{x\}$ where $x$ is a square vertex, and $E(\alpha_{4'}) = \{\{u,x\}, \{u,x\}, \{u,x\},\{u,x\}\}$; see Figure~\ref{fig:alpha4prime}.
\end{itemize}

\begin{figure}
    \centering
    \begin{subfigure}[t]{0.3\textwidth}
        \includegraphics[width=\textwidth]{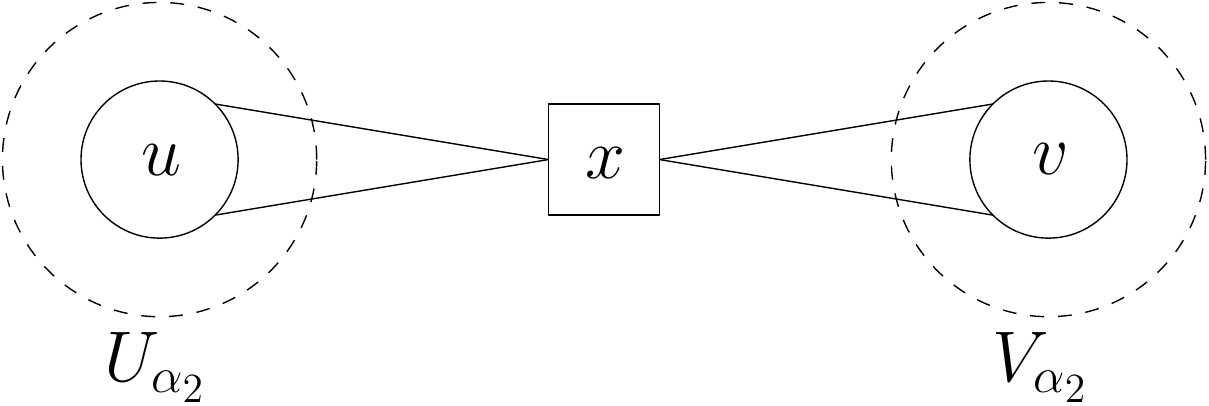}
        \caption{Shape $\alpha_{2}$.}
        \label{fig:alpha2}
    \end{subfigure}
    \qquad
    \begin{subfigure}[t]{0.3\textwidth}
        \includegraphics[width=\textwidth]{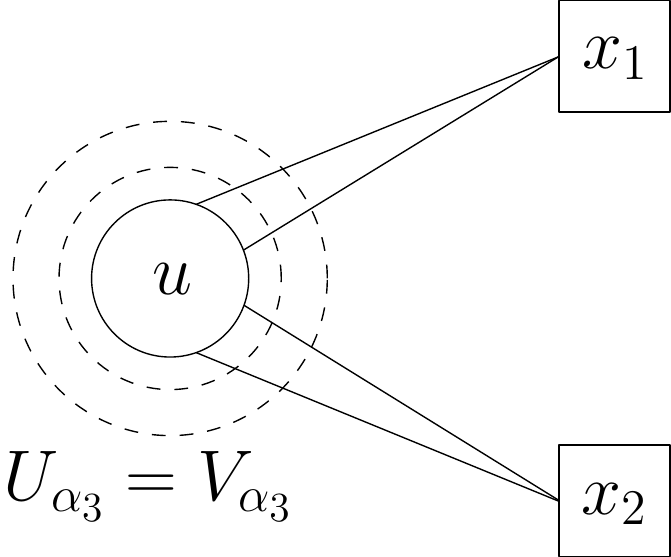}
        \caption{Shape $\alpha_{3}$.}
        \label{fig:alpha3}
    \end{subfigure}

    \begin{subfigure}[t]{0.3\textwidth}
        \includegraphics[width=\textwidth]{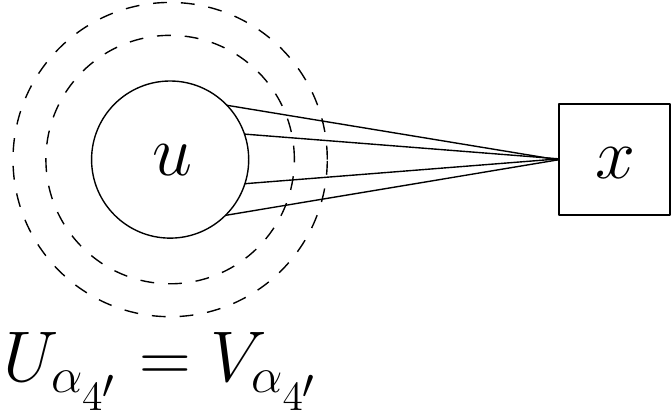}
        \caption{Shape $\alpha_{4'}$.}
        \label{fig:alpha4prime}
    \end{subfigure}

    \caption{Improper shapes appearing in $\calA \calA^*$.}\label{fig:AAstar-improper-shapes}
\end{figure} 

Finally, we express the vectors $w, 1_n \in \R^n$ as $n \times 1$ graph matrices. Recall that $w_i = \norm{v_i}_2^2 -d$ and that $h_2(z) = \frac{1}{\sqrt{2}} (z^2 - 1)$. So, $w$ is represented by the shape $\alpha_w$ with leading coefficient $\sqrt{2}$, and $1_n$ is represented by the shape $\alpha_{1_n}$ with leading coefficient $1$:
\begin{itemize}
    \item $U_{\alpha_w} = (u)$ where $u$ is a circle vertex, $V_{\alpha_w} = \emptyset$, $W_{\alpha_w} = \{x\}$ where $x$ is a square vertex, and $E(\alpha_w) = \{ \{u, x\}_2 \}$.
    \item $U_{\alpha_{1_n}} = (u)$ where $u$ is a circle vertex, $V_{\alpha_{1_n}} = \emptyset$, and $E(\alpha_{1_n}) = \emptyset$. 
\end{itemize}

\subsubsection{Resolving multi-edges}
As we demonstrate later, it is important for the purposes of our analysis that all shapes we work with are proper. To shift from improper shapes (i.e.\ ones with multi-edges) to proper shapes (i.e.\ ones without multi-edges), we record the following proposition.

\begin{proposition}
\label{prop:resolve-multi-edges}
Let $\alpha$ be a shape which contains two or more copies of an edge $e$. Consider two such copies of $e$ that have labels $i,j \in \Z_{\geq 1}$, respectively. Then, we have
\[
M_{\alpha} = \sum_{k=0}^\infty c_k M_{\alpha_k},
\]
where $\alpha_k$ is the shape that is identical to $\alpha$, except that the two labeled copies of $e$ are replaced by a single copy of $e$ with label $k$, and $\{c_k: k \in \Z_{\geq 0}\}$ are coefficients that satisfy:
\[
h_i(x)h_j(x) = \sum_{k = 0}^{\infty}{c_{k}h_k(x)}.
\]
That is, the coefficients are obtained by writing the polynomial $h_i \cdot h_j$ in the Hermite basis. In particular, it holds that $c_{k} = 0$ unless $i + j + k$ is even and $k \leq i + j$. In other words, for each term we obtain, the parity of $k$ is the same as the parity of $i + j$.
We regard any edge with label 0 as a non-edge and say that such an edge \textit{vanishes}.
\end{proposition}
Note that we may convert two or more parallel labeled edges into a single labeled edges by repeated application of Proposition \ref{prop:resolve-multi-edges}. The proof of this result follows from Definition~\ref{def:graph_matrix}. The parity result follows from elementary calculations involving the Hermite polynomials, which we defer to Section~\ref{sec:hermite-calculations}.

Given Proposition~\ref{prop:resolve-multi-edges}, we replace the improper shapes from the previous section with proper shapes to obtain the following graph matrix representation:
\begin{align*}
{\calA}{\calA}^{*} &= M_{\alpha_1} + 2M_{\alpha_{2a}} + \sqrt{2}M_{\alpha_{2b}} + \sqrt{2}M_{\alpha_{2c}} + dM_{\alpha_{2d}} + 2M_{\alpha_{3a}} + 2\sqrt{2}(d-1)M_{\alpha_{3b}} + (d^2 - d)M_{\alpha_{3c}} \\
&\qquad+ \sqrt{24}M_{\alpha_{4}} + 6\sqrt{2}M_{\alpha_{3b}} + 3dM_{\alpha_{3c}}   
\end{align*}
where $\alpha_1$, $\alpha_{2a}$, $\alpha_{2b}$, $\alpha_{2c}$, $\alpha_{2d}$, $\alpha_{3a}$, $\alpha_{3b}$, $\alpha_{3c}$, and $\alpha_{4}$ are the following proper shapes that we define below. First, $\alpha_1$ is the same as above since it is already proper. 

Second, observe that $\alpha_2$ has two sets of double edges. Using the following identities (see also Section~\ref{sec:hermite-calculations}):
\begin{align*}
h_1(x)^2 &= \sqrt{2} h_2(x) + 1 \\
h_1(x)^2 h_1(y)^2 &= 2 h_2(x) h_2(y) + \sqrt{2} h_2(x) + \sqrt{2} h_2(y) + 1,
\end{align*}
we can replace each of these double edges by a linear combination of an edge with label 2 and a non-edge. We then write $M_{\alpha_2} = 2M_{\alpha_{2a}} + \sqrt{2} M_{\alpha_{2b}} + \sqrt{2} M_{\alpha_{2c}} + d \cdot M_{\alpha_{2d}}$ as a linear combination of $2 \times 2 = 4$ graph matrices associated with the shapes $\alpha_{2a}, \alpha_{2b}, \alpha_{2c}, \alpha_{2d}$ defined as follows: 
\begin{itemize}
    \item $U_{\alpha_{2a}} = (u)$ and  $V_{\alpha_{2a}} = (v)$ where $u, v$ are circle vertices, $W_{\alpha_{2a}} = \{x\}$ where $x$ is a square vertex, and $E(\alpha_{2a}) = \{\{u,x\}_2, \{x,v\}_2\}$. The matrix $M_{\alpha_{2a}}$ has dimensions $n \times n$. Its $(i, j)$ entry, for $i,j \in [n]$ with $i \neq j$ is given by:
    \[
    M_{\alpha_{2a}}(i, j) = \sum_{k \in [d]} h_2((v_i)_k) h_2((v_j)_k).
    \]
    If $i = j$, then note that $M_{\alpha_{2a}}(i,j) = 0$.
    
    \item $U_{\alpha_{2b}} = (u)$ and  $V_{\alpha_{2b}} = (v)$ where $u, v$ are circle vertices, $W_{\alpha_{2b}} = \{x\}$ where $x$ is a square vertex, and $E(\alpha_{2b}) = \{\{u,x\}_2\}$.
    \item $U_{\alpha_{2c}} = (u)$ and  $V_{\alpha_{2c}} = (v)$ where $u, v$ are circle vertices, $W_{\alpha_{2c}} = \{x\}$ where $x$ is a square vertex, and $E(\alpha_{2c}) = \{\{x,v\}_2\}$.
    \item $U_{\alpha_{2d}} = (u)$ and  $V_{\alpha_{2d}} = (v)$ where $u, v$ are circle vertices, $W_{\alpha_{2d}} = \{\}$, and $E(\alpha_{2d}) = \{\}$. Note that we have made the following simplification in describing $\alpha_{2d}$, which arises when we replace each of the double edges in $\alpha_2$ by non-edges. This will leave $\alpha_{2d}$ with an isolated middle square vertex $x$. However, observe that from Equation~\eqref{eqn:graph-matrix-entries}, we may equivalently delete the isolated vertex $x$ and multiply the resulting graph matrix by a $d$ factor. In summary, we work with the definition of $\alpha_{2d}$ which does not contain a middle square vertex, but which has an associated scalar coefficient of $d$.
\end{itemize}

Third, using the same approach as in re-expressing $\alpha_2$, we write $M_{\alpha_3} = 2M_{\alpha_{3a}} + 2\sqrt{2} (d-1) M_{\alpha_{3b}} + d(d-1) M_{\alpha_{3c}}$ as a linear combination of $3$ graph matrices associated with the shapes $\alpha_{3a}, \alpha_{3b}, \alpha_{3c}$ defined as follows:
\begin{itemize}
    \item $U_{\alpha_{3a}} = V_{\alpha_{3a}} = (u)$ where $u$ is a circle vertex, $W_{\alpha_{3a}} = \{x_1,x_2\}$ where $x_1,x_2$ are square vertices, and $E(\alpha_{3a}) = \{\{u,x_1\}_2, \{u,x_2\}_2\}$.
    \item $U_{\alpha_{3b}} = V_{\alpha_{3b}} = (u)$ where $u$ is a circle vertex, $W_{\alpha_{3b}} = \{x\}$ where $x$ is a square vertex, and $E(\alpha_{3b}) = \{\{u,x\}_2\}$.
    \item $U_{\alpha_{3c}} = V_{\alpha_{3c}} = (u)$ where $u$ is a circle vertex, $W_{\alpha_{3c}} = \{\}$, and $E(\alpha_{3c}) = \{\}$.
\end{itemize}

Fourth, for $\alpha_{4'}$, we use the following identity to replace its quadruple edge by a linear combination of an edge with label 4, an edge with label 2, and a non-edge (see also Section~\ref{sec:hermite-calculations}):
\[
h_1(x)^4 = (x^4 - 6x^2 + 3) + 6(x^2 - 1) + 3 = \sqrt{24}h_4(x) + 6\sqrt{2}h_2(x) + 3.
\]
We write $M_{\alpha_{4'}} = \sqrt{24} M_{\alpha_{4}} + 6\sqrt{2} M_{\alpha_{3b}} + 3d M_{\alpha_{3c}}$ as a linear combination of $3$ graph matrices associated with the shapes $\alpha_{4}, \alpha_{3b}, \alpha_{3c}$, where $\alpha_4$ is defined as follows:
\begin{itemize}
    \item $U_{\alpha_{4}} = V_{\alpha_{4}} = (u)$ where $u$ is a circle vertex, $W_{\alpha_{4}} = \{x\}$ where $x$ is a square vertex, and $E(\alpha_{4}) = \{\{u,x\}_4\}$.
\end{itemize}
We may further simplify by observing that $M_{\alpha_{2d}} = 1_{n}1_{n}^T - I_n$ and $M_{\alpha_{3c}} = I_n$, which leads to following graph matrix representation involving only proper shapes: 
\[
{\calA}{\calA}^{*} = (d^2 + d)I_{n} + d{1_n}{1_n^T} + M_{\alpha_1} + 2M_{\alpha_{2a}} + \sqrt{2}M_{\alpha_{2b}} + \sqrt{2}M_{\alpha_{2c}} + 2M_{\alpha_{3a}} + (2\sqrt{2}d + 4\sqrt{2})M_{\alpha_{3b}} + \sqrt{24}M_{\alpha_{4}}.
\]

In order to decompose $B = {\calA}{\calA}^{*} - W$ in terms of graph matrices, we first decompose $W$ as follows:
\[
W = w 1_n^T + 1_n w^T + d 1_n 1_n^T = \sqrt{2}M_{\alpha_{2b}} + \sqrt{2}M_{\alpha_{2c}} + 2\sqrt{2}M_{\alpha_{3b}} + d{1_n} 1_n^T.
\]
Combining these decompositions, we have:
\begin{align}
    \label{eqn:Bdef}
    B &= \calA \calA^* - W = (d^2 + d)I_{n} + M_{\alpha_1} + 2M_{\alpha_{2a}} +  2M_{\alpha_{3a}} + (2\sqrt{2}d + 2\sqrt{2})M_{\alpha_{3b}} + \sqrt{24}M_{\alpha_{4}}, \\
    \Delta &= -M_{\alpha_1} -2M_{\alpha_{2a}} - 2M_{\alpha_{3a}} - (2\sqrt{2}d + 2\sqrt{2})M_{\alpha_{3b}} - \sqrt{24}M_{\alpha_{4}}.
    \label{eqn:Delta_def} 
\end{align}
Define the index set $\calI = \{1,2a,3a,3b,4\}$, which collects the indices of non-identity shapes appearing in $B$; see Figure~\ref{fig:delta-shapes}. For a given index $i \in \calI$, we define $\lambda_i$ to be the scalar coefficient appearing in front of $M_{\alpha_i}$ in the expression for $B$ above.

\begin{figure}
    \centering
    \begin{subfigure}[b]{0.3\textwidth}
        \includegraphics[width=\textwidth]{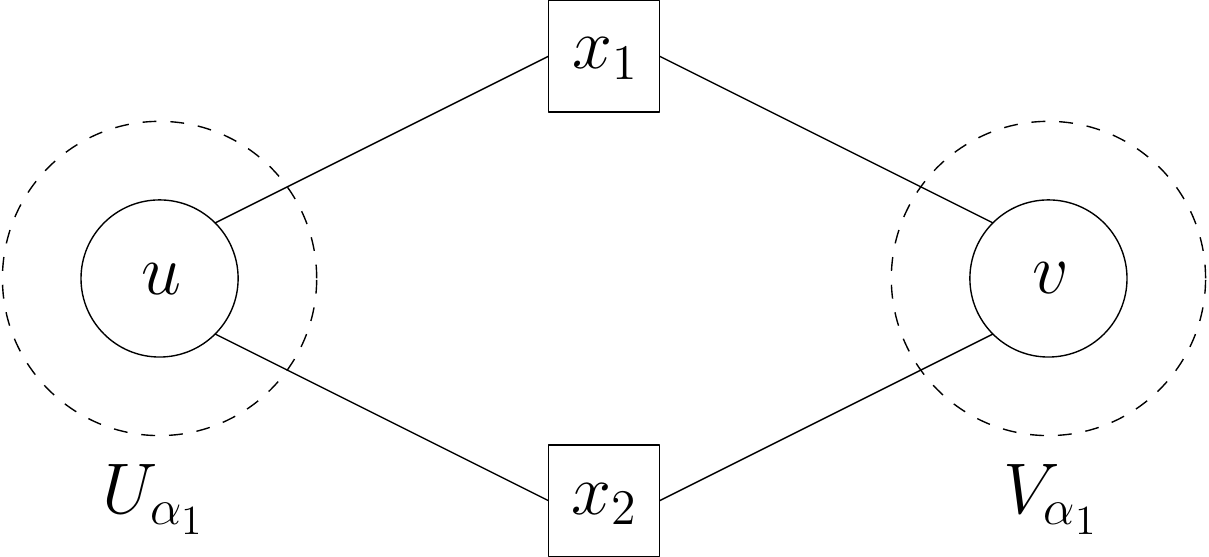}
        \caption{Shape $\alpha_1$}
        \label{fig:alpha1}
    \end{subfigure}

    \begin{subfigure}[b]{0.3\textwidth}
        \includegraphics[width=\textwidth]{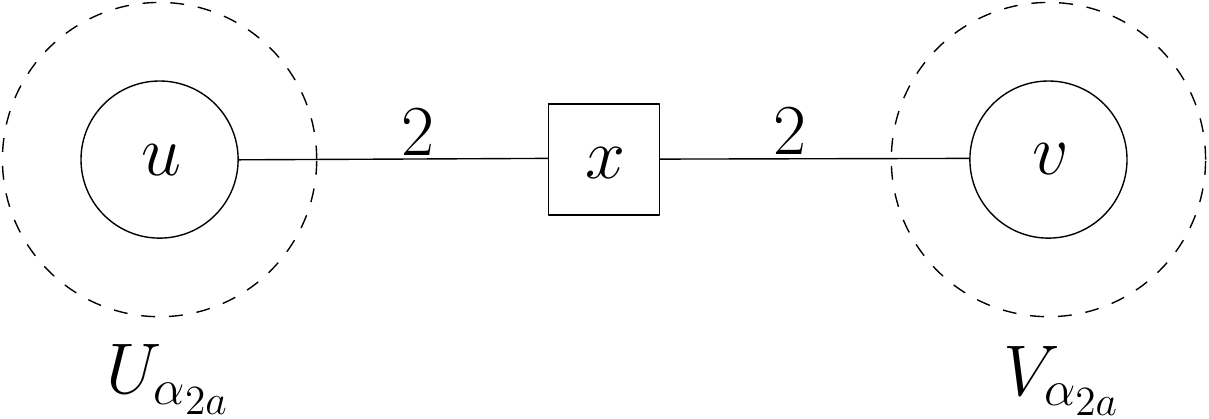}
        \caption{Shape $\alpha_{2a}$}
        \label{fig:alpha2a}
    \end{subfigure}

    \begin{subfigure}[b]{0.3\textwidth}
        \includegraphics[width=\textwidth]{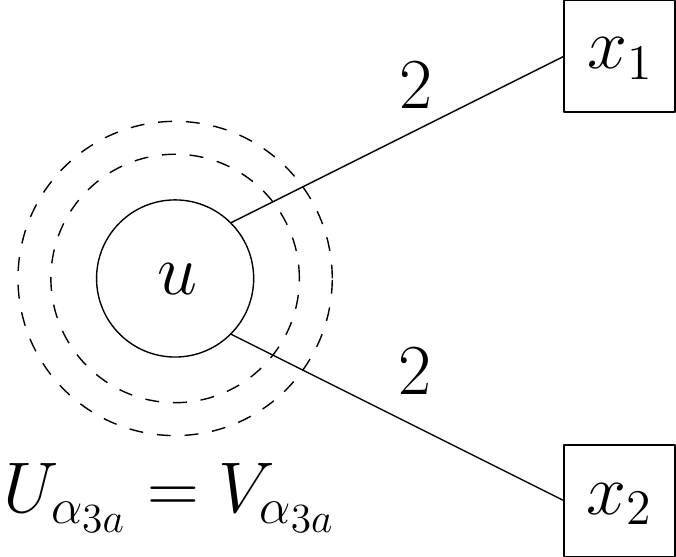}
        \caption{Shape $\alpha_{3a}$}
        \label{fig:alpha3a}
    \end{subfigure}
    
    \begin{subfigure}[b]{0.3\textwidth}
        \includegraphics[width=\textwidth]{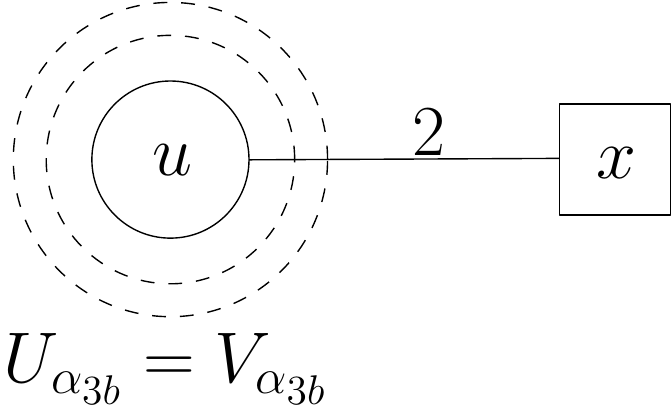}
        \caption{Shape $\alpha_{3b}$}
        \label{fig:alpha3b}
    \end{subfigure}
    
    \begin{subfigure}[b]{0.3\textwidth}
        \includegraphics[width=\textwidth]{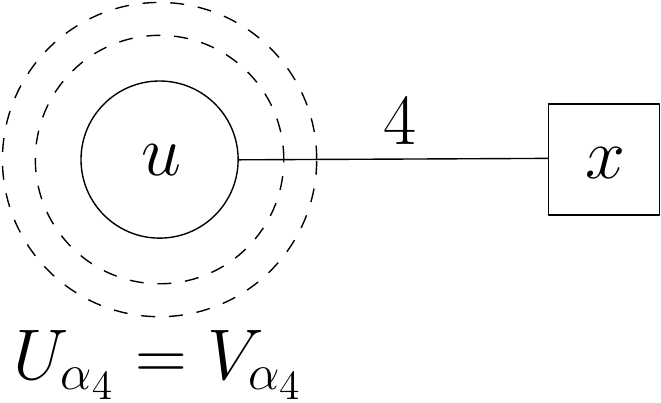}
        \caption{Shape $\alpha_{4}$}
        \label{fig:alpha4}
    \end{subfigure}
    \caption{Shapes appearing in $\Delta$.}\label{fig:delta-shapes}
\end{figure} 

\subsection{Graph matrix norm bounds}
As mentioned earlier, graph matrices admit norm bounds that only depend on certain combinatorial parameters associated with the graph, as expressed by the following theorem that follows from~\cite{ahn2016graph}. We defer its proof to Appendix \ref{sec:norm_bounds}. To complete the proofs of Lemmas~\ref{lemma:B-concentrates}, \ref{lemma:Astar-Delta-k-ones}, \ref{lemma:Astar-Delta-k-w}, and \ref{lemma:1T-Delta-k-w}, we will derive graph matrix representations of the relevant matrices, estimate their associated combinatorial parameters, and then invoke the theorem below. To state the theorem, we borrow some notions from \cite{ahn2016graph}.

Given a shape $\alpha$ and a vertex $a \in \calV(\alpha)$, we define 
\[
\varphi(a) =\begin{cases}
1 \quad &\text{if } a \text{ is a circle,} \\
\log_n(d) \quad &\text{if } a \text{ is a square.}
\end{cases} 
\]
For a (potentially empty) subset of vertices $S \subset \calV(\alpha)$, define $\varphi(S) = \sum_{s \in S} \varphi(s)$.
\begin{definition}
Define the \textit{min-vertex separator} $S_{\min}$ of $\alpha$ to be a (potentially empty) subset of  vertices of $ \alpha$ with the smallest value of $\varphi(S)$ such that all paths from $U_\alpha$ to $V_\alpha$ intersect $S$. Here, we allow for paths of length $0$, so any separator between $U_\alpha$ and $V_\alpha$ must contain $U_\alpha \cap V_\alpha$. 
We also define $\mathrm{Iso}(\alpha)$ to consist of all isolated vertices lying in $\calV(\alpha)\backslash ( U_\alpha \cup V_\alpha)$. 
\end{definition}

\begin{theorem}
\label{thm:graph-matrix-norm-bound}
Given $D_V,D_E \in \mathbb{N}$ such that $D_E \geq D_V \geq 2$ and $\epsilon > 0$, with probability at least $1 - \epsilon$, for all shapes $\alpha$ on square and circle vertices such that $|\calV(\alpha)| \leq D_V$, $|E_{\alpha}| \leq D_E$, $|U_{\alpha}| \leq 1$, and $|V_{\alpha}| \leq 1$,
\[
\|M_{\alpha}\|_{op} \leq \left((2D_E + 2) \ln(D_V) + \ln(11n) + \ln\left(\frac{1}{\epsilon}\right)\right)^{|\calV(\alpha)| + |E(\alpha)|}n^{\frac{\varphi(\calV(\alpha)) - \varphi(S_{\text{min}}) + \varphi(\mathrm{Iso}(\alpha))}{2}}
\]
where $S_{\text{min}}$ is a min-vertex separator of $\alpha$.
\end{theorem}

In our applications of this result, we invoke it with $\epsilon = 1/\poly(n)$ on shapes $\alpha$ for which $|\calV(\alpha)|, |E_\alpha
| \leq \polylog(n)$, so that the norm bounds take the form:
\[
\| M_\alpha \|_{op} = (\log n)^{O(|\calV(\alpha)| + |E(\alpha)|)} \cdot n^{\frac{\varphi(\calV(\alpha)) - \varphi(S_{\text{min}}) + \varphi(\mathrm{Iso}(\alpha))}{2}}.
\]
With Theorem~\ref{thm:graph-matrix-norm-bound} and the graph matrix representations from Section~\ref{sec:shapes} in hand, we can immediately prove Lemma~\ref{lemma:B-concentrates}. The proof also demonstrates how the usage of graph matrices allows us to reduce the challenging problem of bounding the norm of a ``complicated'' random matrix to a significantly simpler combinatorial problem.

\begin{proof}[Proof of Lemma~\ref{lemma:B-concentrates}]
Recall from Equation~\eqref{eqn:Bdef} that we have
\[
\norm{B- \alpha I_n}_{op} \leq \norm{M_{\alpha_1}}_{op} + 2 \norm{M_{\alpha_{2a}}}_{op} +  2 \norm{M_{\alpha_{3a}}}_{op} + (2\sqrt{2}d + 2\sqrt{2}) \norm{M_{\alpha_{3b}}}_{op} + \sqrt{24} \norm{M_{\alpha_{4}}}_{op}.
\]
To complete the proof, we will invoke Theorem~\ref{thm:graph-matrix-norm-bound} to upper bound the norm of each of the 5 graph matrices above. In the following, for each of the 5 shapes, we identify the min-vertex separator and then estimate the combinatorial parameters appearing in the bound in Theorem~\ref{thm:graph-matrix-norm-bound}.

Recall that for $a \in \calV(\alpha)$, we have $\varphi(a) = \log_n(n)=1$ if $a$ is a circle and $\varphi(a) = \log_n(d) \approx 1/2$ if $a$ is a square (since we consider the regime $n \leq d^2/\polylog(d)$) and that $\mathrm{Iso}(\alpha)$ is the set of isolated vertices that do not lie in $U_\alpha$ or $V_\alpha$.

\begin{itemize}
    \item \textbf{Term $M_{\alpha_1}$:} Consider the following vertex separators: $\{ u\}, \{ v \}, \{ x_1, x_2 \}$. By inspection, any other vertex separator contains one of these three. The weights of $\{ u\}$  and $\{ v \}$ are both $1$. The weight of $\{ x_1, x_2\}$ is $2\log_n(d) > 1$. Thus we may choose $u$ as a min-vertex separator without loss of generality. Thus for $\alpha_1$, we have
	\[
	\frac{\varphi(\calV(\alpha)) - \varphi(S_{\text{min}}) + \varphi(\mathrm{Iso}(\alpha))}{2}
	= \frac{ (2 \log_n(d) + 2) - 1 + 0 }{2} 
	= \frac{2 \log_n(d) + 1}{2},
	\]
	leading to a norm bound $\|M_{\alpha_1}\|_{op} = \tilde O( d\sqrt{n})$ with high probability by Theorem \ref{thm:graph-matrix-norm-bound}. 
	
	
	\item \textbf{Term $M_{\alpha_{2a}}$:} Every vertex is a separator of $U_{\alpha_{2a}}$ and $V_{\alpha_{2a}}$. Since $x$ has weight $\log_n(d)<1$ and $u, v$ have weight $1$ in $\alpha_2$, the minimum weight vertex separator is $x$. Thus for $\alpha_{2a}$, we have
	\[
	\frac{\varphi(\calV(\alpha)) - \varphi(S_{\text{min}}) + \varphi(\mathrm{Iso}(\alpha))}{2}
	= \frac{  ( 2 + \log_n(d)  ) - \log_n(d) + 0}{2}
	= 1,
	\] 
	leading to a high probability norm bound $\|M_{\alpha_{2a}}\|_{op} = \tilde O(n)$ by Theorem \ref{thm:graph-matrix-norm-bound}.
	
\end{itemize}
The remaining shapes represent matrices that are diagonal; thus $u$ is the min-vertex separator. 
\begin{itemize}
    \item \textbf{Term $M_{\alpha_{3a}}$:} 
	By similar arguments to the above,  Theorem \ref{thm:graph-matrix-norm-bound} yields $\| M_{\alpha_{3a}} \| =  \tilde O( d ) $.
	\item \textbf{Term $M_{\alpha_{3b}}$:} We obtain a norm bound $\tilde O( n^{(1+ \log_n d - 1 + 0)/2} = \tilde O( d^{1/2} )$, so its contribution to $\Delta$ has operator norm at most $\tilde O(d^{3/2})$ (see \eqref{eqn:Delta_def}).
	\item \textbf{Term $M_{\alpha_{4}}$:} Similarly, we obtain the norm bound $\tilde O( d^{1/2})$.
\end{itemize}
Assembling these bounds completes the proof.
\end{proof}

\subsection{Tools for dealing with products of shapes}
\label{sec:product_rules}
As mentioned earlier, our analysis involves large powers the matrix $\Delta$. In the following, we will derive a graph matrix representation for $\Delta$. Unfortunately, explicitly writing down such a representation for $\Delta^k$ for arbitrary $k \in \N$ is complicated. To overcome this issue, we now introduce some definitions and technical results that allow us to express in a systematic way the graph matrix representation of a product of matrices in terms of the representations of the individual matrices. The following result follows directly from the formula in Definition~\ref{def:graph_matrix}.

\begin{proposition}[Multiplication rule]
\label{prop:mult-rule}
Given shapes $\alpha$ and $\beta$ such that $V_{\alpha}$ and $U_{\beta}$ match (i.e.\ $V_{\alpha}$ and $U_{\beta}$ have the same number of circle and square vertices), the product $M_{\alpha}M_{\beta}$ is a linear combination of graph matrices $M_\gamma$ of shapes $\gamma$ of the following form:
\begin{enumerate}
    \item Glue $\alpha$ and $\beta$ together by setting $V_{\alpha} = U_{\beta}$; these vertices now become middle vertices of $\gamma$\footnote{Note that if, say, $V_{\alpha}$ has repeated vertices, then $U_{\beta}$ must have the same number and type of repeated vertices. Otherwise, the dimensions of the two matrices are not compatible for multiplication.}. We set $U_\gamma = U_{\alpha}$ to be the left side of $\gamma$ and we set $V_\gamma = V_{\beta}$ to be the right side of $\gamma$.
    \item The possible realizations of $\gamma$ are obtained by considering all possible ways in which the vertices in $\calV(\beta) \setminus U_{\beta}$ may intersect with the vertices in $\calV(\alpha) \setminus V_{\alpha}$. For a possible intersection of $\calV(\beta) \setminus U_{\beta}$ and $\calV(\alpha) \setminus V_{\alpha}$ to give rise to $\gamma$, it must satisfy the following constraints:
    \begin{enumerate}
        \item For a given intersection of $\calV(\beta) \setminus U_{\beta}$ and $\calV(\alpha) \setminus V_{\alpha}$ and a vertex $v$ which is in this intersection, we say that the occurrence of $v$ in $\calV(\beta) \setminus U_{\beta}$ is \textit{identified with} the occurrence of $v$ in $\calV(\alpha) \setminus V_{\alpha}$. Circle vertices can only be identified with other circle vertices and square vertices can only be identified with other square vertices.
        \item The vertices in $\calV(\alpha) \setminus V_{\alpha}$ must remain distinct as well as the vertices in $\calV(\beta) \setminus U_{\beta}$. In other words, each vertex can only be identified with at most one other vertex (which must be in the other shape).
    \end{enumerate}
\end{enumerate}
\end{proposition}
\noindent We illustrate Proposition~\ref{prop:mult-rule} in Figure~\ref{fig:mult-example} by multiplying two shapes that arise in the multiplication $\Delta w$.
\begin{figure}
    \includegraphics[width=\textwidth]{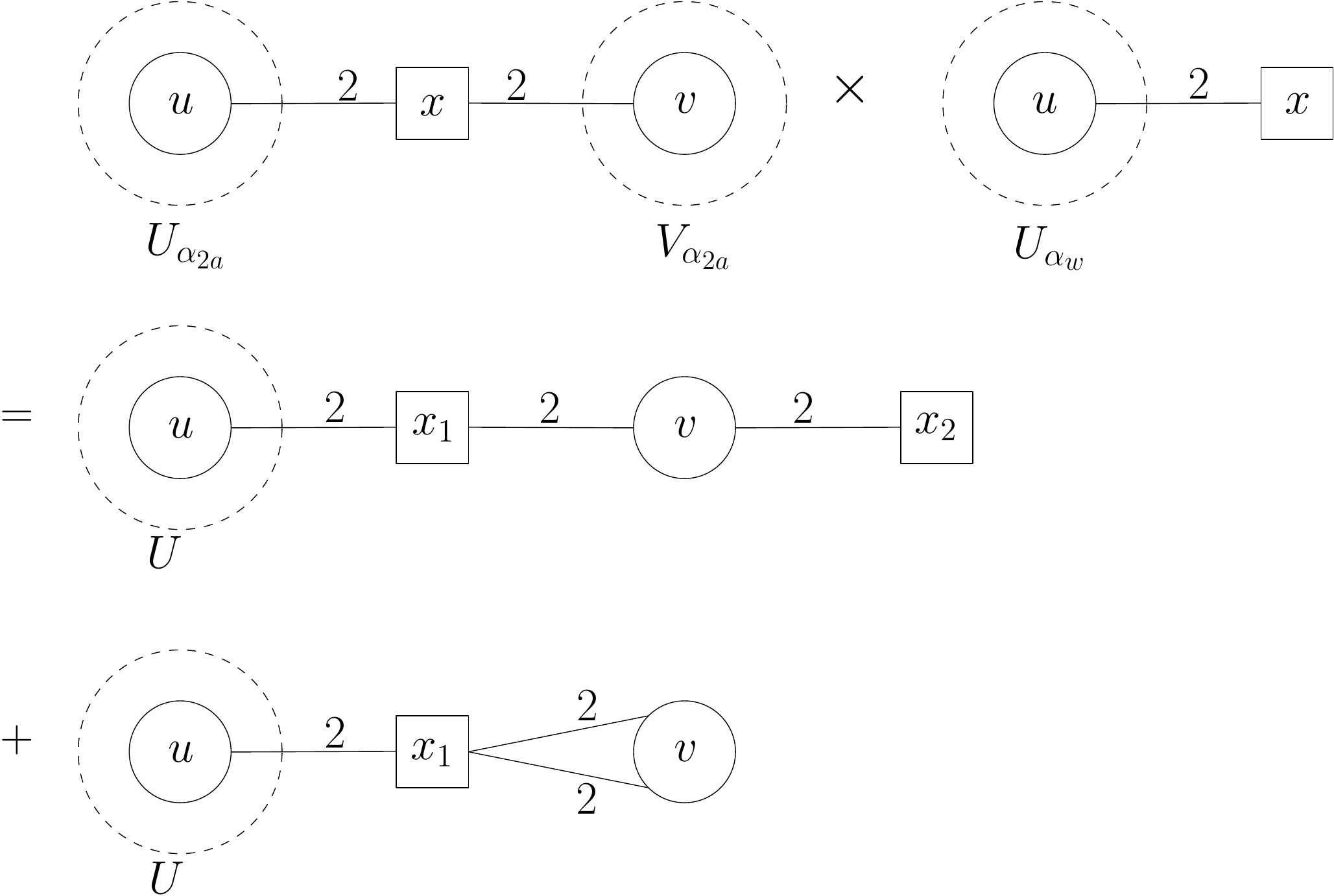}
    \caption{Multiplication of $M_{\alpha_{2a}}$ and $M_{\alpha_w}$. Since $V_{\alpha_w} = \emptyset$, it is not depicted above; as result, the shapes in the resulting product also have $V = \emptyset$. Also, note that the second shape of the resulting product has a multi-edge. Using Proposition~\ref{prop:resolve-multi-edges} as before, we can express this multi-edge as a linear combination of 3 labeled edges, with labels 0, 2, and 4, respectively.}\label{fig:mult-example}
\end{figure}

\subsubsection{Action of $\calA^*$ as a graph matrix}
In this section, we record for later use how to express the action of the linear operator $\calA^*$ in terms of graph matrices. Taking the transpose of the shapes $\alpha_{A1}, \alpha_{A2}$ from Section~\ref{sec:shapes} and applying Proposition~\ref{prop:resolve-multi-edges} to resolve multi-edges, we define the following shapes $\alpha_{A^{*},1}$, $\alpha_{A^{*},2}$, and $\alpha_{A^{*},3}$; see also Figure~\ref{fig:Astar}.
\begin{enumerate}
    \item $U_{\alpha_{A^{*},1}} = (u)$ and  $V_{\alpha_{A^{*},1}} = (v)$ where $u, v$ are square vertices, $W_{\alpha_{A^{*},1}} = \{x\}$ where $x$ is a circle vertex, and $E(\alpha_{A^{*},1}) = \{\{u,x\}, \{x,v\}\}$. This shape has an associated coeffcient of $1$.
    \item $U_{\alpha_{A^{*},2}} = V_{\alpha_{A^{*},2}} = (u)$ where $u$ is a square vertex, $W_{\alpha_{A^{*},2}} = \{x\}$ where $x$ is a circle vertex, and $E(\alpha_{A^{*},2}) = \emptyset$. This shape has an associated coeffcient of $1$.
    \item $U_{\alpha_{A^{*},3}} = V_{\alpha_{A^{*},3}} = (u)$ where $u$ is a square vertex, $W_{\alpha_{A^{*},3}} = \{x\}$ where $x$ is a circle vertex, and $E(\alpha_{A^{*},3}) = \{\{u,x\}_2\}$. This shape has an associated coeffcient of $\sqrt{2}$.
\end{enumerate}

\begin{figure}
    \centering
    \begin{subfigure}[b]{0.3\textwidth}
        \includegraphics[width=\textwidth]{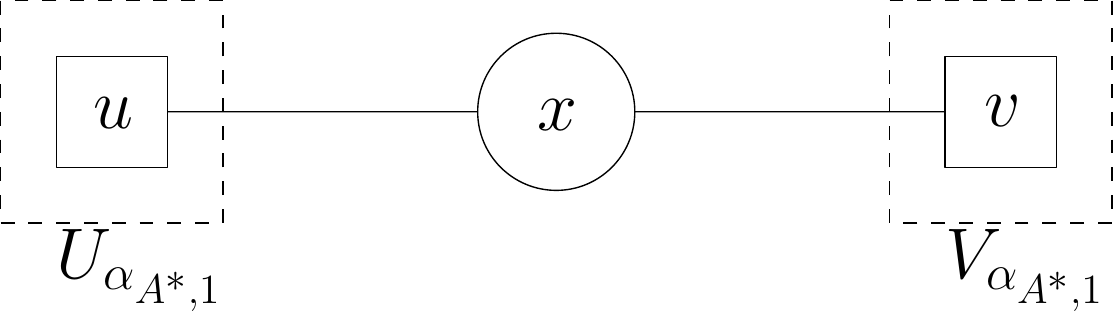}
        \caption{Shape $\alpha_{A^*,1}$}
        \label{fig:alphaAstar1}
    \end{subfigure}
    
    \begin{subfigure}[b]{0.3\textwidth}
        \includegraphics[width=\textwidth]{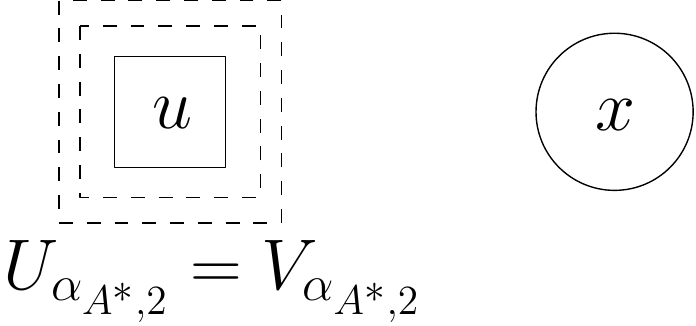}
        \caption{Shape $\alpha_{A^{*},2}$}
        \label{fig:alphaAstar2}
    \end{subfigure}
    
    \begin{subfigure}[b]{0.3\textwidth}
        \includegraphics[width=\textwidth]{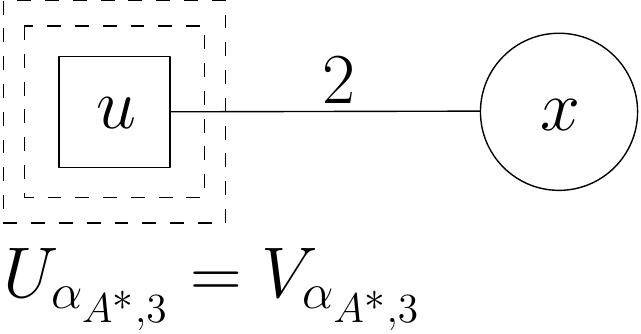}
        \caption{Shape $\alpha_{A^*,3}$}
        \label{fig:alphaAstar3}
    \end{subfigure}
    \caption{Shapes appearing in $\calA^*$.} \label{fig:Astar}
\end{figure}

Let $x \in \R^n$ be a ($n \times 1$) graph matrix represented by the shape $\beta$. The $d \times d$ matrix $\calA^*(x)$ can be represented as a linear combination of graph matrices in the following way. 
\begin{enumerate}
    \item First, ``re-shape'' each of the $\alpha_{A^*,j}$ for $j=1,2,3$ into shapes that represent $d^2 \times n$ matrices by redefining $U_{\alpha_{A^*,j}} \leftarrow U_{\alpha_{A^*,j}} \cup V_{\alpha_{A^*,j}}$ and $V_{\alpha_{A^*,j}} \leftarrow \{x\}$.
    \item Invoke Proposition~\ref{prop:mult-rule} to multiply each of these shapes by $x$.
    \item Reshape the resulting shapes, which represent $d^2 \times 1$ matrices, into shapes representing $d \times d$ matrices by defining the left vertex set to contain only $u$ and the right vertex set to contain only $v$ (for $\alpha_{A^*,1}$) or the left vertex set and right vertex set to contain only $u$ (for $\alpha_{A^*,2}, \alpha_{A^*,3}$) and defining all other vertices to be middle vertices.
\end{enumerate}
See Figures~\ref{fig:ident-pattern} and \ref{fig:ident-pattern-result} for an example of such a multiplication arising in the product $\calA^*(\Delta w)$.

\subsection{Norm bound strategy using graph matrices}
\label{sec:normbound_strat}
We need to bound norms of matrices of the following forms: 
\begin{enumerate}
    \item $\{\mathcal{A}^{*}({\Delta^k}w):  k \in \mathbb{N}\}$ 
    \item $\{\mathcal{A}^{*}({\Delta^k}{1_n}): k \in \mathbb{N}\}$
    \item $\{1_n^T{\Delta^k}w: k \in \mathbb{N}\}$. (Here, we regard the scalars as $1 \times 1$ matrices.)
\end{enumerate}

For a fixed $k$, consider one of the matrices in 1--3 above. In Section \ref{sec:shapes}, we expressed $\calA^*$ and $\Delta$ as a linear combination of shapes. Thus, for fixed $k$, any matrix in 1--3 above can be written as linear combination of terms, where each term is a product of shapes, say 
\begin{align}
\label{eqn:general_product}
    M_{\beta_{0}} \cdot M_{ \beta_1} \cdots M_{\beta_k} \cdot M_{\beta_{k+1}},
\end{align} where each $\beta_i$ is a (proper) shape from Section \ref{sec:shapes}. For each term as in \eqref{eqn:general_product}, we further decompose it into a linear combination of several sub-terms represented by new (proper) shapes using Proposition~\ref{prop:resolve-multi-edges}. Let $\alpha_{P}$ denote a shape that arises as a subterm. We also define an scalar $c_P$ associated with $\alpha_P$ that is used to form the its coefficient in the aforementioned linear combination. The shape $\alpha_{P}$ and scalar $c_P$ are constructed by the following procedure: 


\begin{enumerate}

    \item We start with $\beta_{0}$. If $\beta_{0}$ is $\alpha_{A^{*},1}$, $\alpha_{A^{*},2}$, or $\alpha_{A^{*},3}$, we call the circle vertex $x$ to be the (initial) \textit{loose end}. If $\beta_{0} = \alpha_{1_n^T}$, we call the single  vertex in this shape, which is a circle, the (initial) \textit{loose end}. In all cases, we set $U_{\alpha_P} = U_{\beta_{0}}$ and $V_{\alpha_P} = V_{\beta_{0}} $. Note that $U_{\alpha_{P}} = V_{\alpha_{P}} = \emptyset$ if $\beta_{0} = \alpha_{1_n^T}$.
    

    \item We now do the following for each $j \in [k]$  
    \begin{enumerate}
        \item[(a)] Append the shape $\beta_{j}$ by identifying $U_{\beta_j} = (u)$ with the current loose end and making $V_{\beta_j} = (v)$ the new loose end.
        \item[(b)] For each vertex in $\calV(\beta_j) \setminus U_{\beta_j}$, either leave it alone or identify it with an existing vertex of the same type (circle or square) which is not $u$ and has not yet been identified with a vertex in $\calV(\beta_j) \setminus U_{\beta_j}$.
        \item[(c)] If this creates two parallel edges with integer labels $N$ and $M$, we either (i) remove these parallel edges if $N+M$ is even or (ii) replace those parallel edges with a single labeled edge that has the same parity as $N+M$ and lies in $[N+ M]$. 

        In either case (i) or (ii), assign the edge (or empty edge) a coefficient according to the rule in Proposition~\ref{prop:resolve-multi-edges}. 
    \end{enumerate}

    \item Finally, apply the same procedure as described in 2(a--c) to $\beta_{k+1}$. Concretely, if the last term in the product is $1_n$ (so $\beta_{k+1} = \alpha_{1_n}$), we stop here. If the last term in the product is $w$ (so $\beta_{k+1} = \alpha_{w}$), we append $\alpha_w$ to the existing shape by identifying the current loose end with $U_{\alpha_w} = (u)$. Then, for each set of resulting parallel edges, we remove or replace them as described in 2(c) and assign them coefficients.

    \item Form the scalar $c_P$ by multiplying together all coefficients of the labeled edges (including non-edges) that are output by the conversion procedure in Steps 2(c) and 3.
    
    
    

 
\end{enumerate}

As we described, the matrices $\calA^*(\Delta^k w)$, $\calA^*(\Delta^k 1_n)$, and $1_n^T \Delta^k w$ are linear combinations of terms of the form $\alpha_{P}$. We now describe the coefficients associated to a particular term $\alpha_P$ in this linear combination. To do so, we introduce the following definitions.
\begin{definition}
For each shape $\beta_j$ for $j \in [k]$, we define its \textbf{coefficient} $c(\beta_j)$ to be its coefficient in the graph matrix decomposition of  $\Delta$.
\end{definition}

\begin{definition}
Let shapes $\beta_{0}, \beta_1, \ldots, \beta_{k}, \beta_{k+1}$ be as described above.
\begin{enumerate}
    \item An \textbf{identification pattern} $P$ on $\beta_{0}, \beta_1, \ldots, \beta_{k}, \beta_{k+1}$ specifies which vertices are identified with each other (according to Proposition~\ref{prop:mult-rule}) and which labelled edge is chosen when we convert parallel labeled edges into a single labeled edge (according to Proposition~\ref{prop:resolve-multi-edges}) as in step 2(c) above.
    \item We define $\calP_{\beta_{0}, \beta_1, \ldots, \beta_{k}, \beta_{k+1}}$ to be the set of all identification patterns on the shapes $\beta_{0}, \beta_1, \ldots, \beta_{k}, \beta_{k+1}$.
    \item Given an identification pattern $P$, we define $\alpha_{P}$ to be the shape resulting from $P$ and we define $c_{P}$ to be the coefficient, so that the resulting  term is ${c_P}\cdot \prod_{j=1}^{k}{c(\beta_j)} \cdot M_{\alpha_P}$.
\end{enumerate}
\end{definition}
In other words, $c_P$ captures the part of the coefficient of $M_{\alpha_P}$ which comes from converting parallel labeled edges into a single labeled edge (or non-edge). Note that the (constant) coefficients coming from $\beta_{0}$ and $\beta_{k+1}$ are also absorbed into $c_P$. In our argument it is particularly important to keep track of any non-edges that result from resolving two or more parallel labeled edges into a single labeled edge, which we make precise in the definition below.

\begin{definition}[Vanishing edges]
\label{def:vanish}
Consider an identification pattern $P$ on $\beta_0, \beta_1, \ldots, \beta_{k+1}$. The \textbf{vanishing edges} are the edges with label $0$ (i.e., non-edges) that result from resolving parallel edges according to Proposition \ref{prop:resolve-multi-edges}, as in step 2(c) above.  Moreover, we say that a non-edge in $\alpha_P$ \textbf{vanishes} if it is in the set of vanishing edges.
\end{definition}

Next, we define a method for concisely summarizing certain information about a given identification pattern. Given $j \neq j'$, an identification pattern $P$, and  a vertex $y \in \beta_j$, we say below that $y$ \textit{appears} in $\beta_{j'}$ if $P$ identifies $y$  with a vertex $y' \in \beta_{j'}$.

\begin{definition}
For a given identification pattern $P$, we define a \textbf{decoration} $\tau: \cup_{j = 0}^{k+1} \calV(\beta_j) \rightarrow \{\emptyset, L, R, LR\}$ to summarize information about $P$ in the following way. For each $j$ and each vertex $y \in \beta_j$, define:
\[
\tau(y) = \begin{cases}
    L & \text{if $y$ appears in $\beta_{j'}$ for $j' < j$ and does not appear in any $\beta_{j''}$ for $j'' > j$},\\
    R & \text{if $y$ appears in $\beta_{j'}$ for $j' > j$ and does not appear in any $\beta_{j''}$ for $j'' < j$},\\
    LR & \text{if $y$ appears in $\beta_{j'}$ for $j' < j$ and also appears in some $\beta_{j''}$ for $j'' > j$},\\
    \emptyset & \text{otherwise}.
\end{cases}
\]
In particular, note that:
\begin{itemize}
    \item For any $j \geq 1$ and $y \in U_{\beta_j}$, $y$ automatically appears in $\beta_{j-1}$, so $\tau(y) \in \{L,LR\}$.
    \item For any $j < k+1 $ and $y \in V_{\beta_j}$, $y$ automatically appears in $\beta_{j+1}$, so $\tau(y) \in \{R,LR\}$.
    \item For any $y \in \calV(\beta_0)$, $\tau(y) \in \{\emptyset, R\}$.
    \item For any $y \in \calV(\beta_{k+1})$, $\tau(y) \in \{\emptyset, L\}$.
\end{itemize}
\end{definition}

\begin{figure}
    \includegraphics[width=\textwidth]{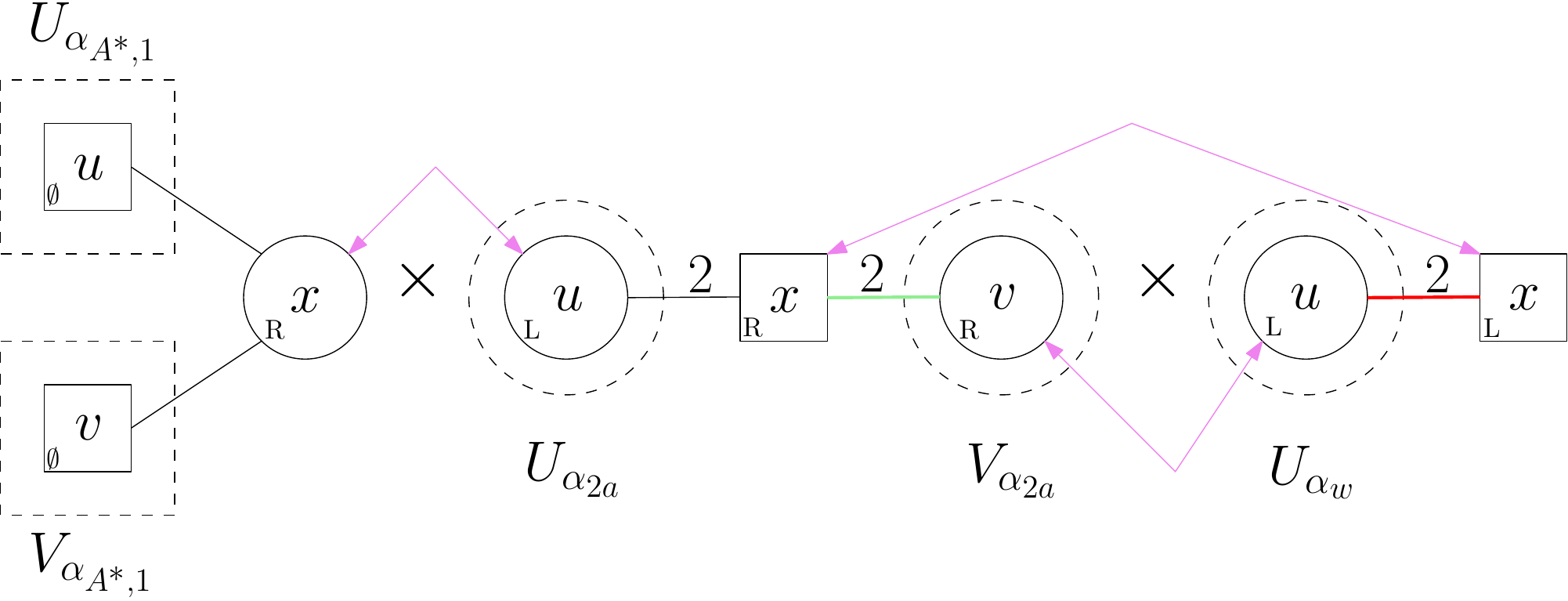}
    \caption{Example of an identification pattern for multiplying $M_{\alpha_{A^*,1}}$, $M_{\alpha_{2a}}$, and $M_{\alpha_w}$. The violet arrows denote which vertices are identified with each other. As in Figure~\ref{fig:mult-example}, if these vertices are identified with each other, the red and green edges become a multi-edge which gets converted to a linear combination of labeled edges using Proposition~\ref{prop:resolve-multi-edges}. The identification pattern that is consistent with the violet vertex identifications and which also picks edge label 0 to replace the multi-edge results in a shape that is depicted in Figure~\ref{fig:ident-pattern-result}. The decorations of each vertex are written in the bottom left corner of each square or circle.
    The red edge is a right-critical edge and the green edge is a left-critical edge.}
    \label{fig:ident-pattern}
\end{figure}

\begin{figure}
    \centering\includegraphics[width=0.7\textwidth]{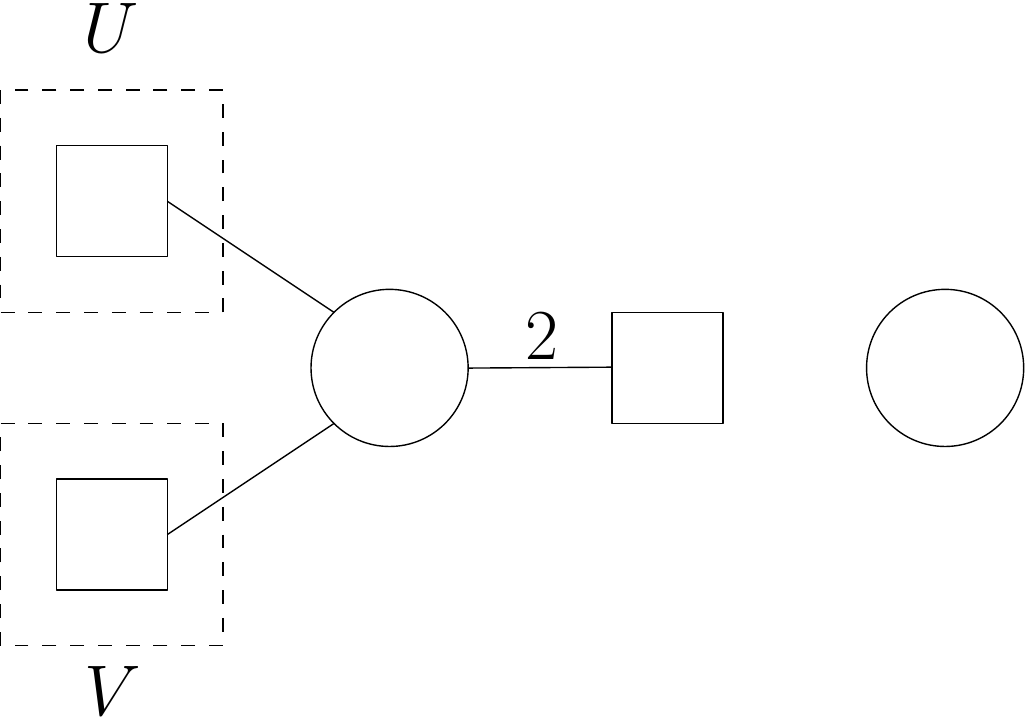}
    \caption{Resulting shape from the identification pattern in Figure~\ref{fig:ident-pattern}. Observe that the minimum weight vertex separator can be taken to be one of the square vertices in $U$ or $V$.}
    \label{fig:ident-pattern-result}
\end{figure}

See Figures~\ref{fig:ident-pattern} and \ref{fig:ident-pattern-result} for an example of an identification pattern and the associated decoration that can arise when multiplying shapes from the product $\calA^* (\Delta w)$. With these definitions and for a fixed 
$k$, each of $\calA^*( \Delta^k 1_n)$, $\calA^*( \Delta^k w)$, and $1_n^T(\Delta^k) w$ can each be expressed as a summation of the following form:
\begin{align} 
\label{eqn:pattern_decomposition} 
\sum_{\beta_1,\ldots,\beta_k \in \{\alpha_i: i \in \calI\}} \, \, {\sum_{P \in \calP_{\beta_{0},\beta_1,\ldots,\beta_k,\beta_{k+1}}}{\bigg( {c_P}\cdot \prod_{j=1}^{k}{c(\beta_j)} \bigg) M_{\alpha_P}}}.
\end{align} 
To bound the norm of this expression, we apply the triangle inequality and bound separately $\norm{M_{\alpha_P}}_{op}$ for each identification pattern $P$.
\subsubsection{Upper bounding $\norm{M_{\alpha_P}}_{op}$ via weights}
\label{sec:norm-bounds-via-weights}
In order to upper bound $\norm{M_{\alpha_P}}_{op}$, we will consider the contribution to $\norm{M_{\alpha_P}}_{op}$ from each of the shapes $\beta_{0},\beta_1,\ldots,\beta_k,\beta_{k+1}$. In order to invoke the bound in Theorem~\ref{thm:graph-matrix-norm-bound}, we must be able to calculate the min-vertex separator of $\alpha_P$ and status (i.e.\ square vs.\ circle and isolated vs.\ non-isolated) of each of the vertices of $\alpha_P$. Specifically, define the following short-hand notation for the dominant term in the norm bound of Theorem~\ref{thm:graph-matrix-norm-bound}:
\[
\calB(\alpha_P) = n^{\frac{\varphi(V(\alpha_P)) - \varphi(S_{\text{min}}) + \varphi(\mathrm{Iso}(\alpha_P))}{2}}
\]
where $S_{\text{min}}$ is a min-vertex separator of $\alpha_P$.

To compute the combinatorial quantities that define $\calB(P)$, we will design an ideal weight function $w_{\text{ideal},P}$ and an actual weight function $w_{\text{actual},P}$, each of which assigns a \textit{weight} to each of the shapes 
$\beta_{0},\beta_1,\ldots,\beta_k,\beta_{k+1}$. Intuitively, the ideal weight function allows us to accurately estimate the right-hand side of the norm bound in Theorem~\ref{thm:graph-matrix-norm-bound}. However, as we explain later, determining this ideal weight function exactly is intractable. Instead, we show that the actual weight function is a faithful ``relaxation'' of the ideal weight function, tractable to calculate and still leads to sufficiently good norm bounds. 

More precisely, the following properties must be satisfied: 
\begin{enumerate}[label=\textbf{P.\arabic*}]
    \item \label{property:ideal_vs_norm_bound} $\calB(\alpha_P) = \prod_{j=0}^{k+1}{w_{\text{ideal},P}(\beta_{j})}$ where $\beta_{0} = \beta_{0}$ and $\beta_{k+1} = \beta_{k+1}$. This ensures that the product of the ideal weights over the shapes $\beta_j$ faithfully estimates the dominant term of the norm bound in Theorem~\ref{thm:graph-matrix-norm-bound}.
    \item \label{property:actual_vs_ideal} $\prod_{j=0}^{k+1}{w_{\text{actual},P}(\beta_{j})} \geq \prod_{j=0}^{k+1}{w_{\text{ideal},P}(\beta_{j})}$. In fact, for almost all shapes $\beta_j$ we will have that \\ $w_{\text{actual},P}(\beta_{j}) \geq w_{\text{ideal},P}(\beta_{j})$. This ensures that the actual weight function gives a norm bound that is valid (i.e.\ no smaller than the ``true'' norm bound given by the ideal weight function).
    \item \label{property:actual_bound_small} For all $j \in [k]$, $|w_{\text{actual},P}(\beta_j)c(\beta_j)| \leq d^{\frac{3}{2}} \sqrt[4]{n}$. This ensures that the norm bound given by the actual weight function is sufficiently small to complete the proofs of our technical lemmas.
\end{enumerate}
We note that the number of possibilities for $\beta_1,\ldots,\beta_k$, the number of possible identification patterns on $\beta_{0},\beta_1,\ldots,\beta_k,\beta_{k+1}$, the maximum coefficient $c_P$ for any identification pattern $P$, and the ratio $\frac{\norm{M_{\alpha_P}}_{op}}{\calB(\alpha_P)}$ (assuming the probabilistic norm bound in Theorem~\ref{thm:graph-matrix-norm-bound} holds) are all at most $(\log n)^{O(k)}$. This follows from the observation that we have $k+2$ shapes $\beta_{0},\beta_1,\ldots,\beta_k,\beta_{k+1}$, each consisting of $O(1)$ vertices and edges, and a simple combinatorial fact which we defer to the appendix (Proposition~\ref{prop:number-of-shapes-bound}).

Thus, if we can specify weight functions satisfying the above properies, then we may bound the expression in Equation~\ref{eqn:pattern_decomposition} as follows:
\begin{align*}
    &\sum_{\beta_1,\ldots,\beta_k \in \{\alpha_i : i \in \calI\}}{\sum_{P \in \calP_{\beta_{0},\beta_1,\ldots,\beta_k,\beta_{k+1}}}{{|c_P|}\left(\prod_{j=1}^{k}{|c(\beta_j)|}\right)\norm{M_{\alpha_P}}_{op}}} \\
    &=\sum_{\beta_1,\ldots,\beta_k}{\sum_{P \in \calP_{\beta_{0},\beta_1,\ldots,\beta_k,\beta_{k+1}}}{{|c_P|}\left(\prod_{j=1}^{k}{|c(\beta_j)|}\right)\left(\prod_{j=0}^{k+1}{w_{\text{ideal},P}(\beta_j)}\right)\frac{\norm{M_{\alpha_P}}_{op}}{\calB(\alpha_P)}}} \\
    &\leq (\log n)^{O(k)} \sum_{\beta_1,\ldots,\beta_k} \sum_{P \in \calP_{\beta_{0},\beta_1,\ldots,\beta_k,\beta_{k+1}}} |c_P| \left(\prod_{j=1}^k |c(\beta_j)| \right) \left( \prod_{j=0}^{k+1}{w_{\text{ideal},P}(\beta_j)} \right)\\
    &\leq (\log n)^{O(k)} \sum_{\beta_1,\ldots,\beta_k} \sum_{P \in \calP_{\beta_{0},\beta_1,\ldots,\beta_k,\beta_{k+1}}} |c_P| \left(\prod_{j=1}^k |c(\beta_j)| \right) \left( \prod_{j=0}^{k+1}{w_{\text{actual},P}(\beta_j)} \right)\\
    &\leq (\log n)^{O(k)} \left(d^{\frac{3}{2}}\sqrt[4]{n}\right)^{k}\max_P{\{w_{\text{actual},P}(\beta_{0})w_{\text{actual},P}(\beta_{k+1})\}} \num \label{eqn:w_actual_bound}.
\end{align*}
\subsubsection{Min-vertex separator of $\alpha_P$}
Before specifying the weight functions, we recall that: 
\[
\calB(\alpha_P) = d^{\frac{|\calV_{\Box}(\alpha_P)| + |\mathrm{Iso}(\alpha_P) \cap \calV_{\Box}(\alpha_P)| - |S_{\text{min}} \cap \calV_{\Box}(\alpha_P)| }{2}} \times n^{\frac{|\calV_{\circ}(\alpha_P)| + |\mathrm{Iso}(\alpha_P) \cap \calV_{\circ}(\alpha_P)| - |S_{\text{min}} \cap \calV_{\circ}(\alpha_P)| }{2}}
\]
where $S_{\text{min}}$ is a min-vertex separator of $\alpha_P$. We now determine the min-vertex separator for $\alpha_P$ so that we can apply the bound in Theorem~\ref{thm:graph-matrix-norm-bound}. When $\beta_{0} = \alpha_{1_n^T}$, $U_{\alpha_P} = V_{\alpha_P} = \emptyset$ so $S_{\text{min}} = \emptyset$. As we now show, when $\beta_{0}$ is $\alpha_{A^{*},1}$, $\alpha_{A^{*},2}$, or $\alpha_{A^{*},3}$, the min-vertex separator of $\alpha_P$ consists of a single square. See Figure~\ref{fig:ident-pattern-result} for an example.
\begin{lemma}
\label{lem:min-vertex-sep}
If $\beta_{0}$ is $\alpha_{A^{*},1}$, $\alpha_{A^{*},2}$, or $\alpha_{A^{*},3}$ then the min-vertex separator of $\alpha_P$ consists of a single square.
\end{lemma}
\begin{proof}
Since $U_{\alpha_P}$ consists of a single square and is a vertex separator, the minimum weight vertex separator is either a single square or no vertices at all. To show that the minimum weight vertex separator has at least one vertex, we prove that $U_{\alpha_P}$ must be connected to $V_{\alpha_P}$.

To prove this, it is sufficient to prove the following lemma. Here by `degree' of a vertex $a$ in a shape, we mean the sum of all edge labels of edges incident to $a$.

\begin{lemma}\label{lem:connectivitylemma}
If $\beta_{0}$ is $\alpha_{A^{*},1}$, $\alpha_{A^{*},2}$, or $\alpha_{A^{*},3}$ then either $U_{\alpha_P} = V_{\alpha_P} = (u)$ where $u$ is a square vertex, or $U_{\alpha_P} = (u)$, $V_{\alpha_P} = (v)$ where $u$ and $v$ are distinct square vertices and $u$ and $v$ are the only vertices with odd degree.
\end{lemma}
With this lemma, the result follows easily. If $U_{\alpha_P} = V_{\alpha_P} = (u)$ then the result is trivial. If $U_{\alpha_P} = (u)$ and $V_{\alpha_P} = (v)$ where $u$ and $v$ are distinct square vertices and are the only vertices with odd degree then $u$ and $v$ must be in the same connected component of $\alpha_P$ due to the following fact.

\begin{proposition}
\label{prop:handshake}
For any undirected graph $G$ with integer edge-labels, for any connected component $C$ of $G$, $\sum_{v \in C}{deg(v)}$ is even.
\end{proposition}
\begin{proof}
This is the handshaking lemma and can be proved by observing that $\sum_{v \in C}{deg(v)} = 2|E(C)|$ which is even.
\end{proof}

To be concrete, suppose that for the sake of contradiction $u$ and $v$ lie in distinct connected components $C_u$ and $C_v$ of $\alpha_P$. Then the sum of degrees in $C_u$ is odd by Lemma \ref{lem:connectivitylemma}, which contradicts Proposition \ref{prop:handshake}.

We proceed to prove Lemma \ref{lem:connectivitylemma}.
\begin{proof}[Proof of Lemma \ref{lem:connectivitylemma}]
Let $U_{\alpha_P} = (u)$ and $V_{\alpha_P} = (v)$. We make the following observations about the process for building the shape $\alpha_{P}$
\begin{enumerate}
    \item In $\beta_{0}$, either $u = v$ (in which case $u$ has even degree) or $u$ and $v$ are distinct square vertices which have odd degree. The circle vertex in $\beta_{0}$ always has even degree.
    
    \item For all of the shapes $\beta_1, \ldots, \beta_{k}, \beta_{k+1}$, all of the vertices have even degree.

    \item Whenever two vertices are identified, the parity of the degree of the resulting vertex is equal to the parity of the sum of the degrees of the original vertices.
    
    \item No other operation affects the parities of the degrees of the vertices.

\end{enumerate}
Together, these observations imply that either $u = v$ or $u$ and $v$ are the only vertices with odd degree. 
\end{proof}
\noindent The proof of Lemma~\ref{lem:min-vertex-sep} is now complete.
\end{proof}

\begin{corollary}\label{cor:alphaPapproximatenorm}
If $\beta_{0}$ is $\alpha_{A^{*},1}$, $\alpha_{A^{*},2}$, or $\alpha_{A^{*},3}$ then 
\[
\calB(\alpha_P) = d^{\frac{|\calV_{\Box}(\alpha_P)| + |\mathrm{Iso}(\alpha_P) \cap \calV_{\Box}(\alpha_P)| - 1}{2}} \cdot n^{\frac{|\calV_{\circ}(\alpha_P)| + |\mathrm{Iso}(\alpha_P) \cap \calV_{\circ}(\alpha_P)|}{2}}.
\]
If $\beta_{0}$ is $\alpha_{1_n^T}$ then 
\[
\calB(\alpha_P) = d^{\frac{|V_{\Box}(\alpha_P)| + |\mathrm{Iso}(\alpha_P) \cap \calV_{\Box}(\alpha_P)|}{2}} \cdot n^{\frac{|V_{\circ}(\alpha_P)| + |\mathrm{Iso}(\alpha_P) \cap \calV_{\circ}(\alpha_P)|}{2}}.
\]
\end{corollary}
We remark that the factor of $\frac{1}{\sqrt{d}}$ appearing in the first expression turns out to be crucial to obtaining satisfactory norm bounds.

\subsubsection{The ideal weight function}
Given that we have identified the min-vertex separator of the shape $\alpha_P$, we now specify an ideal weight function $w_{\text{ideal},P}$ such that $\calB(\alpha_P) = \prod_{j=0}^{k+1}{w_{\text{ideal},P}(\beta_{j})}$. First we introduce and formalize some intuitive terminology. We say that $v \in \calV(\alpha_P)$ \textit{appears} in a shape $\beta_j$ (or that $\beta_j$ \textit{contains} $v$) if $v$ is the result of identifying one or more vertices according to the identification pattern $P$, at least one of which lies in $\beta_j$. We order the indices $j_1 < \ldots < j_r$ (which we refer to as discrete \textit{times})  of the shapes in $\beta_{j_1} , \ldots, \beta_{j_r}$ where $v$ appears. We refer to $j_1$ as the \textit{first time} $v$ appears and $j_r$ as the \textit{last time} $v$ appears. 

Each vertex $v \in \calV(\alpha_P)$ has an associated  value 
coming from the expression in Corollary~\ref{cor:alphaPapproximatenorm}. This value, which we call $w_{\text{ideal},P}(v)$, is as follows:
\begin{equation*}
w_{\text{ideal},P}(v) = \begin{cases}
      \sqrt{d} & \text{if $v$ is a square and not isolated}, \\
      d & \text{if $v$ is a square and isolated\footnotemark}, \\
      \sqrt{n} & \text{if $v$ is a circle and not isolated}, \\
      n & \text{if $v$ is a circle and isolated}. 
    \end{cases}
    \footnotetext{Note that vertices in $U_{\alpha_P} \cup V_{\alpha_P}$ do not count as isolated.}
\end{equation*}

We ``split'' this value among (at most) 2 shapes that contain $v$ by assigning the square root of the value of $v$ to each of $\beta_j$ and  $\beta_{j'}$, where $j$ is the smallest index for which $v$ appears in $\beta_j$ and $j'$ is the largest index for which $v$ appears in $\beta_{j'}$. If $v$ only appears once, then we assign the full value of $v$ to the shape in which it appears. Formally, we define this procedure in the following way. 
For $j \in [k+1]$, we let $w_{\text{ideal},P}(\beta_j)$ be the total weight which is assigned to $\beta_j$. This weight is 
\[
w_{\text{ideal},P}(\beta_j) = \prod_{v \in \calV(\beta_j)}{w_{\text{ideal},P}(v)^{1 - \frac{1}{2} \I \{\tau(v) \in \{L, LR\}\} - \frac{1}{2} \I \{\tau(v) \in \{R, LR\}\}}}.
\]
For $\beta_0$, we adjust this to take the minimum weight vertex separator into account. In particular, if $\beta_0 = \alpha_{1_n^T}$ then:
\[
w_{\text{ideal},P}(\beta_0) = \prod_{v \in \calV(\beta_0)}{w_{\text{ideal},P}(v)^{1  - \frac{1}{2} \I \{\tau(v) = R\}}},
\]
while if $\beta_0$ is $\alpha_{A^{*},1}$, $\alpha_{A^{*},2}$, or $\alpha_{A^{*},3}$ then 
\[
w_{\text{ideal},P}(\beta_0) = \frac{1}{\sqrt{d}}\prod_{v \in \calV(\beta_0)}{w_{\text{ideal},P}(v)^{1  - \frac{1}{2}\I\{\tau(v) = R\}}}.
\]
From these definitions, we may immediately conclude that $\calB(\alpha_P) = \prod_{j=0}^{k+1}{w_{\text{ideal},P}(\beta_{j})}$ (i.e.\ Property~\ref{property:ideal_vs_norm_bound} is satisfied). See the first row of Table~\ref{table:weights-example} for an example of how $w_{\text{ideal}}$ is computed for a particular shape and identification pattern arising in $\calA^*(\Delta w)$.

\subsubsection{The local weight function}
\label{sec:local_weights} 
While the ideal weight function yields the correct norm bound, it cannot be computed separately for each shape $\beta_j$ because in order to determine if a vertex in $\beta_j$ is isolated or not, we need to consider the entire identification pattern $P$. To handle this, we introduce a different weight function $w_{\text{local},P}$ which can be computed separately for each shape $\beta_j$ by considering only the ``local data'' consisting of the decorations on vertices in $\calV(\beta_j)$. To define  $w_{\text{local},P}(v)$, for each $j$ and each $v \in \calV(\beta_j)$, we upper bound $w_{\text{ideal},P}(v)$ based on the local data of $\beta_j$. In particular, if a vertex is incident to an edge which cannot vanish based on the local data at $\beta_j$, then we know it cannot be isolated and $w_{\text{ideal},P}(v)$ is $\sqrt{d}$ or $\sqrt{n}$. For $j = 0$, we also know $w_{\text{ideal},P}(v) = \sqrt{d}$ for the vertices $v \in U_{\alpha_P} \cup V_{\alpha_P}$ since we never consider vertices in $U_{\alpha_P} \cup V_{\alpha_P}$ to be isolated.
For other vertices, we conservatively upper bound $w_{\text{ideal},P}(v)$ by $d$ or $n$. 
We introduce the following definitions in order to formally define $w_{\text{local},P}$.
\begin{definition}
We say that an edge $e = \{u,v\}_l$ is \textbf{safe} for shape $\beta_j$ if both of the following hold:
\begin{enumerate}
    \item Either $\tau(u) \in \{\emptyset, R\}$ or $\tau(v) \in \{\emptyset, R\}$.
    \item Either $\tau(u) \in \{\emptyset, L\}$ or $\tau(v) \in \{\emptyset, L\}$.
\end{enumerate}
\end{definition}
Observe that a safe edge cannot vanish (i.e.\ it appears in $\alpha_P$ with a positive edge label). For every $j$ and $v \in V(\beta_j)$, we define the ``full'' local weight as:
\[
b_j(v) = \begin{cases}
      \sqrt{d} & \text{if $v$ is a square and incident to a safe edge for $\beta_j$}, \\
      \sqrt{d} & \text{$j=0$ and $v \in U_{\alpha_P} \cup V_{\alpha_P}$}, \\
      d & \text{if $v$ is any other square}, \\
      \sqrt{n} & \text{if $v$ is a circle and incident to a safe edge for $\beta_j$}, \\
      n & \text{if $v$ is any other circle}.
    \end{cases}
\]

As before, if a vertex is identified with other vertices, then we ``split'' the full weight between the first and last times it appears. So we define, for every $j$ and $v \in V(\beta_j)$, the ``split'' local weight as: 
\begin{equation}
    \label{eqn:bbar_def}
    \overline{b}_j(v) = b_j(v)^{1- \frac{1}{2} \I \{\tau(v) \in \{L,LR\}) - \frac{1}{2} \I(\tau(v) \in \{R,LR\}\}}.  
\end{equation}
Recall that for $1 \leq j \leq k$ ,  $w_{\text{local},P}(\beta_j)$ is the product of the $\overline{b}_j(y)$'s, see \eqref{eqn:local_weight}. Using these definitions, we define the weight function $w_{\text{local},P}$ as: 
\begin{align} 
\label{eqn:local_weight}
w_{\text{local},P}(\beta_j) = \left(\frac{1}{\sqrt{d}} \right)^{\mathcal{I}(\beta_j)} \cdot 
\prod_{v \in \calV(\beta_j)} \overline{b}_j(v),
\end{align} 
where 
\[
\mathcal{I}(\beta_j) = {\I\bigg\{j = 0 \text{ and } \, \beta_0 \in \{ \alpha_{A^*,1},  \alpha_{A^*,2},\alpha_{A^*,3} \}\bigg\}}. 
\]
In other words, we divide by $\sqrt{d}$ for $j = 0$ if $\beta_0$ is $\alpha_{A^{*},1}$, $\alpha_{A^{*},2}$, or $\alpha_{A^{*},3}$. We may immediately conclude from these definitions that property~\ref{property:actual_vs_ideal} is satisfied; in the next section we verify Property~\ref{property:actual_bound_small}. We record for future use a simple upper bound on the weights of vertices.
\begin{proposition}
\label{prop:trivial-weight-bound}
Let $j \in \{0, \ldots, k+1\}$. If $v \in \calV_\circ (\beta_j)$, then the contribution of $v$ to $w_{\text{local},P}(\beta_j)$ is at most $\sqrt{n}$. If $v \in \calV_\Box (\beta_j)$, then the contribution of $v$ to $w_{\text{local},P}(\beta_j)$ is at most $\sqrt{d}$. The same results also apply to $w_{\text{ideal},P}$.  
\end{proposition}
\begin{proof}
Consider the case $v \in \calV_\circ (\beta_j)$. If $v$ is not identified with any other vertex, then it cannot be isolated. Hence, it contributes no more than $\sqrt{n}$ to $w_{\text{local},P}(\beta_j)$. On the other hand, if $v$ is identified with some other vertex, then the maximum possible weight of $n$ is split between $v$ and some other vertex. Thus, it contributes no more than $\sqrt{n}$ to $w_{\text{local},P}(\beta_j)$.
\end{proof}

See the second row of Table~\ref{table:weights-example} for an example of how $w_{\text{local}}$ is computed for a particular shape and identification pattern arising in $\calA^*(\Delta w)$.

\subsection{Proof of Lemma~\ref{lemma:Astar-Delta-k-ones}}
In this section, we show that the local weight function $w_{\text{local},P}$ from the previous section is sufficient to complete the proof of Lemma~\ref{lemma:Astar-Delta-k-ones}. As mentioned earlier, by definition of $w_{\text{local},P}$, it satisfies Property~\ref{property:actual_vs_ideal}. It remains to verify the following:
\begin{enumerate}
    \item If $\beta_{0}$ is any of $\alpha_{A^*,1}, \alpha_{A^*,2}, \alpha_{A^*,3}$ and $\beta_{k+1} = \alpha_{1_n}$, then $w_{\text{local},P}(\beta_{0}) w_{\text{local},P}(\beta_{k+1}) \leq \max(\sqrt{d}n^{3/4},n)$.
    
    \item For all $j \in [k]$, $|w_{\text{local},P}(\beta_j)c(\beta_j)| \leq d^{\frac{3}{2}} \sqrt[4]{n}$ (corresponding to Property~\ref{property:actual_bound_small}).
\end{enumerate}
Given these and Equation~\eqref{eqn:w_actual_bound}, we have the following bound with probability $1 - n^{-\Omega(1)}$:
\begin{align*}
\norm{\calA^*(\Delta^k 1_n)}_{op}   &\leq (\log n)^{O(k)} \cdot \left(d^{\frac{3}{2}} \sqrt[4]{n}\right)^{k} \cdot \max_{P \in \calP_{\beta_{0},\beta_1,\ldots,\beta_k,\beta_{k+1}}}{|c_P| \cdot  w_{\text{local},P}(\beta_{0}) \cdot w_{\text{local},P}(\beta_{k+1}) }\\
&\leq (\log n)^{O(k)} \cdot \left(d^{\frac{3}{2}} \sqrt[4]{n}\right)^{k} \cdot \max(\sqrt{d}n^{3/4},n),
\end{align*}
which completes the proof of Lemma~\ref{lemma:Astar-Delta-k-ones}.

We now verify the two conditions on the local weight function. For the first condition, note that $w_{\text{local},P}(\beta_{k+1}) = w_{\text{local},P}(\alpha_{1_n}) = \sqrt{n}$. To handle the contribution from $\beta_{0}$, we enumerate the following cases:
\begin{itemize}
    \item \textbf{Case $\beta_{0} = \alpha_{A^*,1}$:} We know that the two squares $u, v$ are not in $\text{Iso}(\alpha_P)$, so $b_0(u) = b_0(v) = \sqrt{d}$. Suppose at least one of the two edges in $\alpha_{A^*,1}$ are safe. Then, $b_0(x) = \sqrt{n}$ and $w_{\text{local},P} (\alpha_{A^*,1}) \leq \frac{1}{\sqrt{d}} b_0(u) b_0(v) b_0(x)^{1/2} = \sqrt{d} n^{1/4}$. Otherwise, $\tau(u) = \tau(v) = R$, so 
    \[
    w_{\text{local},P} (\alpha_{A^*,1}) = \frac{1}{\sqrt{d}} b_0(u)^{1/2} b_0(v)^{1/2} b_0(x)^{1/2} \leq \sqrt{n}.
    \]
    \item \textbf{Case $\beta_{0} = \alpha_{A^*,2}$ or $\alpha_{A^*,3}$:}  Again, we know $b_0(u) = \sqrt{d}$. So, $w_{\text{local},P} (\alpha_{A^*,2}) = \frac{1}{\sqrt{d}} b_0(u) b_0(x)^{1/2} \leq \sqrt{n}$.
\end{itemize}
This completes the proof that $w_{\text{local},P}(\beta_{0}) w_{\text{local},P}(\beta_{k+1}) \leq \max(\sqrt{d}n^{3/4},n)$.

For the second condition, we fix $j \in [k]$. Below, for each shape $\alpha_i$ for $i \in \mathcal{I}$, we consider all possibilities of the decorations of the vertices of $\alpha_i$, reducing the number of cases when possible by symmetry. The tables below handle the essential cases. In the tables below, we use the term `Any' to denote that a vertex may have any of the decorations described above. We now analyze the possible cases for $\beta_j$, repeatedly making use of Proposition~\ref{prop:trivial-weight-bound} when appropriate:
\begin{itemize}
    \item \textbf{Case $\beta_j = \alpha_1$:} A priori, there are a total of $2\cdot 4 \cdot 4 \cdot 2 = 64$ cases since $\tau(u) \in \{ L, LR \}, \tau(x_1), \tau(x_2) \in \{ \emptyset, L, R, LR\}$, and $\tau(v) \in \{ R, LR \}$. By symmetry, each case reduces to one considered in Table \ref{table:alpha1-casework}. Note that the first row of Table \ref{table:alpha1-casework} stands for $32$ different cases. For this row, we slightly abuse notation and use $\bar{b}_j(y)$ to specify an upper bound on the split local weight of $y \in \{u, x_1, x_2, v\} \subset V(\alpha_{1})$ for all of these $32$ cases. 
    
    By inspection of Table \ref{table:alpha1-casework} and using that $n = d^2/\polylog(d)$ (see Remark \ref{rmk:n=d2-polylog}), we conclude that 
    \begin{align} 
    \label{eqn:alpha2a_local}
    w_{\text{local},P}(\beta_j) = w_{\text{local},P}(\alpha_{1}) 
    \leq d \sqrt{n}.  
    \end{align}
    Also note that $c(\beta_j) = c(\alpha_{1}) = O(1)$.
    
    \item \textbf{Case $\beta_j = \alpha_{2a}$:} A priori, there are a total of $2 \cdot 4 \cdot 2 = 16$ cases since $\tau(u) \in \{ L, LR \}, \tau(x) \in \{ \emptyset, L, R, LR\}$, and $\tau(v) \in \{ R, LR \}$. By symmetry, each case reduces to one considered in Table \ref{table:alpha2a-casework}. Note that the first row of Table \ref{table:alpha2a-casework} stands for $4$ different cases. For this row, we slightly abuse notation and use $\bar{b}_j(y)$ to specify an upper bound on the split local weight of $y \in \{u, x, v\} \subset  V(\alpha_{2a})$ for all of these $4$ cases. 
    
    By inspection of Table \ref{table:alpha2a-casework} and using that $n = d^2/\polylog(d)$ (see Remark \ref{rmk:n=d2-polylog}), we conclude that 
    \begin{align} 
    \label{eqn:alpha2a_local}
    w_{\text{local},P}(\beta_j) = w_{\text{local},P}(\alpha_{2a}) 
    \leq n.  
    \end{align}
    Also note that $c(\beta_j) = c(\alpha_{2a}) = O(1)$.
    
    \item \textbf{Case $\beta_j = \alpha_{3a}$:} First, note that $w_{\text{local},P}(\beta_j) = b_j(u)^0 \Bar{b}_j(x_1) \Bar{b}_j(x_2) = \Bar{b}_j(x_1) \Bar{b}_j(x_2)$ because $u$ is identified with vertices to the left and right of $\beta_j$. It is straightforward to enumerate the possible decorations of $x_1,x_2$ and verify that $\Bar{b}_j(x_1) \Bar{b}_j(x_2) \leq d$. Thus, $w_{\text{local},P}(\alpha_{3a}) \leq d$. Also, recall that $|c(\beta_j)| = |c(\alpha_{3a})| = O(1)$.
    
    \item \textbf{Case $\beta_j = \alpha_{3b}$:} Recall that $\alpha_{3b}$ has $U_{\alpha_{3b}} = V_{\alpha_{3b}} = (u)$. By the rules for graph matrix multiplication (see Section \ref{sec:product_rules}), $u \in \alpha_{3b}$ is identified with a circle vertex in $\beta_{j-1}$ and $\beta_{j+1}$. Hence $\tau(u) = LR$. By \eqref{eqn:bbar_def}, we have $\overline{b}_j(u) = 1$. Moreover, $\overline{b}_j(x) \leq \sqrt{d}$. Therefore,
    \begin{align}
    \label{eqn:alpha3b_local}
    w_{\text{local},P}(\beta_j)
    = w_{\text{local},P}(\alpha_{3b}) \leq \sqrt{d}. 
    \end{align}
    Also recall that $|c(\beta_j)| = |c(\alpha_{3b})| = O(d)$. 

    \item \textbf{Case $\beta_j = \alpha_4$:} The proof is very similar to the one for $\alpha_{3b}$. Since $U_{\alpha_{4}} = V_{\alpha_4} = (u)$, we see that $\tau(u) = LR$, so $\overline{b}_j(u) = 1$. And automatically, $\overline{b}_j(x) \leq \sqrt{d}$. Hence
    \begin{align}
    \label{eqn:alpha3b_local}
    w_{\text{local},P}(\beta_j)
    = w_{\text{local},P}(\alpha_{4}) \leq \sqrt{d}. 
    \end{align}
    Also recall that $|c(\beta_j)| = |c(\alpha_{4}) |= O(1)$.
\end{itemize}

\begin{table}[h!]
\centering
\begin{tabular}{ c|c|c|c|c|c|c|c|c } 

 $\tau(u)$ & $\tau(x_1)$ & $\tau(x_2)$ & $\tau(v)$ & $\overline{b}_j(u)$ & $\overline{b}_j(x_1)$ & $\overline{b}_j(x_2)$ & $\overline{b}_j(v)$ & Product \\ 
 \hline
 LR & Any & Any & Any & $1$ & $d^{1/2}$ & $d^{1/2}$ & $n^{1/2}$ & $n^{1/2} d$ \\
 L & $\emptyset$ & $\emptyset$ & R & $n^{1/4}$ & $d^{1/2}$ & $d^{1/2}$ & $n^{1/4}$ & $n^{1/2}d$ \\
 L & $\emptyset$ & L & R & $n^{1/4}$ & $d^{1/2}$ & $d^{1/4}$ & $n^{1/4}$ & $n^{1/2} d^{3/4}$ \\
 L & $\emptyset$ & LR & R & $n^{1/4}$ & $d^{1/2}$ & $1$ & $n^{1/4}$ & $n^{1/2} d^{1/2}$ \\
 L & L & L & R & $n^{1/2}$ & $d^{1/4}$ & $d^{1/4}$ & $n^{1/4}$ & $n^{3/4} d^{1/2}$ \\
 L & L & R & R & $n^{1/4}$ & $d^{1/4}$ & $d^{1/4}$ & $n^{1/4}$ & $n^{1/2} d^{1/2}$ \\
 L & L & LR & R & $n^{1/2}$ & $d^{1/4}$ & $1$ & $n^{1/4}$ & $n^{3/4} d^{1/4}$ \\
 L & LR & LR & R & $n^{1/2}$ & $1$ & $1$ & $n^{1/2}$ & $n$ \\
 
\end{tabular}
\caption{Case work for $\alpha_1$.}
\label{table:alpha1-casework}
\end{table}

\begin{table}[h!]
\centering
\begin{tabular}{ c|c|c|c|c|c|c } 
 $\tau(u)$ & $\tau(x)$ & $\tau(v)$ & $\overline{b}_j(u)$ & $\overline{b}_j(x)$ & $\overline{b}_j(v)$ & Product \\ 
\hline
 LR & Any & LR & $1$ & $\sqrt{d}$ & $1$ & $\sqrt{d}$\\
 LR & R & R & $1$ & $d^{1/2}$ & $n^{1/2}$ & $n^{1/2} d^{1/2}$ \\ 
 LR &  $\emptyset$ & R & $1$ & $d^{1/2}$ & $n^{1/4}$ & $n^{1/4} d^{1/2}$ \\ 
 LR & L  & R & $1$ & $d^{1/4}$ & $n^{1/4}$ & $n^{1/4} d^{1/4}$ \\ 
 L & $\emptyset$ & R & $n^{1/4}$ & $d^{1/2}$ & $n^{1/4}$ & $n^{1/2} d^{1/2}$ \\ 
 L & L & R & $n^{1/2}$ & $d^{1/4}$ & $n^{1/4}$ & $n^{3/4} d^{1/4}$ \\
 L & LR & R & $n^{1/2}$ & $1$ & $n^{1/2}$ & $n$ \\
\end{tabular}
\caption{Case work for $\alpha_{2a}$.}
\label{table:alpha2a-casework}
\end{table}

\subsection{Modifying the local weighting scheme}
Unfortunately, while the local weight function is sufficient for proving Lemma~\ref{lemma:Astar-Delta-k-ones}, the bound it gives is too loose for terms arising in $\calA^* (\Delta^k w)$ (for Lemma~\ref{lemma:Astar-Delta-k-w}) and $1^T \Delta^k w$ (for Lemma~\ref{lemma:1T-Delta-k-w}). Specifically, it is too conservative when assigning weight to the vertices in $\alpha_w$ in the case that its single edge vanishes with respect to an identification pattern $P$ (see Definition \ref{def:vanish}). To handle this bad case, we define a modified weight function $w_{\text{actual},P}$, for any given identification pattern $P$,  by decreasing the weight on the square vertex of $\alpha_w$ when its edge vanishes. While this guarantees that $w_{\text{actual}}(\beta_{k+1})$ is small (making it possible to satisfy \ref{property:actual_bound_small}), previous arguments do not immediately imply that \ref{property:actual_vs_ideal} holds in the case that the edge $\{u,x_1\}_2$ of $\alpha_w$ vanishes. We will show this is compensated for by an increase in the weight on squares in other shapes in a way that ensures \ref{property:actual_vs_ideal} and \ref{property:actual_bound_small} are satisfied simultaneously. To carry out this strategy, we introduce the notion of \textit{critical edges}; see also Figure~\ref{fig:ident-pattern} for an example.
\begin{definition}
\label{def:right_critical_edge}
We define the following two types of edges to be \textbf{right-critical edges}:
\begin{enumerate}
    \item If the square vertex $x_1$ in $\alpha_w$ satisfies $\tau (x_1) = L$ then the edge in $\alpha_w$ is a right-critical edge.
    \item If $\beta_j = \alpha_{2a}$, the circle vertex $u$ in $U_{\beta_j}$ satisfies $\tau(u) = L$, the circle vertex $v$ in $V_{\beta_j}$ satisfies $\tau(v) = R$, and the square vertex $x$ satisfies $\tau(x) = LR$, then the edge $\{u,x\}_2$ in $\beta_j = \alpha_{2a}$ is a right-critical edge.
\end{enumerate}
\end{definition}
\begin{definition}
We define a \textbf{left-critical edge} of $\beta_j$ to be an edge $e = \{u,v\}$ such that one of the following two cases holds:
\begin{enumerate}
    \item $l_e = 2$, $\tau(u) \in \{R, LR\}$, and $\tau(v) \in \{R, LR\}$.
    \item $l_e = 1$, $\tau(u) = \tau(v) = LR$.
\end{enumerate}
\end{definition}

With these definitions in hand, our high-level strategy is as follows. If the right-critical edge in $\beta_{k+1} = \alpha_w$ does not vanish, then the proof strategy of Lemma \ref{lemma:Astar-Delta-k-ones} that employs the local weight scheme $w_{\text{local},P}$ of Section \ref{sec:local_weights} suffices to directly yield the bounds of Lemmas \ref{lemma:Astar-Delta-k-w} and \ref{lemma:1T-Delta-k-w}. If instead the right-critical edge $e$ in $\beta_{k+1} = \alpha_w$ vanishes, we use Lemma \ref{lemma:vanishing-rc-to-lc} to pair $\beta_j$ with a shape $\beta_{j'}$ that contains a left-critical edge. We adjust the weights of $\beta_{j'}$ and $\alpha_w$ directly according to the \textit{actual} weight scheme defined in Section \ref{sec:defn-w-actual} in order to satisfy \ref{property:actual_vs_ideal} and \ref{property:actual_bound_small}. On the remaining shapes $\beta_j$ where $j \in [k]\backslash \{j'\}$, we employ the local weight scheme of Section \ref{sec:local_weights} (i.e.\ the actual weights correspond with the local weights on these remaining shapes). Multiplying together the actual weights for all shapes then yields the bounds of Lemmas \ref{lemma:Astar-Delta-k-w} and \ref{lemma:1T-Delta-k-w}.


\begin{lemma}
\label{lemma:vanishing-rc-to-lc}
Given an identification pattern $P$ on $\beta_{0},\beta_1,\ldots,\beta_{k},\beta_{k+1}$, if there is a $j \in [k+1]$ such that $\beta_j$ has a vanishing right-critical edge $e$ then there is a $j' < j$ such that $\beta_{j'}$ has no vanishing right-critical edge and has a left-critical edge $e'$ whose square endpoint is identified with the square endpoint of $e$.
\end{lemma}
\begin{proof}
We prove this lemma by induction on $j$. Assume the result is true for $j = m$ and assume that $\beta_{j}$ has a vanishing right-critical edge where $j = m+1$.
Recall that we say a vertex $v$ of $\alpha_P$ \textit{appears} in $\beta_\ell$ if it is the result of identifying several vertices according to $P$, one of which lies in $\beta_\ell$. We say that an edge $e$ of $\alpha_P$ \textit{appears} in a shape $\beta_\ell$ if both of its endpoints appear in $\beta_\ell $. Suppose that $\beta_j$ contains a right-critical edge $e$. We claim that $e$ does not appear in $\beta_{j'}$ for $j' > j$. If $\beta_j = \alpha_w$, this claim follows immediately because then $j = k+1$. Now suppose that $e$ is the second type of right-critical edge in Definition \ref{def:right_critical_edge}, in which case we have $\beta_j = \alpha_{2a}$. Since $\tau(u) = L$, it also holds in this case that $e$ does not appear in $j'> j$. 

Let $j'$ denote the largest index such that $j' < j$ and $e$ appears in $\beta_{j'}$ (such an edge must exist as otherwise $e$ cannot vanish). We claim that $e$ must be a left-critical edge in $\beta_{j'}$. To see this, observe that $e$ appears in $\beta_{j}$ where $j > j'$. If $l_e = 2$, then $e$ is automatically a left-critical edge. If $l_e = 1$, then $e$ must also appear in $\beta_{j''}$ for some $j'' < j'$ as otherwise $e$ cannot vanish (two parallel edges with labels $1$ and $2$ give rise to a term with label $1$ and a term with label $3$, so they do not vanish). Thus, $e$ is a left-critical edge in this case as well.

If $\beta_{j'}$ does not have a vanishing right-critical edge, then we are done. If $\beta_{j'}$ does have a vanishing right critical edge (in which case it must be $\alpha_{2a}$) then by the inductive hypothesis there is a $j'' < j'$ which has a left-critical edge but does not have a vanishing right-critical edge, as needed.



\end{proof}

\subsection{Formal definition of $w_{\text{actual}}$}
\label{sec:defn-w-actual}
We now give a formal definition of $w_{\text{actual}}$. Let $u_w, x_{\text{extra}}$ denote the vertices corresponding to $u, x$, respectively, in $\alpha_w$. If $\beta_{k+1} = \alpha_w$, then define the \textit{per-vertex} \textit{actual weights} for $\beta_{k+1}$ as follows:
\begin{enumerate}
\item If the edge $\{u_w,x_{\text{extra}}\}_2$ in $\alpha_w$ does not vanish, then set $w_{\text{actual}}(u_w) = \sqrt[4]{n}$ and $w_{\text{actual}}(x_{\text{extra}}) = \sqrt{d}$. Furthermore, set $w_{\text{actual}}$ equal to $w_{\text{local}}$ for all other vertices and shapes.
\item If the edge $\{u_w,x_{\text{extra}}\}_2$ in $\alpha_w$  vanishes, then set $w_{\text{actual},P}(u_w) = \sqrt{n}$ and $w_{\text{actual},P}(x_{\text{extra}}) = \frac{\sqrt{d}}{\sqrt[4]{n}}$. For the remaining shapes, define $w_{\text{actual}}$ as below.
\end{enumerate}
In the second case above, we modify $w_{\text{local}}$ further to define $w_{\text{actual}}$. Let $j < k+1$ be such that $\beta_j$ has a left-critical edge and no vanishing right-critical edge (whose existence is guaranteed by Lemma~\ref{lemma:vanishing-rc-to-lc}). Note that $\beta_j$ must be one of $\alpha_{A^*,3}, \alpha_{1}, \alpha_{2a}, \alpha_{3a}, \alpha_{3b}$ or $\alpha_4$ by definition of left- and right-critical edges. We set $w_{\text{actual}}$ to be equal to $w_{\text{local}}$ on all shapes $\beta_l$ for $l \neq k+1, j$. To compensate for the reduction in weight on the shape $\beta_{k+1}$ in the case that its edge vanishes, we define $w_{\text{actual}}(\beta_j)$ in the following way. 

For $\ell \in \{j, k+1\}$ we set  
\begin{align*}
   \num  \label{eqn:w_actual_def}
    w_{\text{actual},P}(\beta_\ell)
    = \prod_{v \in \beta_\ell}  w_{\text{actual},P}(v)
\end{align*}
where $w_{\text{actual},P}(v)$ are \textit{per-vertex} actual weights. If $\ell = k+1$, the per-vertex actual weights are defined in 1 and 2 above. Note that this ensures $w_{\text{actual},P}(\beta_{k+1}) \leq \sqrt{d} \sqrt[4]{n}$ and that $w_{\text{actual},P}(u_w) \geq  w_{\text{ideal},P}(u_w)$. If $\ell = j$, the per-vertex actual weights are defined below according to the cases of $\beta_j$. 

\begin{itemize}
    \item \textbf{Case $\beta_j = \beta_0 = \alpha_{A^{*},3}$:} Define $w_{\text{actual},P}(u) = \sqrt[4]{\frac{n}{d}} \cdot w_{\text{ideal},P}(u) = \sqrt[4]{n}$. Here, $w_{\text{ideal},P}(u) = \sqrt[4]{d}$ follows from the fact that $u \in U_{\beta_0} = V_{\beta_0}$, so it is not a middle vertex and thus cannot be in $\text{Iso}(\alpha_P)$.

    \item \textbf{Case $\beta_j = \alpha_1$:} By symmetry, it suffices to consider the case where the edge $\{u_w,x_{\text{extra}}\}_2$ of $\alpha_w$ is identified with the edge $\{u, x_1\}$ of $\alpha_1$.  We now define $w_{\text{actual},P}(u) = 1$, $w_{\text{actual},P}(x_1) = \sqrt[4]{n}$, $w_{\text{actual},P}(x_2) = \sqrt{d}$, $w_{\text{actual},P}(v) = \sqrt{n}$.
    
    \item \textbf{Case $\beta_j = \alpha_{2a}$:} We divide the definition for this case into two sub-cases, based on whether or not the edge $\{u,x\}_2$ in $\alpha_{2a}$ vanishes. 
    \begin{itemize}
        \item \textbf{Sub-case $\{u,x\}_2$ does not vanish:} We define $w_{\text{actual},P}(u) = n^{1/4}$, $w_{\text{actual},P}(x) = n^{1/4}$, $w_{\text{actual},P}(v) = n^{1/2}$.
        
        \item \textbf{Sub-case $\{u,x\}_2$ vanishes:} We define $w_{\text{actual},P}(u) = 1$, $w_{\text{actual},P}(x) = d$, $w_{\text{actual},P}(v) = \sqrt{n}$.
    \end{itemize}

    \item \textbf{Case $\beta_j = \alpha_{3a}$:} We define $w_{\text{actual},P}(x_1) = \sqrt{d}\sqrt[4]{n}$ and $w_{\text{actual},P}(x_2) = \sqrt{d}\sqrt[4]{n}$. 
    
    \item \textbf{Case $\beta_j = \alpha_{3b}$ or $\alpha_{4}$:} We define $w_{\text{actual},P}(x) = \sqrt{d}\sqrt[4]{n}$.
\end{itemize}

See Table~\ref{table:weights-example} for an example of how $w_{\text{actual}}$ is computed for a particular shape and identification pattern arising in $\calA^*(\Delta w)$. This example also demonstrates a shape and identification pattern for which $w_{\text{local}}$ is too conservative and overestimates $w_{\text{ideal}}$ (which corresponds to the ``correct'' norm bound), yet $w_{\text{actual}}$ corrects this issue.

\begin{table}[h!]
\centering
\begin{tabular}{ |c|c|c|c|c|c|c|c|c| } 
\hline
& \multicolumn{8}{|c|}{Shapes} \\
\hline
& \multicolumn{3}{|c|}{$\alpha_{A^*,1}$} & \multicolumn{3}{|c|}{$\alpha_{2a}$} & \multicolumn{2}{|c|}{$\alpha_w$} \\
\hline
  & $u$ & $v$ & $x$ & $u$ & $x$ & $v$ & $u$ & $x$ \\
\hline
Ideal & $\sqrt{d}$ & $\sqrt{d}$ & $\sqrt[4]{n}$ & $\sqrt[4]{n}$ & $\sqrt[4]{d}$ & $\sqrt{n}$ & $\sqrt{n}$ & $\sqrt[4]{d}$  \\ 
Local & $\sqrt{d}$ & $\sqrt{d}$ & $\sqrt[4]{n}$ & $\sqrt[4]{n}$ & $\sqrt[4]{d}$ & $\sqrt{n}$ & $\sqrt{n}$ & {\color{red} $\sqrt{d}$}  \\ 
Actual & $\sqrt{d}$ & $\sqrt{d}$ & $\sqrt[4]{n}$ & $\sqrt[4]{n}$ & {\color{Green} $\sqrt[4]{n}$} & $\sqrt{n}$ & $\sqrt{n}$ & {\color{Green} $\sqrt{d}/ \sqrt[4]{n}$}  \\  \hline
\end{tabular}
\caption{Comparison of the three different weighting schemes applied to the shape and identification pattern from Figures~\ref{fig:ident-pattern} and \ref{fig:ident-pattern-result}. Each column indicates the weight contributions of a vertex to the total weight of the shape that contains it, for each of the three weighting schemes. So, the first row corresponds to weights under $w_{\text{ideal}}$, but which are ``split'' if a vertex is identified with other vertices. The second row corresponds to the values $\overline{b}_j (\cdot)$ from \ref{eqn:bbar_def}. The third row is the same as the second, but adjusted according to the definition of $w_{\text{actual}}$ in Section~\ref{sec:defn-w-actual}. The red entry indicates that $w_{\text{local}}$ assigns more weight than the ``true'' weight as in $w_{\text{ideal}}$; this leads to an over-estimate of the true norm bound by a $\sqrt[4]{d}$ factor. The green entries indicate the weights that are modified so that $w_{\text{actual}}$ gives the correct norm bound.}
\label{table:weights-example}
\end{table}

\subsection{Paying for the extra square: completing the proofs of Lemmas~\ref{lemma:Astar-Delta-k-w} and \ref{lemma:1T-Delta-k-w}}
To prove Lemmas~\ref{lemma:Astar-Delta-k-w} and \ref{lemma:1T-Delta-k-w}, we follow the proof of Lemma~\ref{lemma:Astar-Delta-k-ones}, but use $w_{\text{actual}}$ in place of $w_{\text{local}}$ and the following crucial lemma: 
\begin{lemma}
\label{lem:pay-for-extra-square}
Suppose that the right-critical edge $e$ in $\alpha_w$ vanishes. Let 
$0\leq j < k+1$ be such that $\beta_j$ has no vanishing right-critical edge and has a left-critical edge $e'$ whose square endpoint is identified with the square endpoint of $e$ (whose existence is guaranteed by Lemma~\ref{lemma:vanishing-rc-to-lc}). Then,
\begin{align} 
\label{eqn:actual_condition}
w_{\text{actual},P}(\beta_j)w_{\text{actual},P}(\beta_{k+1}) \geq w_{\text{ideal},P}(\beta_j)w_{\text{ideal},P}(\beta_{k+1}).
\end{align} 
Moreover, it holds that $c(\beta_j) w_{\text{actual},P}(\beta_j) \leq O(d^{3/2} \sqrt[4]{n})$ if $j > 0$, and $c(\beta_j) w_{\text{actual},P}(\beta_j) \leq O(n^{3/4} / \sqrt{d})$ if $j = 0$.
\end{lemma}
Let $j < k+1$ be the special index as in Lemma~\ref{lem:pay-for-extra-square}. Then, as in the proof of Lemma~\ref{lemma:Astar-Delta-k-ones}, the proof of Lemma~\ref{lemma:Astar-Delta-k-w} will be complete provided we can show the following:
\begin{enumerate}
    \item If $\beta_{0}$ is any of $\alpha_{A^*,1}, \alpha_{A^*,2}, \alpha_{A^*,3}$ and $\beta_{k+1} = \alpha_{w}$, then $w_{\text{actual},P}(\beta_{0}) w_{\text{actual},P}(\beta_{k+1}) \leq \max(\sqrt{d}n^{3/4},n)$.
    
    \item For all $l \in [k] \setminus \{j\}$, $|w_{\text{actual},P}(\beta_l)c(\beta_l)| \leq d \sqrt{n}$.
\end{enumerate}
The second condition above follows in exactly the same way as in the proof of Lemma~\ref{lemma:Astar-Delta-k-ones}, since $w_{\text{actual},P}(\beta_l) = w_{\text{local},P}(\beta_l)$ for $l \in [k] \setminus \{j\}$. To verify the first condition, we enumerate two cases:
\begin{itemize}
    \item \textbf{Case $j \neq 0$:} If $j \neq 0$, then $w_{\text{actual},P}(\beta_0) = w_{\text{local},P}(\beta_0) \leq \sqrt{n}$ (from the proof of Lemma~\ref{lemma:Astar-Delta-k-ones}) and $w_{\text{actual},P}(\beta_{k+1}) \leq \sqrt{d} \sqrt[4]{n}$ by definition. So, we immediately have $w_{\text{actual},P}(\beta_{0}) w_{\text{actual},P}(\beta_{k+1}) \leq \sqrt{d}n^{3/4}$.
    \item \textbf{Case $j = 0$:} If $j = 0$, then $w_{\text{actual},P}(\beta_{0}) w_{\text{actual},P}(\beta_{k+1}) \leq (n^{3/4}/\sqrt{d}) \cdot \sqrt{d} \sqrt[4]{n} = n$, by Lemma~\ref{lem:pay-for-extra-square} and definition of $w_{\text{actual},P}(\beta_{k+1}) = w_{\text{actual},P}(\alpha_{w})$.
\end{itemize}

Given Lemma \ref{lem:pay-for-extra-square}, the proof of Lemma~\ref{lemma:1T-Delta-k-w} follows in a similar manner, but taking $\beta_0 = \alpha_{1_n^T}$ instead and noting that $w_{actual, P}(\beta_0) = w_{actual,P}(\alpha_{1_n^T}) \leq \sqrt{n}$. Also, note that if the edge $\{u_w, x_{\text{extra}}\}_2$ of $\alpha_w$ does not vanish, we use the local weight scheme of Lemma \ref{lemma:Astar-Delta-k-ones} to directly obtain a bound of $(\log n)^{O(k)} (d\sqrt{n})^{k+1}$ for Lemma \ref{lemma:Astar-Delta-k-w} and $(\log n)^{O(k)} \sqrt{d} n^{3/4} (d\sqrt{n})^{k}$ for Lemma \ref{lemma:1T-Delta-k-w}. 

Thus we complete the proofs of Lemmas \ref{lemma:Astar-Delta-k-w} and \ref{lemma:1T-Delta-k-w} by proving Lemma~\ref{lem:pay-for-extra-square} below. 

\begin{proof}[Proof of Lemma~\ref{lem:pay-for-extra-square}]
We enumerate the possible cases for the shape $\beta_j$. Because $\beta_j$ contains a left-critical edge and no vanishing right-critical edge, we have that $\beta_j \in \{\alpha_{A^*,3}, \alpha_1, \alpha_{2a}, \alpha_{3a}, \alpha_{3b}, \alpha_{4}\}$. In each case, we refer to the definition of $w_{\text{actual},P}(\beta_j)$ in Section~\ref{sec:defn-w-actual} to verify that \eqref{eqn:actual_condition} holds. Note that because $w_{\text{actual}}(u_w) \geq w_{\text{ideal}}(u_w)$ by definition of $w_{\text{actual}}$, it suffices to show
\[
w_{\text{actual},P}(\beta_j)w_{\text{actual},P}(x_{\text{extra}}) \geq w_{\text{ideal},P}(\beta_j)w_{\text{ideal},P}(x_{\text{extra}})
\]
in order to conclude that
\[
w_{\text{actual},P}(\beta_j)w_{\text{actual},P}(\beta_{k+1}) \geq w_{\text{ideal},P}(\beta_j)w_{\text{ideal},P}(\beta_{k+1}).
\]

\begin{itemize}
    \item \textbf{Case $\beta_j = \beta_0 = \alpha_{A^{*},3}$:}  $w_{\text{ideal},P}(\beta_{0}) = \frac{1}{\sqrt{d}} \sqrt[4]{d} \sqrt{n} = \sqrt{n} / \sqrt[4]{d}$ and $w_{\text{ideal},P}(x_{\text{extra}}) = \sqrt[4]{d}$ since we know the square in $\beta_0$ is not isolated. On the other hand, $w_{\text{actual},P}(x_{\text{extra}}) = \sqrt{d} \sqrt[4]{n}$ and $w_{\text{actual},P}(\beta_{j}) = n^{3/4} / \sqrt{d}$. We immediately observe that
    \[
    w_{\text{actual},P}(\beta_j)w_{\text{actual},P}(x_{\text{extra}}) \geq \sqrt{n} = w_{\text{ideal},P}(\beta_j)w_{\text{ideal},P}(x_{\text{extra}}).
    \]
    
    \item \textbf{Case $\beta_j = \alpha_1$:} By symmetry, it suffices to consider the case where the edge $\{u_w,x_{\text{extra}}\}_2$ of $\alpha_w$ is identified with the edge $\{u, x_1\}$ of $\alpha_1$. Note that this means $\tau(u) = LR$ , $\tau(x_1) = LR$ and $\tau(x_{\text{extra}}) = L$, so $w_{\text{ideal},P}(u) = w_{\text{ideal},P}(x_1) = 1$, $w_{\text{ideal},P}(x_2) \leq \sqrt{d}$, $w_{\text{ideal},P}(v) \leq \sqrt{n}$ and $w_{\text{ideal},P}(x_{\text{extra}}) \leq \sqrt{d}$. Assembling this information, we see:
    \[
     w_{\text{actual},P}(\beta_j)w_{\text{actual},P}(x_{\text{extra}}) \geq n^{3/4} \sqrt{d} \cdot (\sqrt{d}/ \sqrt[4]{n})  = d \sqrt{n} \geq w_{\text{ideal},P}(\beta_j)w_{\text{ideal},P}(x_{\text{extra}}),
    \]
    and $c(\beta_j) w_{\text{actual},P} (\beta_j) = O(n^{3/4} \sqrt{d})$.
    
    \item \textbf{Case $\beta_j = \alpha_{2a}$:} We divide the argument for this case into two sub-cases, based on whether or not the edge $\{u,x\}_2$ in $\alpha_{2a}$ vanishes. 
    \begin{itemize}
        \item \textbf{Sub-case $\{u,x\}_2$ does not vanish:} Note that $\tau(u) \in \{L, LR\}$, $\tau(x) \in \{R, LR\}$, $\tau(v) \in \{R, LR\}$, $\tau(x_{\text{extra}}) = L$ and $u,x, x_{\text{extra}}$ are not isolated, so $w_{\text{ideal},P}(u) \leq n^{1/4}, w_{\text{ideal},P}(x) \leq d^{1/4}, w_{\text{ideal},P}(v) \leq n^{1/2}$ and $w_{\text{ideal},P}(x_{\text{extra}}) \leq d^{1/4}$. Assembling this information, we see:
        \[
        w_{\text{actual},P}(\beta_j)w_{\text{actual},P}(x_{\text{extra}}) \geq n^{3/4} \sqrt{d} \geq w_{\text{ideal},P}(\beta_j)w_{\text{ideal},P}(x_{\text{extra}})
        \]
        and $c(\beta_j ) w_{\text{actual},P} (\beta_j) = O(n)$.
        
        \item \textbf{Sub-case $\{u,x\}_2$ vanishes:} If $\{u,x\}_2$ vanishes, it cannot be right-critical, so either $\tau(u) = LR$ or $\tau(v) = LR$. By symmetry, it suffices to consider the case that $\tau(u) = LR$. Note that $\tau(x) \in \{R, LR\}, \tau(v) \in \{R,LR\}$, and $\tau(x_{\text{extra}}) = L$, so $w_{\text{ideal},P}(u) = 1, w_{\text{ideal},P}(x) = \sqrt{d}, w_{\text{ideal},P}(v) \leq \sqrt{n}$ and $w_{\text{ideal},P}(x_{\text{extra}}) \leq \sqrt{d}$. Assembling this information, we see:
        \[
        w_{\text{actual},P}(\beta_j)w_{\text{actual},P}(x_{\text{extra}}) \geq d^{3/2} n^{1/4} \geq d \sqrt{n} \geq w_{\text{ideal},P}(\beta_j)w_{\text{ideal},P}(x_{\text{extra}})
        \]
        and $c(\beta_j)w_{\text{actual},P}(\beta_j) = O(d \sqrt{n})$.
    \end{itemize}

    \item \textbf{Case $\beta_j = \alpha_{3a}$:} Note that $w_{\text{ideal},P}(\beta_j) \leq d$.
    We may immediately conclude 
    \[
    w_{\text{actual},P}(\beta_j)w_{\text{actual},P}(x_{\text{extra}}) \geq d \sqrt{n} \cdot (\sqrt{d}/\sqrt[4]{n}) \geq d^{3/2} \geq w_{\text{ideal},P}(\beta_j)w_{\text{ideal},P}(x_{\text{extra}})
    \]
    and $c(\beta_j) w_{\text{actual},P} (\beta_j) = O(d\sqrt{n})$.
    
    \item \textbf{Case $\beta_j = \alpha_{3b}$ or $\alpha_{4}$:} Note that $w_{\text{ideal},P}(\beta_j) \leq \sqrt{d}$. So, we have 
    \[
    w_{\text{actual},P}(\beta_j)w_{\text{actual},P}(x_{\text{extra}}) \geq d \geq w_{\text{ideal},P}(\beta_j)w_{\text{ideal},P}(x_{\text{extra}})
    \]
    and $c(\beta_j) w_{\text{actual},P} (\beta_j) \leq O(d^{3/2} \sqrt[4]{n})$.
\end{itemize}
\end{proof}
\appendix

\section{Connection to an average-case discrepancy problem}
\label{sec:disc}

Recently, Aubin, Perkins, and Zdeborov\'{a}~\cite{AubPerZde19} and Turner, Meka and Rigollet~\cite{turner2020balancing} studied the discrepancy of random matrices. Formally, they showed that if $A$ is an $m \times n$ matrix with i.i.d. standard Gaussian entries and $m = \Theta(n)$, then $\disc (A) = \Theta (\sqrt{n})$ with high probability, where the discrepancy of $A$ is defined to be $\disc(A) = \min_{\sigma \in \{\pm 1\}^n} \norm{A \sigma}_{\infty}$. Since the proof of the lower bound in this result is via a union bound over $\sigma \in \{\pm 1\}^n$, we pose the following question: is there a \emph{computationally efficient} algorithm for certifying a lower bound on $\disc(A)$ for random $A$? By certification algorithm, we mean an algorithm that on input $A$ always outputs a value that lower bounds $\disc(A)$, but for random $A$, the value is close to the true value $\Theta(\sqrt{n})$ with high probability. This question is inspired by a long line of work on certifying unsatisfiability of random constraint satisfaction problems (see e.g.\ \cite{raghavendra2017strongly} and references therein), but also has an application to the detection problem in the negatively-spiked Wishart model defined below.

Consider the problem of distinguishing which of the following two distributions a matrix $A \in \R^{m \times n}$ is generated from:
\begin{itemize}
    \item \textbf{Null:} $A_{ij} \sim \calN(0,1)$, for all $i \in [m], j\in [n]$ independently.
    \item \textbf{Planted:} The rows $A_i$ are independently sampled from $\calN(0, I_n - \frac{1}{n} vv^T)$, where $v \sim \Unif \{\pm 1\}^n$.
\end{itemize} 
As mentioned, under the null model, $\disc (A) = \Theta (\sqrt{n})$ with high probability~\cite{turner2020balancing}. On the other hand, it is straightforward to verify that $\disc (A) = 0$ under the planted model. Hence, any algorithm that can certify non-trivial lower bounds on the discrepancy of a Gaussian matrix $A$ can also solve the above detection problem. Bandeira, Kunisky, and Wein~\cite{bandeira2019computational} show that in the regime $m = \alpha n$, for $\alpha > 0$ a constant, distinguishing the above two distributions is hard for the class of low-degree polynomial distinguishers when $\alpha < 1$ and easy when $\alpha > 1$. While the class of low-degree polynomial algorithms is conjectured to match the performance of all polynomial-time algorithms for a wide variety of average-case problems~\cite{hopkins2018statistical,kunisky2022notes}, the above result does not have any formal implication for the powerful class of SDP-based algorithms.

Define the following SDP relaxation (also known as \emph{vector discrepancy}~\cite{nikolov2013komlos}) of discrepancy:
\begin{align*}
   \SDP (A) := & \min_{X \in \R^{n \times n}} \max_{i \in [m]} A_i^T X A_i\\
   & \text{s.t. } X \succeq 0 \\
   & \diag(X) = 1_n.
\end{align*}
It can be verified that for all $A$, it holds that $\SDP(A) \leq \disc(A)^2$. We now state a formal connection, implicit in the  work of Saunderson et al.~\cite{saunderson2012diagonal}, between the ability of the SDP to certify a non-trivial lower bound on the discrepancy and the ellipsoid fitting problem.

\begin{theorem}[\cite{saunderson2012diagonal}]
\label{thm:sdp-to-ellipsoid-fitting}
Let $A \in \R^{m \times n}$ have i.i.d.\ standard Gaussian entries and $m \leq n$. Then
\[
\Prob (\SDP (A) = 0) = \Prob (v_1, \ldots, v_n \text { have the ellipsoid fitting property}),
\]
where $v_1, \ldots, v_n$ are independent samples from $\calN(0, I_d)$ and $d= n-m$.
\end{theorem}

Combining Theorems~\ref{thm:sdp-to-ellipsoid-fitting} and \ref{thm:main}, we conclude that the SDP fails to solve the detection problem in the negatively-spiked Wishart model when $m < n - \sqrt{n} \polylog(n)$. In particular, the inability of the SDP to distinguish between instances with discrepancy 0 and $\Theta (\sqrt{n})$ matches the \emph{worst-case} hardnes of approximation result due to Charikar, Newman and Nikolov~\cite{charikar2011tight}. Further, if Conjecture~\ref{conj:ellipsoid-fitting} is true, the threshold for success of the SDP is exactly $m = n - 2\sqrt{n}$. These results complement those of Mao and Wein~\cite{mao2021optimal} by confirming the phase transition for the SDP takes place at the same finite-scale corrected value $m = n - \sqrt{n} \polylog(n)$ as for low-degree polynomials.

The proof of Theorem~\ref{thm:sdp-to-ellipsoid-fitting}, which we provide for the sake of completeness, makes use of the following lemma. 
\begin{lemma}[Lemma 2.4 and Proposition 3.1 of \cite{saunderson2012diagonal}]
\label{lem:EF_vs_realizable}
Let $\calU \subseteq \R^{n}$ be a subspace. There exists $X \in \R^{n \times n}$ with $X \succeq 0$ and $\diag(X) = 1_n$ such that $\calU$ is contained in the kernel of $X$ if and only if there is a matrix $V$ whose row span is the orthogonal complement of $\calU$ and whose columns have the ellipsoid fitting property.
\end{lemma}

\begin{proof}[Proof of Theorem~\ref{thm:sdp-to-ellipsoid-fitting}]
We begin the proof with a definition from~\cite{saunderson2012diagonal}: a subspace $\calU$ has the ellipsoid fitting property if there exists a matrix whose row span is $\calU$ and whose columns satisfy the ellipsoid fitting property. By definition, $\SDP (A) = 0$ means that there exists $X \in \R^{n \times n}$ with $X \succeq 0$ and $\diag(X) = 1_n$ satisfying $A_i^T X A_i = 0$ for $i = 1, \ldots, m$. Equivalently, the subspace $\calU = \spn \{ A_1, \ldots, A_m\}$ is contained in the kernel of $X$. Defining $\calU^\perp$ to be the orthogonal complement of $\calU$, Lemma~\ref{lem:EF_vs_realizable} tells us that $\SDP (A) = 0$ is equivalent to $\calU^\perp$ having the ellipsoid fitting property. However, note that $\calU^\perp$ has the same distribution as the span of $v_1, \ldots, v_n$, independent samples from $\calN(0, I_d)$ with $d= n-m$. Altogether, we have:
\begin{align*}
\Prob (\SDP (A) = 0) &= \Prob (\calU^{\perp} \text{ has the ellipsoid fitting property}) \\
&= \Prob (v_1, \ldots, v_n \text{ have the ellipsoid fitting property}). \qedhere
\end{align*}
\end{proof}

\section{Invertibility lemma}
\begin{lemma}
\label{lemma:invertible}
If $n > d(d+1)/2$, then $\calA \calA^*$ is not invertible for any $v_1, \ldots, v_n$. If $n \leq d(d+1)/2$, then $\calA \calA^*$ is invertible with probability 1.
\end{lemma}
\begin{proof}
Consider the vectors $v_1 v_1^T, \ldots, v_n v_n^T$ in the $d(d+1)/2$-dimensional vector space $\mathbb{S}^{d \times d}$ of symmetric $d \times d$ matrices. Since $\calA \calA^*$ is the Gram matrix of the vectors $v_1 v_1^T, \ldots, v_n v_n^T$, it $\calA \calA^*$ is invertible iff $v_1 v_1^T, \ldots, v_n v_n^T$ are linearly independent. If $n > d(d+1)/2$, then clearly there must be a linear dependency since the vector space $\mathbb{S}^{d \times d}$ has dimension $d(d+1)/2$. 

We now show linear independence with probability 1 provided $n \leq d(d+1)/2$. Defining $\Pi: \mathbb{S}^{d \times d} \rightarrow \mathbb{S}^{d \times d}$ to be the projector onto the orthogonal complement of $\mathsf{span}(v_2 v_2^T, \ldots, v_n v_n^T)$, the proof will be complete by showing $\Prob(\Pi (v_1 v_1^T) = 0) = 0$. Observe that since $\Pi$ is a projector, all of its eigenvalues are 0 or 1 and since $n \leq d(d+1)/2$, $\Pi$ has rank at least 1, so it has an eigenvector $M \in \mathbb{S}^{d \times d}$ with eigenvalue 1. As a consequence, we have that $\Prob (\Pi (v_1 v_1^T) = 0) \leq \Prob (\ip{M}{v_1 v_1^T} = 0)$. Since $M \in \mathbb{S}^{d \times d}$, we may write its eigendecomposition as $M = \sum_{i=1}^d \lambda_i w_i w_i^T$, where $\lambda_1, \ldots, \lambda_d \in \R$ and $\{w_1, \ldots, w_d\}$ form an orthonormal basis of $\R^d$. Now, we have that
\[
\ip{M}{v_1 v_1^T} = \sum_{i=1}^d \lambda_i \ip{w_i}{v_1}^2.
\]
We see that $\ip{M}{v_1 v_1^T}$ is a non-zero linear combination (because $M \neq 0$) of $d$ independent and identically distributed random variables $\ip{w_1}{v_1}^2, \ldots, \ip{w_d}{v_1}^2$ with distribution $\chi^2_1$ each. Hence, we conclude that $\Prob (\ip{M}{v_1 v_1^T} = 0) = 0$, completing the proof.
\end{proof}

\section{Details of experiments}
\label{sec:experiment-details}
In this section, we elaborate on how the plots in Figure~\ref{fig:plots} were generated. The plots corresponding to the SDP and the least-squares construction appeared in~\cite{saunderson_parrilo_willsky13}, but for different ranges of $(n,d)$. To generate Figure~\ref{fig:sdp}, we used the the CVXPY package to test feasibility of the original ellipsoid fitting SDP. We implemented two shortcuts to reduce the computation time. First, for Figure~\ref{fig:sdp}, we performed the simulation only for $n \leq d(d+1)/2$ since the linear system will be infeasible with probability 1 for $n > d(d+1)/2$. For all $n > d(d+1)/2$, we filled in the cell corresponding to $(n,d)$ with black without actually performing the simulation, since by Lemma~\ref{lemma:invertible} the ellipsoid fitting property will fail with probability 1 in this regime.  Second, for all of Figures~\ref{fig:sdp}, \ref{fig:ls}, \ref{fig:ip}, for each $n$, we performed the simulations starting from $d=1$ until encountering 5 consecutive values of $d$ for which all 10 of the trials were successful and then filled in all remaining cells corresponding to larger values of $d$ with white. We believe that this shortcut does not affect the final appearance of the plot in any noticeable way.

We remark that for the least-squares construction, the $n = cd^2$ scaling of the phase transition is only apparent for larger values of $n,d$ in Figure~\ref{fig:ls} and that the transition from infeasibility to feasibility is much coarser than in Figures~\ref{fig:sdp} and \ref{fig:ip}. To accentuate these effects, we reproduce the plots in Figure~\ref{fig:plots} on a $\log_2$ scale in Figure~\ref{fig:log-plots} so that the parabola $n = cd^2$ becomes a line with slope 2. 

\begin{figure}
    \centering
    \begin{subfigure}[b]{0.3\textwidth}
        \includegraphics[width=\textwidth]{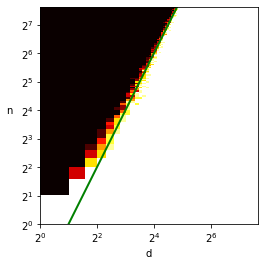}
        \caption{Figure~\ref{fig:sdp} on a $\log_2$ scale}
        \label{fig:sdp-log}
    \end{subfigure}
    ~ 
    \begin{subfigure}[b]{0.3\textwidth}
        \includegraphics[width=\textwidth]{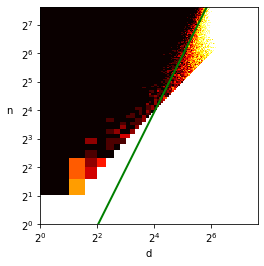}
        \caption{Figure~\ref{fig:ls} on a $\log_2$ scale}
        \label{fig:ls-log}
    \end{subfigure}
    ~ 
    \begin{subfigure}[b]{0.3\textwidth}
        \includegraphics[width=\textwidth]{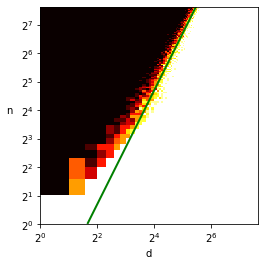}
        \caption{Figure~\ref{fig:ip} on a $\log_2$ scale}
        \label{fig:ip-log}
    \end{subfigure}
    \caption{Figure~\ref{fig:plots} on a $\log_2$ scale}
    \label{fig:log-plots}
\end{figure}

\section{Probabilistic norm bounds}
\label{sec:norm_bounds} 

We restate Theorem \ref{thm:graph-matrix-norm-bound} below for convenience. 

\begin{theorem}[Theorem \ref{thm:graph-matrix-norm-bound}]
	Given $D_V,D_E \in \mathbb{N}$ such that $D_E \geq D_V \geq 2$ and $\epsilon > 0$, with probability at least $1 - \epsilon$, for all shapes $\alpha$ on square and circle vertices such that $|V(\alpha)| \leq D_V$ and $|E_{\alpha}| \leq D_E$, $|U_{\alpha}| \leq 1$, and $|V_{\alpha}| \leq 1$,
 \[
	\|M_{\alpha}\| \leq \left((2D_E + 2)\ln(D_V) + \ln(11n) + \ln\left(\frac{1}{\epsilon}\right)\right)^{|V(\alpha)| + |E(\alpha)|}n^{\frac{\varphi(V(\alpha)) - \varphi(S_{min}) + \varphi(\mathrm{Iso}(\alpha))}{2}}
	\]
	where $S_{min}$ is a minimum vertex separator of $\alpha$.
\end{theorem}
\begin{proof}
Corollary~8.16 of \cite{ahn2016graph} says that for all $\epsilon' > 0$ and all shapes $\alpha$ with square and circle vertices and no isolated vertices outside of $U_{\alpha} \cup V_{\alpha}$, with probability at least $1-\epsilon'$,
	\begin{align*}
		\|M_{\alpha}\| &\leq 2|V_{\circ}(\alpha)|^{|V_{\circ}(\alpha)|}|V_{\Box}(\alpha)|^{|V_{\Box}(\alpha)|}n^{\frac{\varphi(V(\alpha)) - \varphi(S_{min}) + \varphi(\mathrm{Iso}(\alpha))}{2}} \\
		&\qquad \cdot \, \left(6e\left\lceil\frac{\ln(\frac{n^{\varphi(S_{min})}}{\epsilon'})}{6(|V(\alpha) \setminus (U_{\alpha} \cap V_{\alpha})| + |E(\alpha)|)}\right\rceil\right)^{|E(\alpha)| + |V(\alpha) \setminus (U_{\alpha} \cap V_{\alpha})|}.
	\end{align*}
	The result is trivial if $|V(\alpha)| \leq 1$ or $E(\alpha) = \emptyset$ so we can assume that $|V(\alpha)| \geq 1$ and $E(\alpha) \geq  1$. For the shapes $\alpha$ we are considering, $\varphi(S_{min}) \leq 1$ so for each such shape $\alpha$, for all $\epsilon' > 0$, 
	\[
	\|M_{\alpha}\| \leq |V(\alpha)|^{|V(\alpha)| + |E(\alpha)|}n^{\frac{\varphi(V(\alpha)) - \varphi(S_{min}) + \varphi(\mathrm{Iso}(\alpha))}{2}}\left(6e\left\lceil\frac{\ln\left(\frac{n}{\epsilon'}\right)}{6|V(\alpha)|}\right\rceil\right)^{|V(\alpha)| + |E(\alpha)|}.
	\]
	We will now apply this to all such shapes $\alpha$ with $\epsilon' = \frac{\epsilon}{11{D_V}^{(2D_E + 2)}}$ and take a union bound. Since $D_E \geq D_V \geq |V(\alpha)|$, we have that 
	$\ln\left(\frac{n}{\epsilon'}\right) \geq (2D_E + 2)\ln(D_V) \geq 2D_V \geq 2|V(\alpha)|$, so 
	\[	\left\lceil\frac{\ln\left(\frac{n}{\epsilon'}\right)}{6|V(\alpha)|}\right\rceil \leq \frac{\ln\left(\frac{n}{\epsilon'}\right)}{2|V(\alpha)|}.
	\]
	Thus, for each such shape $\alpha$, with probability at least $1 - \frac{\epsilon}{11{D_V}^{(2D_E + 2)}}$,
	\[
	\|M_{\alpha}\| \leq \left((2D_E + 2)\ln(D_V) + \ln(11n) + \ln\left(\frac{1}{\epsilon}\right)\right)^{|V(\alpha)| + |E(\alpha)|}n^{\frac{\varphi(V(\alpha)) - \varphi(S_{min}) + \varphi(\mathrm{Iso}(\alpha))}{2}}.
	\]
Using the following proposition and taking a union bound, we have that with probability at least $1 - \epsilon$, the above bound holds for all such shapes $\alpha$, as needed.
\end{proof}

\begin{proposition}
\label{prop:number-of-shapes-bound}
If $D_V,D_E \in \mathbb{N}$ and $D_V \geq 2$ then there are at most $4{D_V}^{(2D_E + 2)}$ shapes $\alpha$ on square and circle vertices such that $|V(\alpha)| \leq D_V$, $|E_{\alpha}| \leq D_E$, $|U_{\alpha}| \leq 1$, and $|V_{\alpha}| \leq 1$.
\end{proposition}
\begin{proof}
We can specify a non-empty shape $\alpha$ as follows:
\begin{enumerate}
    \item Specify whether $U_{\alpha}$ and $V_{\alpha}$ have a circle vertex, a square vertex, or are empty. If $U_{\alpha}$ and $V_{\alpha}$ both have circle vertices or both have square vertices, specify whether they are the same vertex. There are a total of $11$ choices for this.
    \item Specify the number of circle vertices and square vertices in $V(\alpha) \setminus (U_{\alpha} \cup V_{\alpha})$. There are at most $D_V^2$ choices for this.
    \item For each of the $D_E$ possible edges, either specify its two endpoints or $\emptyset$ if it does not exist. There are at most $\binom{D_V}{2} + 1 \leq D_V^2$ choices for each possible edge.\qedhere
\end{enumerate}
\end{proof}

\section{Proof of Lemma \ref{lemma:woodbury}} 
\label{appendix:woodbury}
For convenience we restate the lemma below. 
\begin{lemma}[Lemma \ref{lemma:woodbury}]
	Let $B = \Gamma + \alpha I_n$.  We have
	\begin{align}
	(\calA \calA^*)^{-1} 1_n = \frac{1}{s^2 - r u} \cdot 
		\bigg( (1 + 1_n^T B^{-1} w)  B^{-1} 1_n  - (1_n^T B^{-1} 1_n) B^{-1} w   \bigg)
	\end{align}
where $r, s, u$ are defined as
\begin{align*}
	 \begin{pmatrix}
		r & s \\
		s & u
	\end{pmatrix} := \begin{pmatrix}
		1_n^T B^{-1} 1_n & 1 + 1_n^T B^{-1} w \\
		1 + 	1_n^T B^{-1} w & -d + w^T B^{-1} w
	\end{pmatrix}. 
\end{align*}
\end{lemma}

\begin{proof}
	
The Woodbury formula \cite{woodbury1950inverting} states that
\begin{align}
	(B + UCV)^{-1} 1_n
	= B^{-1} 1_n - B^{-1} U ( C^{-1} + V B^{-1} U)^{-1}  V B^{-1} 1_n.
\end{align}

We set $B$ as above. Let $U \in \mathbb{R}^{ n \times 2 }$ be defined by
\begin{align*}
	U_{ij} = \begin{cases}
	\, \,1 \quad &\text{ if } j = 1, \text{ and} \\
	\, \, w_{i} = \| v_i \|_2^2 - d &\text{ if } j = 2.	
	\end{cases} 
\end{align*}
So the columns of $U$ are $1_n$ and $w$. Let $V = U^T$, and set
\begin{align*}
	C = \begin{pmatrix}
		d & 1 \\
		1 & 0 
	\end{pmatrix}.
\end{align*}
Observe that $UCV = W$. Next,
\begin{align*}
	C^{-1} + VB^{-1} U
	&=  \begin{pmatrix}
		0 & 1 \\
		1 & -d 
	\end{pmatrix}
+ \begin{pmatrix}
	1_n^T B^{-1} 1_n & 1_n^T B^{-1} w \\
	1_n^T B^{-1} w & w^T B^{-1} w
\end{pmatrix} \\
&= \begin{pmatrix}
	1_n^T B^{-1} 1_n & 1 + 1_n^T B^{-1} w \\
1 + 	1_n^T B^{-1} w & -d + w^T B^{-1} w
\end{pmatrix}
=: \begin{pmatrix}
	r & s \\
	s & u
\end{pmatrix}.
\end{align*}
Thus
\begin{align*}
(	C^{-1} + VB^{-1} U )^{-1}
= \frac{1}{ru - s^2} \begin{pmatrix}
	u & -s \\ 
	-s & r 
\end{pmatrix},
\end{align*}
and
\begin{align*}
V B^{-1} 1_n
	=  \begin{pmatrix}
		1_n^T B^{-1} 1_n \\
		w^T B^{-1} 1_n
	\end{pmatrix} = \begin{pmatrix}
	r \\
	s - 1
\end{pmatrix}.
\end{align*}
Hence
\begin{align*}
	(	C^{-1} + VB^{-1} U )^{-1} V B^{-1} 1_n
	&= \frac{1}{ru - s^2} \cdot \begin{pmatrix}
	ru - s(s-1) \\
	-sr + r(s - 1)
	\end{pmatrix}
= \frac{1}{ru - s^2} \cdot  \begin{pmatrix}
	ru - s^2 + s \\
	-r
\end{pmatrix}. 
\end{align*}
Next, since $U$ has first column $1_n$ and second column $w$, 
\begin{align*}
	(AA^*)^{-1} 1_n &= B^{-1} 1_n - B^{-1} U ( C^{-1} + V B^{-1} U)^{-1}  V B^{-1} 1_n \\
	&= B^{-1} 1_n -  \frac{1}{ru - s^2} B^{-1} U  \cdot  \begin{pmatrix}
		ru - s^2 + s \\
		-r
	\end{pmatrix}
\\&= \big(1 - \frac{ru - s^2 + s}{ru - s^2}  \big) B^{-1} 1_n
+ \frac{r}{ru - s^2}  B^{-1} w
\\&= \frac{1}{ru - s^2} \cdot 
\big( -s B^{-1} 1_n  + r B^{-1} w \big)
\\&= \frac{1}{s^2 - r u} \cdot 
\bigg( (1 + 1_n^T B^{-1} w)  B^{-1} 1_n  - (1_n^T B^{-1} 1_n) B^{-1} w   \bigg). \qedhere
\end{align*}
\end{proof}

\section{Hermite polynomials}
\label{sec:hermite-calculations}
Here, we provide some technical results regarding Hermite polynomials that are useful when applying Proposition~\ref{prop:resolve-multi-edges}. Throughout, we use the convention $\mathbb{N} = \{1,2,3\ldots\}$ (with $0$ not included). 

\begin{lemma}\label{lem:multiplyHermitebyx}
For all $j \in \mathbb{N}$,
\[
h_1(x)h_j(x) = xh_{j}(x) = \sqrt{j+1}h_{j+1}(x) + \sqrt{j}h_{j-1}(x).
\]
\end{lemma}
\begin{proof}
Since the normalized Hermite polynomials $\{h_j: j \in \mathbb{N} \cup \{0\}\}$ are orthonormal with respect to the inner product 
\[
\langle{h_i,h_j}\rangle := \E_{x \sim N(0,1)}[h_i(x)h_j(x)],
\]
we have that 
\[
xh_j(x) = \sum_{k = 0}^{\infty}{\E_{y \sim N(0,1)}[yh_j(y)h_k(y)]h_k(x)}.
\]
We now make the following observations:
\begin{enumerate}
    \item If $k < j-1$ then $\E_{y \sim N(0,1)}[yh_j(y)h_k(y)] = 0$ because $yh_k(y)$ is a degree $k+1$ polynomial and $h_j(y)$ is orthogonal to all polynomials of degree less than $j$.
    \item If $k > j+1$ then $\E_{y \sim N(0,1)}[yh_j(y)h_k(y)] = 0$ because $yh_j(y)$ is a degree $j+1$ polynomial and $h_k(y)$ is orthogonal to all polynomials of degree less than $k$.
    \item If $k = j$ then 
    $\E_{y \sim N(0,1)}[yh_j(y)h_k(y)] = 0$ because $yh_j(y)h_k(y)$ is an odd polynomial.
    \item If $k = j-1$ then the leading term of $yh_k(y) = \frac{x^j}{\sqrt{(j-1)!}}$ so we can write $yh_k(y) = \sqrt{j}h_j(y) + p$ where $p$ has degree at most $j-1$. This implies that 
    $\E_{y \sim N(0,1)}[yh_j(y)h_k(y)] = \E_{y \sim N(0,1)}[\sqrt{j}(h_j(y))^2] = \sqrt{j}$.
    \item If $k = j+1$ then the leading term of $yh_j(y) = \frac{x^k}{\sqrt{j!}}$ so we can write $yh_j(y) = \sqrt{j+1}h_{k}(y) + p$ where $p$ has degree at most $k-1$. This implies that 
    $\E_{y \sim N(0,1)}[yh_j(y)h_k(y)] = \E_{y \sim N(0,1)}[\sqrt{j+1}(h_{k}(y))^2] = \sqrt{j+1}$.\qedhere
\end{enumerate}
\end{proof}


\begin{corollary}\label{cor:xsquaredHermite}
For all $j \in \mathbb{N}$ with $j \geq 2$,
\[
{x^2}h_{j}(x) = \sqrt{(j+1)(j+2)}h_{j+2}(x) + (2j+1)h_{j}(x) + \sqrt{j(j-1)}h_{j-2}(x).
\]
\end{corollary}
\begin{proof}
By Lemma \ref{lem:multiplyHermitebyx},
\begin{align*}
{x^2}h_{j}(x) &= x\sqrt{j+1}h_{j+1}(x) + x\sqrt{j}h_{j-1}(x) \\
&=\sqrt{j+1}(\sqrt{j+2}h_{j+2}(x) + \sqrt{j+1}h_{j}(x)) + \sqrt{j}(\sqrt{j}h_{j}(x) + \sqrt{j-1}h_{j-2}(x)) \\
&= \sqrt{(j+1)(j+2)}h_{j+2}(x) + (2j+1)h_{j}(x) + \sqrt{j(j-1)}h_{j-2}(x).\qedhere
\end{align*}
\end{proof}
\begin{corollary}
For all $j \in \mathbb{N}$,
\[
h_2(x)h_{j}(x) = \frac{x^2 - 1}{\sqrt{2}}h_j(x) =  \sqrt{\frac{(j+1)(j+2)}{2}}h_{j+2}(x) + \sqrt{2}{j}h_{j}(x) + \sqrt{\frac{j(j-1)}{2}}h_{j-2}(x).
\]
\end{corollary}
\begin{remark}
Note that for the case $j = 1$, $\sqrt{j(j-1)} = 0$. Thus, for $j = 1$, even though $h_{j-2}(x) = h_{-1}(x)$ is undefined it does not matter as $\sqrt{j(j-1)}h_{j-2}(x) = 0$ regardless of what $h_{-1}(x)$ is.
\end{remark}

\section{Analysis of the Pseudo-Calibration Construction}
\label{sec:pseudocalibration}
In this section, we describe and analyze the construction of $M$ from \cite{ghosh2020sum}. This construction is obtained by using pseudo-calibration (for background on pseudo-calibration, see \cite{barak2019nearly}) on the following distributions.
\begin{enumerate}
    \item[] Random: Sample $n$ vectors $v_1,\ldots,v_n$ from $\calN(0, I_d)$.
    \item[] Planted: First sample a hidden direction $u$ from $\{-\frac{1}{\sqrt{d}},\frac{1}{\sqrt{d}}\}^d$ and $n$ random $\pm{1}$ variables $b_1,\ldots,b_n$. Then sample $n$ vectors $v'_1,\ldots,v'_n$ from $\calN(0, I_d)$ and replace each vector $v'_i$ with $v_i = v'_i - \langle{v'_i,u}\rangle{u} + {b_i}u$.
\end{enumerate}
For the planted distribution, the rank one matrix $M = u{u^T}$ satisfies ${v_i^T}M{v_i} = 1$ for all $i \in [n]$. For the random distribution, there is no hidden direction $u$ but with high probability, pseudo-calibration will still give us a matrix $M$ such that ${v_i^T}M{v_i} = 1$ for all $i \in [n]$. In order to describe this matrix $M$, we need a few definitions.
\begin{definition}
    Given values $\{\alpha_{i,a}: i \in [n], a \in [d]\}$, we make the following definitions:
    \begin{enumerate}
    \item Define $|\alpha|$  to be $|\alpha| = \sum_{i = 1}^{n}{\sum_{a=1}^{d}{\alpha_{i,a}}}$
    \item Define $\alpha_i$ to be $\alpha_i = \sum_{a=1}^{d}{\alpha_{i,a}}$
    \item Define $\alpha^T_{a}$ to be $\alpha^T_{a} = \sum_{i=1}^{n}{\alpha_{i,a}}$ 
    \item Define $\alpha!$ to be $\alpha! = \prod_{i = 1}^{n}{\prod_{a=1}^{d}{\alpha_{i,a}!}}$
    \item Define $h_{\alpha}(v_1,\ldots,v_n)$ to be $h_{\alpha}(v_1,\ldots,v_n) = \prod_{i = 1}^{n}{\prod_{a=1}^{d}{h_{\alpha_{i,a}}((v_i)_a)}}$
    \end{enumerate}
\end{definition}
Let $T = \Omega(\log n)$ be a truncation parameter. By Lemma 4.4 of \cite{ghosh2020sum}, the construction given by pseudo-calibration with truncation parameter $T$ is as follows.
\begin{definition}
    Define $\tilde{E}[1]$ to be 
\[
\tilde{E}[1] = \sum_{\alpha: |\alpha| \leq T, \text{ For all } i \in [n], a \in [d] , \alpha_{i} \text{ and } \alpha^T_{a}\text{ are even}}{\frac{\left(\prod_{i=1}^{n}{\sqrt{{\alpha_i}!}h_{\alpha_i}(1)}\right)}{\sqrt{\alpha!}d^{\frac{|\alpha|}{2}}}h_{\alpha}(v_1,\ldots,v_n)}  
\]
\end{definition}
\begin{definition}
    For all $a \in [d]$, we define $\tilde{E}[x_a^2]$ to be $\tilde{E}[x_a^2] = \frac{1}{d}\tilde{E}[1]$. For all distinct $a,b \in [d]$, define $\tilde{E}[{x_a}{x_b}]$ to be 
\[
\tilde{E}[{x_a}{x_b}] = \sum_{\alpha: |\alpha| \leq T, \text{ For all } i \in [n], c \in [d]\setminus \{a,b\} , \alpha_{i} \text{ and } \alpha^T_{c}\text{ are even}, \atop 
\alpha^T_{a} \text{ and } \alpha^T_{b} \text{ are odd.}}{\frac{\left(\prod_{i=1}^{n}{\sqrt{{\alpha_i}!}h_{\alpha_i}(1)}\right)}{\sqrt{\alpha!}d^{\frac{|\alpha|}{2}+1}}h_{\alpha}(v_1,\ldots,v_n)}  
\]
\end{definition}
\begin{remark}
    These equations have different coefficients than Lemma 4.4 of \cite{ghosh2020sum} because we are using the normalized Hermite polynomials. 
\end{remark}
\begin{definition}
    For all distinct $a,b \in [d]$, we take $M_{ab} = \frac{\tilde{E}[x_{a}x_{b}]}{\tilde{E}[1]}$. For all $a \in [d]$, we take $M_{aa} = \frac{\tilde{E}[x_a^2]}{\tilde{E}[1]} = \frac{1}{d}$.
\end{definition}
\subsection{Verifying $M$ is PSD}
While the pseudo-calibration construction is more complicated than the least squares and identity perturbation constructions, it is actually easier to give a rough analysis for it. The reason is that $\tilde{E}[1]M$ can be directly decomposed into shapes. Moreover, all of the shapes $\alpha$ appearing in $\tilde{E}[1]M$ have the following properties

\begin{enumerate}
\item $M_{\alpha}$ appears in $\tilde{E}[1]M$ with coefficient $\lambda_{\alpha} = O(d^{-(\frac{|E(\alpha)|}{2}+1)})$ where we take $|E(\alpha)|$ to be the sum of the labels of the edges in $E(\alpha)$.
\item Let $U_{\alpha} = (u)$ and $V_{\alpha} = (v)$, every square vertex has even degree and has degree at least $2$. If $u \neq v$ then $u$ and $v$ have odd degree. If $u = v$ then $u$ has even degree (which may be $0$).

Note that we take the degree of a vertex to be the sum of the labels of the edges incident to that vertex.
\item Every circle vertex has even degree and has degree at least $4$.
\end{enumerate}
Using the same logic that we used to prove Lemma \ref{lem:min-vertex-sep}, each such shape $\alpha$ contains a path from $U_{\alpha}$ to $V_{\alpha}$ so the minimum weight vertex separator of $\alpha$ is a single square. By Theorem \ref{thm:graph-matrix-norm-bound}, with high probability, $||M_{\alpha}||$ is $\tilde{O}(n^{\frac{|\calV_{\circ}(\alpha)|}{2}}d^{\frac{|\calV_{\Box}(\alpha)| - 1}{2}})$. We now make the following observations:
\begin{enumerate}
\item Since every square vertex except $u,v$ has degree at least $2$, $|E(\alpha)| = \sum_{w \in \calV_{\Box}(\alpha)}{deg(w)} \geq 2|\calV_{\Box}(\alpha)| - 2$ which implies that $|\calV_{\Box}(\alpha)| \leq \frac{|E(\alpha)|}{2} + 1$
\item Since every circle vertex has degree at least $4$, $|E(\alpha)| = \sum_{w \in \calV_{\circ}(\alpha)}{deg(w)} \geq 4|\calV_{\circ}(\alpha)|$ which implies that $|\calV_{\circ}(\alpha)| \leq \frac{|E(\alpha)|}{4}$
\end{enumerate}
Putting these observations together, with high probability, $\lambda_{\alpha}||M_{\alpha}||$ is 
\[
\tilde{O}(n^{\frac{|\calV_{\circ}(\alpha)|}{2}}d^{\frac{|\calV_{\Box}(\alpha)| - 1 - |E(\alpha)|}{2} - 1})
\leq \tilde{O}\left(\frac{1}{d}\left(\frac{\sqrt[8]{n}}{\sqrt[4]{d}}\right)^{|E(\alpha)|}\right)\]
This implies that the dominant term is $\frac{1}{d}{I_d}$ as it is the only 
 term that appears which has no edges. Thus, with high probability, $\tilde{E}[1]M$ is PSD. As noted in Remark 5.9 of \cite{ghosh2020sum}, with high probability, $\tilde{E}[1]$ is $1 \pm o(1)$ so this implies that with high probability, $M$ is PSD, as needed.
\begin{remark}
    This analysis is very similar to the analysis on p.21 of \cite{ghosh2020sum} for attempt 1 where each edge splits its factor of $\frac{1}{\sqrt{d}}$ between its two endpoints. While this attempt fails for the higher degree setting of \cite{ghosh2020sum}, it succeeds for degree $2$, which is what we are analyzing here.
\end{remark}
\subsection{Verifying that ${v_i^T}Mv_i \approx 1$}
As discussed in Section 7 of \cite{ghosh2020sum}, pseudo-calibration guarantees that the constraints ${v_i^T}Mv_i = 1$ are satisfied up to a very small truncation error which can be easily repaired. However, looking at the entries of $M$ directly, it is not at all easy to see why ${v_i^T}Mv_i \approx 1$. In this subsection, we show how to directly verify that ${v_i^T}Mv_i \approx 1$. In particular, we give a direct proof that for each $\alpha$ such that $|\alpha| \leq T-2$ (where $T = \Omega(\log n) $ is the truncation parameter), the coefficient of $h_{\alpha}$ in $\sum_{a=1}^{d}{\sum_{b=1}^{d}{\tilde{E}[{x_a}{x_b}](v_i)_a(v_i)_b}}$ matches the coefficient of $h_{\alpha}$ in $\tilde{E}[1]$. This analysis is a special case of the analysis on p.42-45 of \cite{ghosh2020sum}.

There are several ways that $h_{\alpha}$ can appear in $\sum_{a=1}^{d}{\sum_{b=1}^{d}{\tilde{E}[{x_a}{x_b}](v_i)_a(v_i)_b}}$. These ways are as follows:
\begin{enumerate}
    \item For some $a \in [d]$ and $b \in [d] \setminus \{a\}$, $h_{\alpha_{i,a}-1}((v_i)_a)$ is multiplied by $(v_i)_a$ and $h_{\alpha_{i,b}-1}((v_i)_b)$ is multiplied by $(v_i)_b$, giving $\sqrt{\alpha_{i,a}\alpha_{i,b}}h_{\alpha_{i,a}}((v_i)_a)h_{\alpha_{i,b}}((v_i)_b)$.

    Letting $\alpha'$ be $\alpha$ where $\alpha_{i,a}$ and $\alpha_{i,b}$ are decreased by $1$, the coefficient of $h_{\alpha'}$ in $\tilde{E}[x_{a}x_{b}]$ is 
    \[
        \frac{\left(\prod_{i=1}^{n}{\sqrt{{\alpha_i}!}h_{\alpha_i}(1)}\right)}{\sqrt{\alpha!}d^{\frac{|\alpha|}{2}}} \cdot \frac{\sqrt{\alpha_{i,a}\alpha_{i,b}}h_{\alpha_{i}-2}(1)}{\sqrt{\alpha_i(\alpha_i-1)}h_{\alpha_i}(1)}
    \]
    so the total contribution from these terms is 
    \[
    \frac{\left(\prod_{i=1}^{n}{\sqrt{{\alpha_i}!}h_{\alpha_i}(1)}\right)}{\sqrt{\alpha!}d^{\frac{|\alpha|}{2}}}\sum_{a=1}^{d}{\sum_{b \in [d] \setminus \{a\}}{\frac{\alpha_{i,a}\alpha_{i,b}h_{\alpha_{i}-2}(1)}{\sqrt{\alpha_i(\alpha_i-1)}h_{\alpha_i}(1)}}}
    \]
    \item For some $a \in [d]$, $h_{\alpha_{i,a}-2}((v_i)_a)$ is multiplied by $(v_i)_a^2$, giving $\sqrt{\alpha_{i,a}(\alpha_{i,a}-1)}h_{\alpha_{i,a}}((v_i)_a)$.

    Letting $\alpha'$ be $\alpha$ where $\alpha_{i,a}$ is decreased by $2$, the coefficient of $h_{\alpha'}$ in $\tilde{E}[x_{a}^2]$ is 
    \[
        \frac{\left(\prod_{i=1}^{n}{\sqrt{{\alpha_i}!}h_{\alpha_i}(1)}\right)}{\sqrt{\alpha!}d^{\frac{|\alpha|}{2}}} \cdot \frac{\sqrt{\alpha_{i,a}(\alpha_{i,a}-1)}h_{\alpha_{i}-2}(1)}{\sqrt{\alpha_i(\alpha_i-1)}h_{\alpha_i}(1)}
    \]
    so the total contribution from these terms is 
    \[
    \frac{\left(\prod_{i=1}^{n}{\sqrt{{\alpha_i}!}h_{\alpha_i}(1)}\right)}{\sqrt{\alpha!}d^{\frac{|\alpha|}{2}}}\sum_{a=1}^{d}{\frac{\alpha_{i,a}(\alpha_{i,a}-1)h_{\alpha_{i}-2}(1)}{\sqrt{\alpha_i(\alpha_i-1)}h_{\alpha_i}(1)}}
    \]
    Together, the terms in cases 1 and 2 give a total contribution of 
    \[
        \frac{\left(\prod_{i=1}^{n}{\sqrt{{\alpha_i}!}h_{\alpha_i}(1)}\right)} {\sqrt{\alpha!}d^{\frac{|\alpha|}{2}}} \cdot \frac{\sqrt{\alpha_{i}(\alpha_{i}-1)}h_{\alpha_{i}-2}(1)}{h_{\alpha_i}(1)}
    \]
    
    \item For some $a \in [d]$ and $b \in [d] \setminus \{a\}$, $h_{\alpha_{i,a}-1}((v_i)_a)$ is multiplied by $(v_i)_a$ and $h_{\alpha_{i,b}+1}((v_i)_b)$ is multiplied by $(v_i)_b$, giving $\sqrt{\alpha_{i,a}(\alpha_{i,b}+1)}h_{\alpha_{i,a}}((v_i)_a)h_{\alpha_{i,b}}((v_i)_b)$.

    Letting $\alpha'$ be $\alpha$ where $\alpha_{i,a}$ is decreased by $1$ and $\alpha_{i,b}$ is increased by $1$, the coefficient of $h_{\alpha'}$ in $\tilde{E}[x_{a}x_{b}]$ is 
    \[
        \frac{\left(\prod_{i=1}^{n}{\sqrt{{\alpha_i}!}h_{\alpha_i}(1)}\right)}{\sqrt{\alpha!}d^{\frac{|\alpha|}{2}}} \cdot \frac{\sqrt{\alpha_{i,a}}}{d\sqrt{(\alpha_{i,b}+1)}}
    \]
    so the total contribution from these terms is 
    $\frac{(d-1)\alpha_i}{d}\frac{\left(\prod_{i=1}^{n}{\sqrt{{\alpha_i}!}h_{\alpha_i}(1)}\right)}{\sqrt{\alpha!}d^{\frac{|\alpha|}{2}}}
    $
    By symmetry, we have the same contribution from the terms where $h_{\alpha_{i,a}+1}((v_i)_a)$ is multiplied by $(v_i)_a$ and $h_{\alpha_{i,b}-1}((v_i)_b)$ is multiplied by $(v_i)_b$ so the total contribution from all of these terms is $\frac{2(d-1)\alpha_i}{d}\frac{\left(\prod_{i=1}^{n}{\sqrt{{\alpha_i}!}h_{\alpha_i}(1)}\right)}{\sqrt{\alpha!}d^{\frac{|\alpha|}{2}}}
    $
    \item For some $a \in [d]$, $h_{\alpha_{i,a}}((v_i)_a)$ is multiplied by $(v_i)_a^2$, giving $(2\alpha_{i,a}+1)h_{\alpha_{i,a}}((v_i)_a)$.

    The coefficient of $h_{\alpha}$ in $\tilde{E}[x_{a}^2]$ is 
    \[
        \frac{1}{d}\frac{\left(\prod_{i=1}^{n}{\sqrt{{\alpha_i}!}h_{\alpha_i}(1)}\right)}{\sqrt{\alpha!}d^{\frac{|\alpha|}{2}}}
    \]
    so the total contribution from these terms is 
    \[
        \left(\frac{2\alpha_{i}}{d} + 1\right)\frac{\left(\prod_{i=1}^{n}{\sqrt{{\alpha_i}!}h_{\alpha_i}(1)}\right)}{\sqrt{\alpha!}d^{\frac{|\alpha|}{2}}}
    \]
    Together, the terms in cases 3 and 4 give a total contribution of 
    \[
        (2\alpha_i + 1)\frac{\left(\prod_{i=1}^{n}{\sqrt{{\alpha_i}!}h_{\alpha_i}(1)}\right)}{\sqrt{\alpha!}d^{\frac{|\alpha|}{2}}}
    \]
    \item For some $a \in [d]$ and $b \in [d] \setminus \{a\}$, $h_{\alpha_{i,a}+1}((v_i)_a)$ is multiplied by $(v_i)_a$ and $h_{\alpha_{i,b}+1}((v_i)_b)$ is multiplied by $(v_i)_b$, giving $\sqrt{(\alpha_{i,a}+1)(\alpha_{i,b}+1)}h_{\alpha_{i,a}}((v_i)_a)h_{\alpha_{i,b}}((v_i)_b)$.

    Letting $\alpha'$ be $\alpha$ where $\alpha_{i,a}$ and $\alpha_{i,b}$ are increased by $1$, the coefficient of $h_{\alpha'}$ in $\tilde{E}[x_{a}x_{b}]$ is 
    \[
        \frac{\left(\prod_{i=1}^{n}{\sqrt{{\alpha_i}!}h_{\alpha_i}(1)}\right)}{\sqrt{\alpha!}d^{\frac{|\alpha|}{2}}} \cdot \frac{\sqrt{(\alpha_i+1)(\alpha_i + 2)}h_{\alpha_{i}+2}(1)}{d^2\sqrt{(\alpha_{i,a}+1)(\alpha_{i,b}+1)}h_{\alpha_i}(1)}
    \]
    so the total contribution from these terms is 
    \[
    \frac{\left(\prod_{i=1}^{n}{\sqrt{{\alpha_i}!}h_{\alpha_i}(1)}\right)}{\sqrt{\alpha!}d^{\frac{|\alpha|}{2}}} \cdot \frac{d(d-1)\sqrt{(\alpha_i+1)(\alpha_i + 2)}h_{\alpha_{i}+2}(1)}{{d^2}h_{\alpha_i}(1)}
    \]
    \item For some $a \in [d]$, $h_{\alpha_{i,a}+2}((v_i)_a)$ is multiplied by $(v_i)_a^2$, giving $\sqrt{(\alpha_{i,a}+2)(\alpha_{i,a}+1)}h_{\alpha_{i,a}}((v_i)_a^2)$.

    Letting $\alpha'$ be $\alpha$ where $\alpha_{i,a}$ is increased by $2$, the coefficient of $h_{\alpha'}$ in $\tilde{E}[x_{a}^2]$ is 
    \[
        \frac{\left(\prod_{i=1}^{n}{\sqrt{{\alpha_i}!}h_{\alpha_i}(1)}\right)}{\sqrt{\alpha!}d^{\frac{|\alpha|}{2}}} \cdot \frac{\sqrt{(\alpha_i+1)(\alpha_i + 2)}h_{\alpha_{i}+2}(1)}{d^2\sqrt{(\alpha_{i,a}+2)(\alpha_{i,a}+1)}h_{\alpha_i}(1)}
    \]
    so the total contribution from these terms is 
    \[
    \frac{\left(\prod_{i=1}^{n}{\sqrt{{\alpha_i}!}h_{\alpha_i}(1)}\right)}{\sqrt{\alpha!}d^{\frac{|\alpha|}{2}}} \cdot \frac{d\sqrt{(\alpha_i+1)(\alpha_i + 2)}h_{\alpha_{i}+2}(1)}{{d^2}h_{\alpha_i}(1)}
    \]
    Together, the terms in cases $5$ and $6$ give a total contribution of 
    \[
    \frac{\left(\prod_{i=1}^{n}{\sqrt{{\alpha_i}!}h_{\alpha_i}(1)}\right)}{\sqrt{\alpha!}d^{\frac{|\alpha|}{2}}} \cdot \frac{\sqrt{(\alpha_i+1)(\alpha_i + 2)}h_{\alpha_{i}+2}(1)}{h_{\alpha_i}(1)}
    \]
\end{enumerate}
    Putting everything together, it is sufficient to show that 
    \[
        \sqrt{\alpha_{i}(\alpha_{i}-1)}h_{\alpha_{i}-2}(1) + (2\alpha_{i} + 1)h_{\alpha_{i}}(1) + \sqrt{(\alpha_i+1)(\alpha_i + 2)}h_{\alpha_{i}+2}(1) = h_{\alpha_{i}}(1)
    \]
    To show this, recall that by Corollary \ref{cor:xsquaredHermite}, for all $j \in \mathbb{N} \cup \{0\}$ and all $x \in \mathbb{R}$,
    \[
        {x^2}h_{j}(x) = \sqrt{(j+1)(j+2)}h_{j+2}(x) + (2j+1)h_{j}(x) + \sqrt{j(j-1)}h_{j-2}(x).
    \]
    Plugging in $x = 1$ and $j = \alpha_{i}$, the result follows.

\section{Analysis of the Identity Perturbation Construction} 
\label{appendix:identity_perturbation}

In this section, we provide an analysis of the identity perturbation construction and show that it is PSD provided that $ n \leq d^2/\text{polylog}(d)$. Again without loss of generality, we assume that $n \geq d$. Recall that
\[
X := X_{\mathrm{IP}} = \frac{1}{d} I_d + \calA^*(c)
\]
where $c$ is chosen such that $\calA(X) = 1_n$. By direct calculation and the invertibility of $(\calA \calA^*)^{-1}$, there is a unique vector $c$ satisfying the constraint $\calA(X) = 1_n$, and it is given by $c = -\frac{1}{d} (\calA \calA^*)^{-1}w$, where recall $w_i = \|v_i\|^2 - d$. 
Again by the Woodbury formula \cite{woodbury1950inverting}, it holds that 
\begin{align}
	(\calA \calA^*)^{-1}w = (B + UCV)^{-1} w
	= B^{-1}w - B^{-1} U ( C^{-1} + V B^{-1} U)^{-1}  V B^{-1} w
\end{align}
where $B, U, C,$ and $V$ are defined in Section \ref{appendix:woodbury}.
Using a similar calculation as in Section \ref{appendix:woodbury} and recalling also the definitions of $r, s,$ and $u$ given there, we obtain

\begin{align*}
(\calA \calA^*)^{-1}w &= B^{-1}w - B^{-1} U ( C^{-1} + V B^{-1} U)^{-1}  V B^{-1} w
\\&= B^{-1}w - \frac{1}{ru - s^2} B^{-1} U 
\begin{pmatrix}
    u & -s \\
    -s & r
\end{pmatrix} \begin{pmatrix}
    s - 1 \\
    u + d 
\end{pmatrix} 
\\&= (\frac{ u+ sd }{ru - s^2}) B^{-1} 1_n
- (\frac{s + rd}{ru - s^2}) B^{-1} w
\end{align*}
Hence, 
\begin{align}
\label{eqn:XIP_decomp}
    X =   \frac{1}{d} I_d + \frac{1}{d}(\frac{ u+ d }{s^2 - ru}) \calA^* B^{-1} 1_n
+  \frac{1}{d}(\frac{d(s-1) }{s^2 - ru}) \calA^* B^{-1} 1_n
- \frac{1}{d}(\frac{s + rd}{s^2 - ru}) \calA^* B^{-1} w.
\end{align}
By \eqref{eqn:wwT_trace}, it holds that with high probability
\[
u + d = w^T B^{-1} w = \Theta(\frac{1}{d^2}) \| w \|_2^2, 
\]
which implies in particular that $u + d = w^T B^{-1} w > 0$. Moreover, $s^2 - ru = \Omega(n/d)$ with high probability by \eqref{eqn:s2-ru}, so also $s^2 - ru \geq 0$. It follow from Lemma \ref{lemma:Astar-Binv-1} that the second term of \eqref{eqn:XIP_decomp} is PSD with high probability. 

Next we show that the third term of \eqref{eqn:XIP_decomp} satisfies
\begin{align}
\label{eqn:XIP_decomp3}
    \| \frac{1}{d} (\frac{d(s-1)}{s^2 - ru}) \calA^* B^{-1} 1_n \|_{op} = o(1/d)
\end{align}
with high probability.
In the proof of Lemma \ref{lemma:1T-Binv-w}, we in fact showed that $|s-1| = |1_n^T B^{-1} w| = \tilde o(n/d^2)$. Moreover, the proof of Lemma \ref{lemma:Astar-Binv-1} implies also that $\| \calA^* B^{-1} 1_n \|_{op} = O(n/d^2)$. Thus \eqref{eqn:XIP_decomp3} follows from $s^2 - ru = \Omega(n/d)$ (see \eqref{eqn:s2-ru}) assuming that $n \leq d^2/\text{polylog}(d)$. 

Moreover, by Lemmas \ref{lemma:1T-Binv-1}, \ref{lemma:1T-Binv-w}, and \ref{lemma:Astar-Binv-w} as well as \eqref{eqn:s2-ru}, we obtain that the last term of \eqref{eqn:XIP_decomp} is $o(1/d)$ in operator norm assuming that $n \geq d$ and $n \leq d^2/\text{polylog}(d)$.

Combining the results for the last three terms of \eqref{eqn:XIP_decomp}, we conclude that $X = X_{\mathrm{IP}} \succeq 0$ with high probability, as desired. \qed 



\section{Notes on a previous approach}
\label{sec:prev-approach}
In this section, we discuss the mistake appearing in a previous version of this paper, sketch how this mistake can be repaired, and explain why we instead use the Woodbury matrix identity in the current paper.

The approach used in the previous version of this paper was as follows. Taking $M = (d^2 + d)I + d{1_n}1_n^T$ and $\Delta = M - AA^{*}$,\footnote{Note that this $\Delta$ is different from the one used in the rest of the current paper.} we have that
\[
X_{\text{LS}} = \calA^{*}((\calA \calA^{*})^{-1}1_n) = \calA^{*}((I_n - M^{-1}\Delta)^{-1}M^{-1}1_n).
\]
Expanding $(I_n - M^{-1}\Delta)^{-1}$ as a Neumann series:
\[
(I_n - M^{-1}\Delta)^{-1} = I_n + M^{-1}\Delta + \sum_{j=2}^{\infty}{(M^{-1}\Delta)^{j}}
\]
and observing that $M^{-1}1_n = \frac{1}{d^2 + d + dn}1_n$, we have that 
\[
(d^2 + d + dn)X_{\text{LS}} = \calA^{*}(1_n) + \calA^{*}((M^{-1}\Delta)1_n) + \calA^{*} \left(\left(\sum_{j=2}^{\infty}{(M^{-1}\Delta)^{j}}\right)1_n \right).
\]
It is a standard fact that when $n= \Omega(d)$, with high probability $\lambda_{min}(\calA^{*}(1_n)) = \lambda_{min}(\sum_{i=1}^n v_i v_i^T)$ is $\Omega(n)$. In order to show that $X_{\text{LS}} \succeq 0$ with high probability, it is sufficient to show that with high probability:
\begin{enumerate}
    \item $\|M^{-1}\Delta\|_{op} < 1$,
    \item $\|\calA^{*}(M^{-1}\Delta 1_n)\|_{op} = o(n)$,
    \item For all $j \geq 2$, $\|\calA^{*} ((M^{-1}\Delta)^{j}1_n)\|_{op} = o(n)$.
\end{enumerate}
Proposition 5.2 of the previous version of this paper claimed that with high probability, $\|\Delta{1_n}\|_{\infty} = \tilde{O}(d\sqrt{n})$ which implies that $\|M^{-1}\Delta{1_n}\|_{\infty} = \tilde{O}(\frac{d\sqrt{n}}{d^2}) = o(1)$. In turn, this implies that $\|M^{-1}\Delta{1_n}\|_{2} = o(\sqrt{n})$. By Lemma 3 of \cite{saunderson2011subspace}, with high probability $\|A^{*}\|_{2 \to op} = \Theta(d + \sqrt{n})$ so this implies that for all $j \geq 1$, $\|A^{*}(M^{-1}\Delta)^{j}1_n\| = o(n)$.

Unfortunately, this proposition is  incorrect. As we discuss below, the correct bound on $\|\Delta{1_n}\|_{\infty}$ is $\tilde{O}(n\sqrt{d})$. This gives $\|M^{-1}\Delta{1_n}\|_{\infty} = \tilde{O} \left(\frac{n\sqrt{d}}{d^2} \right) = \tilde{O} \left(n / d^{3/2} \right)$. This is sufficient when $n \ll d^{3/2}$ but is not sufficient when $n \gg d^{3/2}$. This means that in order to prove our result using this approach, we cannot consider $\calA^{*}$ and $(M^{-1}\Delta)^{j}1_n$ separately. Instead, we must analyze their product $\calA^{*}((M^{-1}\Delta)^{j}1_n)$.

\begin{remark}
If we showed the stronger statement $\|(M^{-1}\Delta)^{j}1_n\|_{\infty} = o(1)$ for all $j$ then we would have that every coordinate of 
$(I_n - M^{-1}\Delta)^{-1}1_n = 1_n + M^{-1}{\Delta}1_n + \sum_{j=2}^{\infty}{(M^{-1}\Delta)^{j}1_n}$ is positive. This implies that $X_{\text{LS}} \succeq 0$. We found experimentally that this is true when $n \ll d^{3/2}$. However, when $n \gg d^{3/2}$, $(I_n - M^{-1}\Delta)^{-1}{1_n}$ has negative coordinates, so the interaction between $\calA^{*}$ and $(I_n - M^{-1}\Delta)^{-1}{1_n}$ is crucial.
\end{remark}
\subsection{Computing $\Delta$ and $M^{-1}\Delta$}
In order to discuss why Proposition 5.2 of the previous version of this paper is incorrect and how to actually carry out this approach, it is helpful to express $\Delta$ and $M^{-1}\Delta$ in terms of graph matrices. Recall that:
\begin{align*}
{\calA}{\calA}^{*} &= M_{\alpha_1} + 2M_{\alpha_{2a}} + \sqrt{2}M_{\alpha_{2b}} + \sqrt{2}M_{\alpha_{2c}} + dM_{\alpha_{2d}} + 2M_{\alpha_{3a}} + 2\sqrt{2}(d-1)M_{\alpha_{3b}} + (d^2 - d)M_{\alpha_{3c}} \\
&\qquad+ \sqrt{24}M_{\alpha_{4}} + 6\sqrt{2}M_{\alpha_{3b}} + 3dM_{\alpha_{3c}}.
\end{align*}
where $\alpha_1$, $\alpha_{2a}$, $\alpha_{2b}$, $\alpha_{2c}$, $\alpha_{2d}$, $\alpha_{3a}$, $\alpha_{3b}$, $\alpha_{3c}$, and $\alpha_{4}$ are the following proper shapes:
\begin{enumerate}
    \item $\alpha_1$ is the same as in Section~\ref{sec:shapes}.
    \item $U_{\alpha_{2a}} = (u)$ and  $V_{\alpha_{2a}} = (v)$ where $u, v$ are circle vertices, $W_{\alpha_{2a}} = \{w\}$ where $w$ is a square vertex, and $E(\alpha_{2a}) = \{\{u,w\}_2, \{w,v\}_2\}$.
    \item $U_{\alpha_{2b}} = (u)$ and  $V_{\alpha_{2b}} = (v)$ where $u, v$ are circle vertices, $W_{\alpha_{2b}} = \{w\}$ where $w$ is a square vertex, and $E(\alpha_{2b}) = \{\{u,w\}_2\}$.
    \item $U_{\alpha_{2c}} = (u)$ and  $V_{\alpha_{2c}} = (v)$ where $u, v$ are circle vertices, $W_{\alpha_{2c}} = \{w\}$ where $w$ is a square vertex, and $E(\alpha_{2c}) = \{\{w,v\}_2\}$.
    \item $U_{\alpha_{2d}} = (u)$ and  $V_{\alpha_{2d}} = (v)$ where $u, v$ are circle vertices, $W_{\alpha_{2d}} = \{\}$, and $E(\alpha_{2d}) = \{\}$.
    \item $U_{\alpha_{3a}} = V_{\alpha_{3a}} = (u)$ where $u$ is a circle vertex, $W_{\alpha_{3a}} = \{w_1,w_2\}$ where $w_1,w_2$ are square vertices, and $E(\alpha_{3a}) = \{\{u,w_1\}_2, \{u,w_2\}_2\}$.
    \item $U_{\alpha_{3b}} = V_{\alpha_{3a}} = (u)$ where $u$ is a circle vertex, $W_{\alpha_{3b}} = \{w\}$ where $w$ is a square vertex, and $E(\alpha_{3b}) = \{\{u,w\}_2\}$.
    \item $U_{\alpha_{3c}} = V_{\alpha_{3c}} = (u)$ where $u$ is a circle vertex, $W_{\alpha_{3c}} = \{\}$, and $E(\alpha_{3c}) = \{\}$.
    \item $U_{\alpha_{4}} = V_{\alpha_{4}} = (u)$ where $u$ is a circle vertex, $W_{\alpha_{4}} = \{w\}$ where $w$ is a square vertex, and $E(\alpha_{4}) = \{\{u,w\}_4\}$.
\end{enumerate}

Since $M_{\alpha_{2d}} = 1_{n}1_{n}^T - I_n$ and $M_{\alpha_{3c}} = I_n$, we have that:
\[
{\calA}{\calA}^{*} = (d^2 + d)Id_{n} + d{1_n}{1_n^T} + M_{\alpha_1} + 2M_{\alpha_{2a}} + \sqrt{2}M_{\alpha_{2b}} + \sqrt{2}M_{\alpha_{2c}} + 2M_{\alpha_{3a}} + (2\sqrt{2}d + 4\sqrt{2})M_{\alpha_{3b}} + \sqrt{24}M_{\alpha_{4}}.
\]
Since $M = (d^2 + d)I_{n} + d{1_n}1_n^T$, 
\[
\Delta = M - {\calA}{\calA}^{*} = -M_{\alpha_1} - 2M_{\alpha_{2a}} - \sqrt{2}M_{\alpha_{2b}} - \sqrt{2}M_{\alpha_{2c}} - 2M_{\alpha_{3a}} - (2\sqrt{2}d + 4\sqrt{2})M_{\alpha_{3b}} - \sqrt{24}M_{\alpha_{4}}.
\]
We now compute $M^{-1}{\Delta}$. Since $M = (d^2 + d)I + d{1_n}1_n^T$, we have that $M^{-1} = \frac{1}{d^2 + d}\left(I - \frac{1}{n + d + 1}{1_n}1_n^T\right)$.
\begin{definition}
Define $\alpha_J$ to be the shape with $U_{\alpha_J} = (u)$, $V_{\alpha_J} = (v)$, $W_{\alpha_J} = \emptyset$, and $E(\alpha_J) = \emptyset$.
\end{definition}
Since ${1_n}1_n^T = I + M_{\alpha_J}$, we have
\[
M^{-1} = \frac{1}{d^2 + d}\left(\frac{n+d}{n+d+1}I - \frac{1}{n + d + 1}M_{\alpha_J}\right).
\]
We now compute the product of $M_{\alpha_J}$ with each of the graph matrices appearing in $\Delta$.
\begin{enumerate}
    \item $M_{\alpha_J}M_{\alpha_1} = M_{\alpha_5} + M_{\alpha_1}$ where $M_{\alpha_5}$ is the shape such that $U_{\alpha_{5}} = (u)$ and  $V_{\alpha_{5}} = (v)$ where $u, v$ are circle vertices, $W_{\alpha_{5}} = \{w_{\circ},w_1,w_2\}$ where $w_{\circ}$ is a circle vertex and $w_1,w_2$ are square vertices, and $E(\alpha_{1}) = \{\{w_{\circ},w_1\}, \{w_{\circ},w_2\}, \{w_1,v\},\{w_2,v\}\}$.
    \item $M_{\alpha_J}M_{\alpha_{2a}} = M_{\alpha_{6a}} + M_{\alpha_{2a}}$ where $M_{\alpha_{6a}}$ is the shape such that $U_{\alpha_{6a}} = (u)$ and $V_{\alpha_{6a}} = (v)$ where $u, v$ are circle vertices, $W_{\alpha_{6a}} = \{w_{\circ},w\}$ where $w_{\circ}$ is a circle vertex and $w$ is a square vertex, and $E(\alpha_{6a}) = \{\{w_{\circ},w\}_2, \{w,v\}_2\}$.
    \item $M_{\alpha_J}M_{\alpha_{2b}} = M_{\alpha_{6b}} + M_{\alpha_{2c}}$ where $M_{\alpha_{6b}}$ is the shape such that $U_{\alpha_{6b}} = (u)$ and $V_{\alpha_{6b}} = (v)$ where $u, v$ are circle vertices, $W_{\alpha_{6b}} = \{w_{\circ},w\}$ where $w_{\circ}$ is a circle vertex and $w$ is a square vertex, and $E(\alpha_{6b}) = \{\{w_{\circ},w\}_2\}$.
    \item $M_{\alpha_J}M_{\alpha_{2c}} = (n-2)M_{\alpha_{2c}} + M_{\alpha_{2b}}$.
    \item $M_{\alpha_J}M_{\alpha_{3a}} = M_{\alpha_{7a}}$
    where $M_{\alpha_{7a}}$ is the shape such that $U_{\alpha_{7a}} = (u)$ and $V_{\alpha_{7a}} = (v)$ where $u, v$ are circle vertices, $W_{\alpha_{7a}} = \{w_1,w_2\}$ where $w_1,w_2$ are square vertices, and $E(\alpha_{7a}) = \{\{v,w_1\}_2, \{v,w_2\}_2\}$.
    \item $M_{\alpha_J}M_{\alpha_{3b}} = M_{\alpha_{7b}}$
    where $M_{\alpha_{7b}}$ is the shape such that $U_{\alpha_{7b}} = (u)$ and $V_{\alpha_{7b}} = (v)$ where $u, v$ are circle vertices, $W_{\alpha_{7b}} = \{w\}$ where $w$ is a square vertex, and $E(\alpha_{7b}) = \{\{v,w\}_2\}$.
    \item $M_{\alpha_J}M_{\alpha_{4}} = M_{\alpha_{8}}$
    where $M_{\alpha_{8}}$ is the shape such that $U_{\alpha_{8}} = (u)$ and $V_{\alpha_{8}} = (v)$ where $u, v$ are circle vertices, $W_{\alpha_{8}} = \{w\}$ where $w$ is a square vertex, and $E(\alpha_{8}) = \{\{v,w\}_4\}$.
\end{enumerate}

Putting everything together, we have that 
\begin{align*}
    M^{-1}\Delta = &\frac{1}{(d^2 + d)(n + d + 1)}\Big(-(n+d)M_{\alpha_1} + \left(M_{\alpha_5} + M_{\alpha_1}\right) -2(n+d)M_{\alpha_{2a}} + 2\left(M_{\alpha_{6a}} + M_{\alpha_{2a}}\right)\\
    &-\sqrt{2}(n+d)M_{\alpha_{2b}} + \sqrt{2}\left(M_{\alpha_{6b}} + M_{\alpha_{2c}}\right)
    -\sqrt{2}(n+d)M_{\alpha_{2c}} + \sqrt{2}\left((n-2)M_{\alpha_{2c}} + M_{\alpha_{2b}}\right)\\
    &-2(n+d)M_{\alpha_{3a}} + 2M_{\alpha_{7a}} -(2\sqrt{2}d + 4\sqrt{2})(n+d)M_{\alpha_{3b}} + (2\sqrt{2}d + 4\sqrt{2})M_{\alpha_{7b}} \\
    &-\sqrt{24}(n+d)M_{\alpha_{4}} + \sqrt{24}M_{\alpha_{8}}\Big) \\
    =&\frac{1}{(d^2 + d)(n + d + 1)}\Big(-(n+d-1)M_{\alpha_1} -2(n+d-1)M_{\alpha_{2a}}
    -\sqrt{2}(n+d-1)M_{\alpha_{2b}}\\
    &-\sqrt{2}(d-1)M_{\alpha_{2c}} -2(n+d)M_{\alpha_{3a}} -(2\sqrt{2}d + 4\sqrt{2})(n+d)M_{\alpha_{3b}} -\sqrt{24}(n+d)M_{\alpha_{4}}\\
    &+ M_{\alpha_5} + 2M_{\alpha_{6a}}  + \sqrt{2}M_{\alpha_{6b}} + 2M_{\alpha_{7a}}  + (2\sqrt{2}d + 4\sqrt{2})M_{\alpha_{7b}} + \sqrt{24}M_{\alpha_{8}}\Big).
\end{align*}

Now that we have obtained a graph matrix decomposition of $M^{-1}\Delta$, we demonstrate why Proposition~5.2 of the previous version of this paper is incorrect. Note that $M^{-1}\Delta$ contains the term \[ -\frac{1}{(d^2 + d)(n + d + 1)}\sqrt{2}(n+d - 1)M_{\alpha_{2b}}. \]
We now make the following observations:
\begin{enumerate}
\item $M_{\alpha_{2b}}1_{n} = (n-1)M_{\alpha_{w}}$.
\item $M_{\alpha_{w}}$ is an $n \times 1$ vector, $\|M_{\alpha_{w}}\|_{2} = \tilde{O}(\sqrt{dn})$, and $\|M_{\alpha_{w}}\|_{\infty} = \tilde{O}(\sqrt{d})$.
\end{enumerate}
Putting these observations together, the contribution to $\|M^{-1}{\Delta}1_n\|_{\infty}$ from the $M_{\alpha_{2b}}$ term of $M^{-1}{\Delta}$ is $\tilde{O}(\frac{n\sqrt{d}}{d^2}) = \tilde{O}(\frac{n}{d^{3/2}})$. It can be checked that the contribution to $\|M^{-1}{\Delta}1_n\|_{\infty}$ from the other terms of $M^{-1}{\Delta}$ is smaller. Thus, the correct bound is $\norm{M^{-1} \Delta 1_n}_\infty = O(n/d^{3/2})$.

\subsection{Proof sketch for repairing the argument}
While this proposition is not correct, the three statements needed for this approach to succeed are correct. For convenience, we recall these statements here.
\begin{enumerate}
    \item $\|M^{-1}\Delta\|_{op} < 1$,
    \item $\|\calA^{*}(M^{-1}\Delta 1_n)\|_{op} = o(n)$,
    \item For all $j \geq 2$, $\|\calA^{*} ((M^{-1}\Delta)^{j}1_n)\|_{op} = o(n)$.
\end{enumerate}
To show these statements, we can follow the proof of Lemma~\ref{lem:min-vertex-sep} to show that in all of the terms which appear in these expressions, the minimum vertex separator consists of one square vertex. We can then use the similar weighting schemes to bound these terms. Two notable cases for $w_{actual}$ are as follows.
\begin{enumerate}
\item For $\alpha_{2b}$ where the square appears again to the left, but not to the right, and the edge vanishes, we assign $\sqrt{n}$ to each vertex and $1$ to the square (which is an under-assignment). This means that each time $\alpha_{2b}$ appears, we may have a debt for this square. Fortunately, we can use the same ideas as before. In particular, we can pay off this debt by making the edge in $\alpha_{2b}$ a right-critical edge if it vanishes and finding the corresponding left-critical edge.
\item For $\alpha_{2c}$, because of $M^{-1}$ we have a coefficient of $O(\frac{1}{nd})$ rather than $O(\frac{1}{d^2})$. This allows us to assign weight $\sqrt{n}$ to the two circle vertices and $d$ to the square vertex which is sufficient to handle any debt. Note that this is one of the cases which can have a left-critical edge.
\end{enumerate}

While this approach can be made to work, we use the Woodbury matrix identity approach in the current paper for two reasons. First, it gives a better approximation to $\calA \calA^*$. Second, it requires less casework and only requires paying off debt for one square, instead of several such squares. 

\section*{Acknowledgements}
We thank Sinho Chewi for helpful discussions during the early stages of this project and Tselil Schramm for helpful conversations and making us aware of the work of Amelunxen et al.~\cite{amelunxen2014living}. We also thank Yue Lu, Subhabrata Sen and Nati Srebro for helpful discussions. Prayaag Venkat and Alex Wein thank the Simons Institue for hosting them for the Fall 2021 program on Computational Complexity of Statistical Inference, during which part of this work was done.

\bibliographystyle{alpha}
\bibliography{bib}
\end{document}